\documentclass[onecolumn,draftcls]{IEEEtran}

\usepackage{epsfig}
\usepackage{times}
\usepackage{float}
\usepackage{afterpage}
\usepackage{amsmath}
\usepackage{amstext}
\usepackage{amssymb,bm, bbm}
\usepackage{latexsym}
\usepackage{color}
\usepackage{graphicx}
\usepackage{amsmath}
\usepackage{amsthm}
\usepackage{graphicx}
\usepackage{subfigure}
\usepackage{booktabs}
\usepackage{multicol}
\usepackage{lipsum}
\usepackage{dblfloatfix}
\usepackage{mathrsfs}
\usepackage{cite}
\usepackage{tikz}
\usepackage{pgfplots}
\pgfplotsset{compat=newest}
 \usepackage{enumitem}
\usepackage{makecell}
\usepackage{float}
\restylefloat{table}

\allowdisplaybreaks

\newcommand{\expect}[1]{{\mathbb{E}}\left[#1\right]}
\newcommand{\expcnd}[2]{{\mathbb{E}}\left[ #1 \;\middle|\; #2\right]}
\newcommand{\expcndn}[3]{{\mathbb{E}^{#1}}\left[ #2 \;\middle|\; #3\right]}

\newcommand{\lr}[1]{\left( #1 \right)}

\newcommand{\rmd}{\mathrm{d}}
\newcommand{\bbE}{\mathbb{E}}\newcommand{\rme}{\mathrm{e}}

\newcommand{\bbN}{\mathbb{N}}

\newcommand{\bbP}{\mathbb{P}}

\newcommand{\bbR}{\mathbb{R}}

\newcommand{\sfA}{\mathsf{A}}
\newcommand{\sfB}{\mathsf{B}}

\newcommand{\bfe}{\mathbf{e}}

\newcommand{\sfN}{\mathsf{N}}

\newcommand{\bfp}{\mathbf{p}}

\newcommand{\cA}{\mathcal{A}}
\newcommand{\cB}{\mathcal{B}}

\newcommand{\cD}{\mathcal{D}}

\newcommand{\cP}{\mathcal{P}}

\newcommand{\cX}{\mathcal{X}}

\newcommand{\del}{\partial}

\newcommand{\D}{D}
\newcommand{\kl}[2]{{\D}\left(\left.#1 \, \right\| #2 \right)}

\newcommand{\supp}{{\mathsf{supp}}}

\theoremstyle{mystyle}
\newtheorem{theorem}{Theorem}
\theoremstyle{mystyle}
\newtheorem{lemma}{Lemma}
\theoremstyle{mystyle}
\newtheorem{prop}{Proposition}
\theoremstyle{mystyle}
\newtheorem{corollary}{Corollary}
\theoremstyle{mystyle}
\theoremstyle{remark}
\theoremstyle{mystyle}
\theoremstyle{mystyle}
\theoremstyle{mystyle}
\theoremstyle{discussion}
\theoremstyle{mystyle}
\theoremstyle{mystyle}

\usepackage{enumitem}

\begin{document}

\title{The Binomial Channel: On Capacity, Optimal Inputs, and Beta-Binomial Approximation   }

\author{
\IEEEauthorblockN{ Antonino Favano,  Mohammadamin Baniasadi,  Ian Zieder, Luca Barletta, and Alex Dytso} 
\thanks{Part of this work was presented at the 2024 IEEE International Symposium on Information Theory~\cite{Zieder2024}. This is an extension of arXiv:2401.12818.}
\thanks{Antonino Favano and Luca Barletta are with the Dipartimento di Elettronica, Informazione e Bioingegneria, Politecnico di Milano, 20133 Milano, Italy (e-mail: \{antonino.favano, luca.barletta\}@polimi.it). Mohammadamin Baniasadi is with the University of California, Davis, CA, USA (e-mail: mbaniasadi@ucdavis.edu). Ian Zieder is with MECA Electronics, Denville, NJ, USA (e-mail: izieder@e-meca.com). Alex Dytso is with Qualcomm Flarion Technologies, Bridgewater,  NJ, USA (e-mail: odytso2@gmail.com).}
}

\maketitle

\begin{abstract}
We study the binomial channel with input alphabet $[0,1]$ and output alphabet
$\{0,\ldots,n\}$. We investigate its capacity and the structure of the
capacity-achieving input and output distributions. Since the output alphabet is
finite whereas the input alphabet is continuous, different input distributions
may induce the same output distribution; hence, uniqueness and support
properties of optimal inputs do not follow from strict concavity arguments.

We first establish structural properties of the capacity-achieving input
distribution. In particular, we show that it is discrete, unique, symmetric
around $1/2$, and contains the endpoints $\{0,1\}$ in its support. We also derive
location constraints and bounds on the probability masses of support points,
and improve the Witsenhausen-type upper bound on the support size from order
$n$ to order $n/2$.

We derive explicit nonasymptotic upper and lower bounds on capacity $C(n)$.
These bounds imply
$C(n)=\frac{1}{2}\log\!\left(\frac{n\pi}{2e}\right)+o(1)$.
The lower bound is obtained by evaluating the mutual information at the
reference input $X_r\sim\mathrm{Beta}(1/2,1/2)$, which induces a beta-binomial
output distribution, while the upper bound follows from a minimax redundancy
construction.

Finally, we prove an improved lower bound on the support size of the
capacity-achieving input distribution. We show that the beta-binomial output
induced by $X_r$ is asymptotically optimal and close to the capacity-achieving
output distribution in relative entropy and $\chi^2$ divergence. We also prove
a finite-mixture approximation lower bound showing that the beta-binomial
output cannot be approximated too accurately by binomial mixtures with few
components. Combining these results yields a support-size lower bound of order
$\Omega\!\left(\sqrt{n\log\log n}\right)$, with explicit constants. Numerical results illustrate the capacity bounds and optimal input.
\end{abstract}

\section{Introduction}

We consider a  channel for which the relationship between the input $X \in [0,1]$ and the output $Y \in  \{0,\ldots, n \}$ is described by the binomial distribution:
\begin{equation}
    P_{Y|X}(y|x) = \binom{n}{y} x^y (1-x)^{n-y}. \label{eq:Binomial_Channel}
\end{equation}
In \eqref{eq:Binomial_Channel}, we use the convention that $0^0=1$.

The objective of this paper is twofold. First, we study the capacity $C(n)$ of the binomial channel as a function of the number of trials $n$, that is
\begin{equation}\label{eq:capacity_opt_problem}
    C(n) = \max_{ P_X: \,  X \in [0,1] } I(X;Y), 
\end{equation} 
and derive explicit nonasymptotic upper and lower bounds whose gap vanishes as $n\to\infty$. Second, we investigate the structure of the capacity-achieving input distribution (CAID) $P_{X^\star}$. In particular, we study discreteness, uniqueness, symmetry, endpoint optimality, the location and probabilities of the support points, and upper and lower bounds on the cardinality of the support. Along the way, we also derive several related results, including properties of beta-binomial output distributions, estimation-theoretic identities for the binomial channel, and a finite-mixture approximation lower bound tailored to binomial mixtures. Together, these results give a more detailed picture of both the capacity and the optimizer for the binomial channel.

\subsection{Outline and Contributions}

The paper outline and main contributions are as follows. The remaining part of
Sec.~I introduces notation. Sec.~\ref{sec:Preliminaries} collects the
preliminary tools used throughout the paper. In particular,
Sec.~\ref{sec:Lit} goes over relevant literature. 
 Sec.~\ref{sec:KKT} recalls the Karush-Kuhn-Tucker (KKT) conditions for the CAID and introduces the active set $\cA_n$. Sec.~\ref{sec:EstPrelim}
develops several estimation-theoretic identities for the binomial channel,
including derivative formulas for the information density and monotonicity
properties of the conditional mean. The same section also introduces the
beta and beta-binomial distributions, the reference input
$X_r\sim\mathrm{Beta}(1/2,1/2)$, and the associated Chebyshev-polynomial
structure used later in the paper.

Sec.~\ref{sec:binomial-mixture-approximation} develops a best-approximation
result for finite binomial mixtures. In particular, we prove a Parseval-type
identity for the $\chi^2$-divergence between a binomial mixture output and the
reference beta-binomial output $P_{Y_r}$. We then use this identity, together
with a moment-matrix rank argument, to show that $P_{Y_r}$ cannot be
approximated too accurately by the output induced by an input distribution with
few mass points.

Sec.~\ref{sec:cap_bounds} derives nonasymptotic upper and lower bounds on the
capacity. The lower bound is obtained by evaluating the mutual information at
the reference input $X_r$, while the upper bound is based on the minimax
redundancy construction of Xie and Barron. These bounds imply
$C(n)=\frac12\log\frac{n\pi}{2e}+o(1)$ and provide an explicit upper bound on
the gap between $C(n)$ and $I(X_r;Y_r)$.

Sec.~\ref{sec:main_I} studies the structure of capacity-achieving input
distributions. We first show that all capacity-achieving distributions are discrete and then use this fact  to show that the CAID is unique. We further prove
that it is symmetric around $\frac12$, that the endpoints $\{0,1\}$ belong to
its support, and that there is at most one support point in each of the
intervals $\left(0,\frac1n\right]$ and
$\left[1-\frac1n,1\right)$. We also derive bounds on the probability masses of
the support points through a posterior concentration quantity, referred to as
the crest-factor.

The same section contains our main support-size results. First, we improve the
classical Witsenhausen-type upper bound of order $n$ to an upper bound of order
$n/2$. Second, we show that the beta-binomial output induced by $X_r$ is
asymptotically optimal and is close to the capacity-achieving output
distribution in relative entropy and $\chi^2$-divergence. Combining this output
approximation result with the finite-mixture lower bound from
Sec.~\ref{sec:binomial-mixture-approximation}, we obtain the improved lower
bound
$|\supp(P_{X^\star})|\ge \Omega(\sqrt{n\log\log n})$, up to explicit constants.
Furthermore, Sec.~\ref{sec:num_res} provides numerical results illustrating the capacity bounds,
the behavior of the optimal input distribution, and the support-size estimates. Finally, Sec.~\ref{sec:conclusion} concludes this work.    

\subsection{Notation}

All logarithms are to the base $\rme$. Deterministic scalar quantities are denoted by lower-case letters and random variables are denoted by uppercase letters.  
For a random variable $X$ and every measurable subset $\cA \subseteq \bbR$ the probability distribution is written as $P_{X}(\cA) = \bbP[X \in \cA]$. The support set of $P_X$ is
\begin{equation}
\supp(P_{X})=\{x:  \text{ $P_{X}( \mathcal{D})>0$ for every open set $ \mathcal{D} \ni x $ } \}. 
\end{equation} 
When $X$ is discrete, we write $P_X(x)$ for $P_X(\{x\})$, i.e., $P_X$ is a probability mass function (pmf). 

Given a function $f: \bbR \to \bbR$ and a set $\cB \subseteq \bbR$, the number of zeros of $f$ in $\cB$ is given by
\begin{equation}
\sfN\left(\cB; f\right) = \left|     \{x : f(x) =0 \} \cap \cB \right |,  
\end{equation}
where $| \cdot |$ denotes the cardinality. 

The set of the first $n$ positive integers is denoted by $[n]$. The entry in position $(i,j)$ of matrix $\sfA$ is denoted by $[\sfA]_{i,j}$.
We recall two standard divergence measures. Let $P$ and $Q$ be two
probability distributions on the same alphabet with $P\ll Q$. The relative
entropy, or KL divergence, is defined as
\begin{equation}
    \kl{P}{Q}
    =
    \bbE\left[
        \log \frac{\rmd P}{\rmd Q}(X)
    \right],  \, X \sim P, 
\end{equation}
and the  $\chi^2$-divergence is defined as
\begin{equation}
    \chi^2(P\|Q)
    =
    \bbE\left[
        \left(
        \frac{\rmd P}{\rmd Q}(X)-1
        \right)^2
    \right], \, X\sim Q. 
\end{equation}
For a positive integer $k$, the rising and the falling factorial are defined as
\begin{align}
(a)_k &= a(a+1)\cdots(a+k-1), \\
a^{\underline{k}} &= a(a-1)\cdots(a-k+1),
\end{align}
respectively.
Finally, we use standard Landau notation: for positive sequences $a_n$ and $b_n$, we write $a_n=O(b_n)$ if $\limsup_{n\to\infty} a_n/b_n<\infty$, $a_n=\Omega(b_n)$ if $\liminf_{n\to\infty} a_n/b_n>0$, $a_n=\Theta(b_n)$ if both $a_n=O(b_n)$ and $a_n=\Omega(b_n)$, and $a_n=o(b_n)$ if $a_n/b_n\to 0$.

\section{Literature Review}
\label{sec:Lit}

\subsection{Applications of Binomial Channel}
The binomial channel naturally arises in molecular communications and the interested reader is referred to \cite{farsad2020capacities,einolghozati2013design,farsad2017capacity,jamali2019channel} and references therein. The channel is also useful in the study of the deletion channel \cite{levenshtein1966binary,cheraghchi2019capacity}.   The binomial channel is also a special case of a multinomial channel, which finds applications in DNA storage \cite{mutiNomial}.

\subsection{Capacity Bounds}
The capacity of the binomial channel was first considered in \cite{komninakis2001capacity} where the authors used the minimax redundancy theorem in \cite{xie1997minimax} to argue that asymptotically the capacity scales as $\frac{1}{2} \log n $. 
The exact capacity for the $n=1$ case was computed in \cite{farsad2020capacities} where a binary distribution with support on $\{0,1 \}$ was shown to be capacity-achieving.  Subsequent
works derived both upper and lower bounds on the capacity; see, for example,
\cite{Zieder2024,TangCapBoundsBinomil}. 

In this work, one of the goals is to lower-bound the number of mass points in the
CAID $P_{X^\star}$. For this purpose, we need
capacity upper and lower bounds whose gap vanishes as $n\to\infty$. In particular, bounds with
non-vanishing additive constants are not sufficient for our analysis. We
therefore rely on the minimax-redundancy upper bound of~\cite{xie1997minimax}
and derive a new lower bound such that the gap to the upper bound vanishes as $n \to \infty$. 

\subsection{Structure of Capacity-Achieving Distribution} 
It turns out that for many channels of practical interest the CAID is discrete; the binomial channel studied in this work is one such example. In the information-theoretic literature, this observation goes back to \cite{richters1967communication}. The first general technique for proving discreteness of capacity-achieving distributions was developed in \cite{smith1971information} in the context of the amplitude-constrained additive white Gaussian noise channel, and has since been adapted to many other channels:  symmetric coherent vector additive Gaussian channels \cite{shamai1995capacity,rassouli2016capacity,vecGausEst, vectorGaussJonah}; non-coherent block-independent AWGN channels \cite{nuriyev2003capacity}; and Poisson channels  \cite{shamai1990capacity}.  Attempts were made to generalize discreteness to general additive channels in \cite{das2000capacity,tchamkerten2004discreteness, chan2005capacity,fahs2017properties}.  Generalizations to multiuser channels, such as multiple access and wiretap, are given in \cite{mamandipoor2012sum} and \cite{ozel2015gaussian,dytso2018optimal, wireTapFavano}, respectively. The reader is referred to \cite{CISS_2018} for a comprehensive summary of discreteness results. 

However, beyond showing that the support is discrete, probing the finer structure of the capacity-achieving distribution remains difficult.\footnote{There are, of course, exceptions, such as~\cite{Oettli_addive_unif,UniformAdditiveNoise}, where the exact capacity-achieving distribution is characterized for the additive channel with uniform noise. The approach in~\cite{Oettli_addive_unif,UniformAdditiveNoise} appears to exploit special structural properties of that channel, and it is not clear whether it can be extended to broader channel families. } In particular, determining the number of support points, or even its asymptotic behavior, has been an active line of research. One reason for this difficulty is that the problem requires tools from several different areas, including complex analysis tools such as identity theorem~\cite{smith1971information,dytso2019capacity} or zero counting methods~\cite{dytso2019capacity}, oscillation theorems~\cite{karlin1968total,dytso2019capacity,dytso2021properties}, and, more recently, best-approximation theory for mixtures of distributions \cite{wang2025super,improvedBoundsGaussianAmplitude}. Despite recent progress, the right set of techniques for obtaining tight support-size bounds is still not fully understood. In this sense, the present paper is part of a broader program that aims to better understand this question and to identify which techniques are needed to study the support of capacity-achieving distributions.

\subsection{Structure of Capacity-Achieving Distribution: Binomial Channel} 
   In \cite{farsad2020capacities}, for the case of $n=1$, the capacity-achieving input distribution was shown to be uniform on $\{0,1\}$. 

 It is also known that, by using the Witsenhausen technique \cite{WitsenhausenBOund}, there exists a CAID with at most $n+1$ mass points. We note, however, that the Witsenhausen technique does not guarantee that the CAID is unique. In fact, for the binomial channel, uniqueness has not been shown; note that uniqueness is important not just for theoretical purposes, but also for algorithmic purposes.  A conventional way to show that the capacity-achieving distribution is unique is by establishing that the mutual information is a \emph{strictly} concave function of the input distribution. However, as will be shown by an example, for the binomial channel, the mutual information is not strictly concave. 
Other properties, such as the location of the support points, are also not well understood. The main goal of this work is to close some of these gaps.

We note that some structural properties of CAIDs, including symmetry and the fact that the endpoints $\{0,1\}$ belong to the support, were also derived independently and concurrently with the conference version of this paper in~\cite{mutiNomial}.

 In \cite{mattingly2018maximizing}, using numerical evidence, the authors conjectured that the real scaling of the number of points is $|\supp(P_{X^\star})| = \Theta(n^{3/4})$. 

\subsection{Numerical Methods}
The authors of~\cite{komninakis2001capacity} used the ellipsoid method to compute the capacity-achieving distribution and study its structure. In~\cite{farsad2020capacities}, the slow convergence of the ellipsoid-based method was improved by using a Blahut--Arimoto-type algorithm together with an additional step that dynamically adjusts the locations of the mass points using the minimax capacity theorem. The authors of~\cite{farsad2020capacities} refer to this numerical procedure as Dynamic assignment Blahut--Arimoto (DAB) algorithm.

The authors of~\cite{mutiNomial} used a combination of the Blahut--Arimoto algorithm and learning-based methods to study CAIDs of the binomial channel under additional input constraints.

 In this work, we numerically derive the CAID for large values of $n$. In this regime, the large support cardinality of the input distribution can make the DAB algorithm in~\cite{farsad2020capacities} slow, without proper initialization. On the other hand, in our testings, the complexity of the ellipsoid method in~\cite{komninakis2001capacity} scales more favorably with $n$, but produces less accurate estimates of the optimal input distribution. 

Therefore, we derive our numerical results by using the ellipsoid method to evaluate a rough estimate of the CAID. Then, this preliminary estimate is used to initialize the DAB algorithm, which subsequently refines it and produces the final result.

\subsection{Estimation Theoretic Relationships}

In this work, we also rely on estimation theoretic quantities such as the conditional expectation. For the estimation theoretic treatments of the binomial channel, the interested reader is referred to \cite{taborda2014information,BD_Low}.

\subsection{Best Approximation with Finite Mixtures}

Recently, \cite{wang2025super} and \cite{improvedBoundsGaussianAmplitude} introduced a technique for lower-bounding the number of mass points that relies on best-approximation theory for finite mixtures. We follow a similar approach by deriving a lower bound in the setting of finite binomial mixtures. The key ideas behind this method were introduced in \cite{BestApporximationGaussianMixture}, and here we adapt and generalize them to the binomial channel.

Recently, deterministic identification capacity for the binomial channel has been studied in \cite{salariseddigh2023deterministic}.

\section {Preliminaries}
\label{sec:Preliminaries}
We now present some of the  tools needed in our analysis. 

\subsection{Karush-Kuhn-Tucker Conditions }
\label{sec:KKT}

The key that allows one to study properties of the support  of an optimal input distribution is the following lemma which contains the KKT conditions for our optimization problem~\cite{CISS_2018}.

\begin{lemma}\label{lem:KKT}
     $P_{X^\star}$ is a CAID if and only if the following conditions hold:  
\begin{align}
    i(x;P_{Y^\star}) &\le C(n), \qquad x \in [0,1], \\
    i(x;P_{Y^\star}) &= C(n), \qquad x \in \supp(P_{X^\star}) \label{eq:KKT_equality}
\end{align}
where $P_{Y^\star}$ is the output distribution induced by $P_{X^\star}$ and 
\begin{equation}
    i(x;P_{Y^\star})=\kl{P_{Y|X}(\cdot|x)}{P_{Y^\star}}. 
  \end{equation}
  
\end{lemma}

We also define the following set, which will be useful in our study of the uniqueness of $P_{X^\star}$:
\begin{equation}\label{eq:zeros_KKT}
\cA_n =\{x \in [0,1]:\:  i(x;P_{Y^\star}) -C(n)=0 \} .
\end{equation}
The importance of  $\cA_n$ is demonstrated in the following lemma. 
\begin{lemma}\label{lem:Uniquness_A_n} For a given $n \ge 1$
    \begin{itemize}
    \item $\cA_n$ is unique; and 
    \item  $\supp(P_{X^\star} )  \subseteq \cA_n$ for every $P_{X^\star}$. 
    \end{itemize}
\end{lemma}
\begin{proof}
    Note that, for a given $n$, both $P_{Y^\star}$ and $C(n)$ are unique (even if $P_{X^\star}$ is not unique) \cite{kemperman1974shannon} and, since $\cA_n$ only depends on these quantities, the uniqueness follows. 

    The second part follows from the KKT conditions in Lemma~\ref{lem:KKT}, because $x \in \supp(P_{X^\star} )$ implies $x \in \cA_n$. 
\end{proof}


\subsection{Estimation Theoretic Preliminaries}
\label{sec:EstPrelim}
Estimation theoretic quantities will play an important role in our analysis. In what follows, the quantity $\expcndn{n-1}{f(Y)}{X=x}$ denotes expectation with respect to a binomial distribution with $n-1$ trials and success probability $x$ per trial, and 
\begin{equation}
    \ell_b(x,\hat{x}) = x \log\left(\frac{x (1-\hat{x})}{(1-x)\hat{x}}\right)-\frac{x-\hat{x}}{1-\hat{x}}, \quad (x,\hat{x}) \in (0,1)^2
\end{equation}
represents the Bregman divergence \cite{bregman1967relaxation} for the binomial channel~\cite{taborda2014information}.

We now summarize some of these preliminary results. 

\begin{prop}\label{prop:derivatives_info_density} For $n\ge 2$ and $x \in (0,1)$, we have
\begin{align}
    i'(x;P_Y)  &= \frac{n}{x} \expcndn{n-1}{\ell_b(x,\expcndn{n-1}{X}{Y})}{X=x} +\frac{n}{x} \expcndn{n-1}{\frac{x-\expcndn{n-1}{X}{Y}}{1-\expcndn{n-1}{X}{Y}}}{X=x}\label{eq:inf_den_derivative_with_bregman_1}
\end{align}
and
\begin{equation}
i''(x;P_Y) = \frac{n}{x(1-x)}+\frac{1}{(1-x)^2}  G(x) \label{eq:i_second_der} 
\end{equation}
where
\begin{equation}
G(x) = \expcnd{(n-Y)(n-Y-1)\log\frac{\expcnd{X}{Y=Y}}{\expcnd{1-X}{Y=Y+1}}\frac{\expcnd{1-X}{Y=Y+2}}{\expcnd{X}{Y=Y+1}}}{X=x}. \label{eq:G_Definition}
\end{equation}
\end{prop}
\begin{IEEEproof}
See Appendix~\ref{sec:Derivatives_info_density}. 
\end{IEEEproof}

The Bregman divergence in \eqref{eq:inf_den_derivative_with_bregman_1} appeared previously in a different but related result, specifically in 
\cite[Prop.~2]{taborda2014information} it was shown that  for $a \in (0,1)$
\begin{equation}
    \frac{\partial}{ \partial a}  I(  X; \cB_n(a X) ) = \frac{n}{a} \bbE \left[ \ell_b \left(aX, \bbE[aX|\cB_{n-1}(a X')  ] \right) \right] \label{eq:derivative:idenity}
\end{equation}
where $Y=\cB_n(aX)$ denotes the transformation of input $aX$ through a binomial channel with $n$ trials.  

Finally, we also need to show the monotonicity of the conditional mean. 

\begin{lemma}\label{lem:cond_exp_non_increasing}
    The function $y \mapsto \expcnd{X}{Y=y}$ is nondecreasing.
\end{lemma}
\begin{proof}
    First of all, note that
    \begin{equation} \label{eq:Tweedy}
        \expcnd{X}{Y=y} = \frac{\expect{X^{y+1}(1-X)^{n-y}}}{\expect{X^{y}(1-X)^{n-y}}}.
    \end{equation}
    Let us now introduce the functions $f_1, f_2, g_1, g_2$ as follows:
    \begin{align}
        &f_1(x) = x^{y}, \quad f_2(x) = x^{y+1}, \quad g_1(x) = (1-x)^{n-y}, \quad g_2(x) = x(1-x)^{n-y-1},
    \end{align}
    and note that the functions 
    \begin{equation}
        \frac{f_2(x)}{f_1(x)} = x, \qquad \frac{g_2(x)}{g_1(x)} = \frac{x}{1-x}
    \end{equation}
    are both increasing and nonnegative for $x \in [0,1]$. As a consequence, the entries and the determinant of the matrices 
    \begin{equation}
        \left[\begin{array}{cc}
            f_1(x_1) & f_1(x_2) \\
            f_2(x_1) & f_2(x_2)
        \end{array}\right], \qquad \left[\begin{array}{cc}
            g_1(x_1) & g_1(x_2) \\
            g_2(x_1) & g_2(x_2)
        \end{array}\right],
    \end{equation}
    are nonnegative for any choice of $0\le x_1 < x_2 \le 1$. By using the \emph{basic composition formula} of~\cite[Ch.~3.1]{karlin1968total}, we can also say that the entries and the determinant of the matrix
    \begin{equation}
        \left[\begin{array}{cc}
            \expect{f_1(X)g_1(X)} & \expect{f_1(X)g_2(X)} \\
            \expect{f_2(X)g_1(X)} & \expect{f_2(X)g_2(X)}
        \end{array}\right]
    \end{equation}
    are nonnegative. Therefore, we have
    \begin{equation}
        \frac{\expect{f_2(X)g_2(X)}}{\expect{f_1(X)g_2(X)}} \ge \frac{\expect{f_2(X)g_1(X)}}{\expect{f_1(X)g_1(X)}}
    \end{equation}
    or
    \begin{equation}
        \frac{\expect{X^{y+2} (1-X)^{n-y-1}}}{\expect{X^{y+1}(1-X)^{n-y-1}}} \ge \frac{\expect{X^{y+1}(1-X)^{n-y}}}{\expect{X^y (1-X)^{n-y}}},
    \end{equation}
    which, by using \eqref{eq:Tweedy}, is the same as 
    \begin{equation}
        \expcnd{X}{Y=y+1} \ge \expcnd{X}{Y=y}.
    \end{equation}
    This concludes the proof. 
\end{proof}

The following expression, which can be viewed as a binomial analog of Tweedie's formula \cite{robbins1992empirical,Zieder_Mean_identies}, will also be useful. 
\begin{lemma}\label{lem:channel_transformations}
    For all $y=0,1,\ldots, n-1$, we have
        \begin{align}
        \frac{n-y}{y+1}\frac{P_Y(y)}{P_Y(y+1)} 
        &= \frac{\expcnd{1-X}{Y=y+1}}{\expcnd{X}{Y=y}}. \label{eq:ratio_cond_exp_n}
    \end{align}
\end{lemma}
\begin{IEEEproof}
    We first start by noting that
    \begin{align}
        P_{Y|X}(y+1|x) &= \binom{n}{y+1} x^{y+1} (1-x)^{n-y-1} = \frac{nx}{y+1}\binom{n-1}{y} x^{y} (1-x)^{n-y-1} = \frac{nx}{y+1} P_{Y|X}^{n-1}(y|x), \label{eq:tr_ch_1}
    \end{align}
    where we have denoted by $P_{Y|X}^{n-1}$ the binomial distribution with $n-1$ trials.
Next, note that 
    \begin{align}
        P_Y(y+1) &= \sum_x P_X(x) P_{Y|X}(y+1|x) \\
        &= \sum_x P_X(x)\frac{nx}{y+1}P_{Y|X}^{n-1}(y|x) \label{eq:apply_channel_downgrage} \\
        &= \frac{n}{y+1} \expcndn{n-1}{X}{Y=y} P_Y^{n-1}(y) \label{eq:multiply_and_divide_by_p_Y}
    \end{align}
    where in~\eqref{eq:apply_channel_downgrage} we used~\eqref{eq:tr_ch_1}; and in the last step we multiplied and divided by $P_Y^{n-1}(y)$.
    In a similar fashion, we have that 
    \begin{align}
        P_{Y|X}(y|x) &= \binom{n}{n-y} x^{y} (1-x)^{n-y} = \frac{n(1-x)}{n-y} \binom{n-1}{n-y-1} x^{y} (1-x)^{n-y-1} = \frac{n(1-x)}{n-y}P_{Y|X}^{n-1}(y|x), \label{eq:tr_ch_2}
    \end{align}
   which leads to 
    \begin{align}
        P_Y(y) &= \sum_x P_X(x) P_{Y|X}(y|x) \\
        &= \sum_x P_X(x)\frac{n(1-x)}{n-y}P_{Y|X}^{n-1}(y|x) \label{eq:apply_channel_upgrade} \\
        &= \frac{n}{n-y} \expcndn{n-1}{1-X}{Y=y} P_Y^{n-1}(y) \label{eq:multiply_and_divide_by_p_Y_2}
    \end{align}
    where in~\eqref{eq:apply_channel_upgrade} we used~\eqref{eq:tr_ch_2}; and in the last step we multiplied and divided by $P_Y^{n-1}(y)$. 
    Finally, to show~\eqref{eq:ratio_cond_exp_n},  write
    \begin{align}
       P_Y(y)\expcnd{X}{Y=y} &= \sum_x P_X(x) x  \binom{n}{y} x^{y} (1-x)^{n-y} \\
        &= \sum_x P_X(x) (1-x)  \binom{n}{y} x^{y+1} (1-x)^{n-(y+1)} \\
        &= \sum_x P_X(x) (1-x)\frac{y+1}{n-y}  \binom{n}{y+1} x^{y+1} (1-x)^{n-(y+1)} \\
        &=\frac{y+1}{n-y} P_Y(y+1) \left(1-\expcnd{X}{Y=y+1} \right)
    \end{align}
    which is valid for $y\ne n$.
\end{IEEEproof}

\subsection{Beta and Beta-Binomial Distributions: Basic Properties}

An important role in our development will be played by beta distributions and the
induced beta-binomial output distributions. Let
$X_{\alpha,\beta}\sim \mathrm{Beta}(\alpha,\beta)$, where $\alpha,\beta>0$.
That is, $X_{\alpha,\beta}$ has density
\begin{equation}\label{eq:pdf_beta_general}
    f_{X_{\alpha,\beta}}(x)
    =
    \frac{\Gamma(\alpha+\beta)}
    {\Gamma(\alpha)\Gamma(\beta)}
    x^{\alpha-1}(1-x)^{\beta-1},
    \qquad x\in(0,1).
\end{equation}
Equivalently, using the beta function
$B(\alpha,\beta)=\frac{\Gamma(\alpha)\Gamma(\beta)}
{\Gamma(\alpha+\beta)}$, this can be written as
\begin{equation}
    f_{X_{\alpha,\beta}}(x)
    =
    \frac{x^{\alpha-1}(1-x)^{\beta-1}}
    {B(\alpha,\beta)},
    \qquad x\in(0,1).
\end{equation}

When $X_{\alpha,\beta}$ is used as the input to the binomial channel in
\eqref{eq:Binomial_Channel}, the induced output distribution is the
beta-binomial distribution~\cite{JohnsonKempKotz2005} . We denote this output by $Y_{\alpha,\beta}$.
For $y=0,\ldots,n$, its probability mass function is
\begin{align}
    P_{Y_{\alpha,\beta}}(y)
    &=
    \binom{n}{y}
    \mathbb{E}\left[
        X_{\alpha,\beta}^{y}
        (1-X_{\alpha,\beta})^{n-y}
    \right] \\
    &=
    \binom{n}{y}
    \frac{B(y+\alpha,n-y+\beta)}{B(\alpha,\beta)} \label{eq:beta_binomial_general} \\
    &=
    \binom{n}{y}
    \frac{
    \Gamma(y+\alpha)\Gamma(n-y+\beta)\Gamma(\alpha+\beta)
    }
    {
    \Gamma(n+\alpha+\beta)\Gamma(\alpha)\Gamma(\beta)
    } .
\end{align}
An important special case for us will be $X_r:=X_{1/2,1/2}$ and $Y_r:=Y_{1/2,1/2}$. The density of
$X_r$ becomes
\begin{equation}\label{eq:pdf_beta}
    f_{X_r}(x)
    =
    \frac{1}{\pi\sqrt{x(1-x)}},
    \qquad x\in(0,1),
\end{equation}
and the corresponding output distribution is
\begin{equation} \label{eq:output_dist_X_r}
P_{Y_r}(y)
=
\frac{\Gamma\!\left(y+\tfrac12\right)
\Gamma\!\left(n-y+\tfrac12\right)}
{\pi\,\Gamma(y+1)\Gamma(n-y+1)},
\qquad  y=0,\dots,n.
\end{equation}
Thus, \eqref{eq:output_dist_X_r} is the special case of
\eqref{eq:beta_binomial_general} obtained by setting
$\alpha=\beta=\frac12$.

There are two intuitive reasons why beta distributions appear naturally in our analysis. First,  the beta family is conjugate prior for the binomial channel~\cite{fink1997compendium}. 
In other words, the conditional distribution of $X_{\alpha,\beta}$ given $Y_{\alpha,\beta}=y$ remains a beta distribution, with only its parameters updated. 
Fig. \ref{fig:beta_densities} and Fig. \ref{fig:beta_binomial_pmfs} show beta input distributions and their corresponding output distributions, respectively.

\begin{figure}[t]
\centering
\begin{minipage}[t]{0.45\linewidth}
\centering
\begin{tikzpicture}
\begin{axis}[
    width=\linewidth,
    height=6cm,
    xlabel={$x$},
    ylabel={$f_X(x)$},
    xmin=0, xmax=1,
    ymin=0, ymax=8.5,
    xtick={0,0.2,0.4,0.6,0.8,1.0},
    ytick={0,5, 10, 15, 20, 25, 30},
    xmajorgrids=true,
    ymajorgrids=true,
    legend pos=north east,
    legend style={draw=gray!40, fill=white, font=\footnotesize},
    tick label style={font=\footnotesize},
    label style={font=\footnotesize},
]
\addplot[red, very thick, domain=0.001:0.999, samples=500]
    {1/(pi*sqrt(x*(1-x)))};
\addlegendentry{$\alpha=\beta=1/2$}

\addplot[blue, very thick, domain=0.001:0.999, samples=2]
    {1};
\addlegendentry{$\alpha=\beta=1$}

\addplot[black, very thick, domain=0.001:0.999, samples=500]
    {6*x*(1-x)};
\addlegendentry{$\alpha=\beta=2$}
\end{axis}
\end{tikzpicture}
\caption{Beta input densities for representative symmetric beta distributions.}
\label{fig:beta_densities}
\end{minipage}%
\pgfplotstableread{
y pmf_half pmf_one pmf_two
0 0.102578173009 0.0322580645161 0.00568181818182
1 0.0521583930552 0.0322580645161 0.0109970674487
2 0.0398050894369 0.0322580645161 0.0159457478006
3 0.0337740152798 0.0322580645161 0.0205278592375
4 0.0301098532447 0.0322580645161 0.0247434017595
5 0.0276302182716 0.0322580645161 0.0285923753666
6 0.0258445919207 0.0322580645161 0.0320747800587
7 0.0245091570798 0.0322580645161 0.0351906158358
8 0.0234879422015 0.0322580645161 0.0379398826979
9 0.0226989415591 0.0322580645161 0.0403225806452
10 0.022089945566 0.0322580645161 0.0423387096774
11 0.0216265201346 0.0322580645161 0.0439882697947
12 0.0212855614838 0.0322580645161 0.0452712609971
13 0.0210516542147 0.0322580645161 0.0461876832845
14 0.0209149551614 0.0322580645161 0.0467375366569
15 0.0208699767632 0.0322580645161 0.0469208211144
16 0.0209149551614 0.0322580645161 0.0467375366569
17 0.0210516542147 0.0322580645161 0.0461876832845
18 0.0212855614838 0.0322580645161 0.0452712609971
19 0.0216265201346 0.0322580645161 0.0439882697947
20 0.022089945566 0.0322580645161 0.0423387096774
21 0.0226989415591 0.0322580645161 0.0403225806452
22 0.0234879422015 0.0322580645161 0.0379398826979
23 0.0245091570798 0.0322580645161 0.0351906158358
24 0.0258445919207 0.0322580645161 0.0320747800587
25 0.0276302182716 0.0322580645161 0.0285923753666
26 0.0301098532447 0.0322580645161 0.0247434017595
27 0.0337740152798 0.0322580645161 0.0205278592375
28 0.0398050894369 0.0322580645161 0.0159457478006
29 0.0521583930552 0.0322580645161 0.0109970674487
30 0.102578173009 0.0322580645161 0.00568181818182
}\betabinomialpmfdata %
\hspace{4ex} %
\begin{minipage}[t]{0.45\linewidth}
\centering
\begin{tikzpicture}
\begin{axis}[
    width=\linewidth,
    height=6cm,
    xlabel={$y$},
    ylabel={$P_Y(y)$},
    xmin=0, xmax=30,
    ymin=0, ymax=0.11,
    xtick={0,5,10,15,20,25,30},
    ytick={0.02,0.04,0.06,0.08,0.10},
    xmajorgrids=true,
    ymajorgrids=true,
    legend pos=north east,
    legend style={draw=gray!40, fill=white, font=\footnotesize},
    tick label style={font=\footnotesize},
    label style={font=\footnotesize},
]
\addplot+[red, very thick, mark=*, mark size=1.5pt]
    table[x=y,y=pmf_half] {\betabinomialpmfdata};
\addlegendentry{$\alpha=\beta=1/2$}

\addplot+[blue, very thick, mark=square*, mark size=1.5pt]
    table[x=y,y=pmf_one] {\betabinomialpmfdata};
\addlegendentry{$\alpha=\beta=1$}

\addplot+[black, very thick, mark=triangle*, mark size=1.7pt]
    table[x=y,y=pmf_two] {\betabinomialpmfdata};
\addlegendentry{$\alpha=\beta=2$}
\end{axis}
\end{tikzpicture}
\caption{Beta-binomial output distributions induced by the input distributions of Fig.~\ref{fig:beta_densities} for $n=30$.}
\label{fig:beta_binomial_pmfs}
\end{minipage}
\end{figure}

In particular, if $X_r|Y_r=y \sim \text{Beta}(y+1/2, n-y+1/2)$,  the moments of this conditional distribution, which will be useful, are given by  \cite{JohnsonKempKotz2005}
\begin{equation}
  \bbE[X_r^k|Y_r=y] = \frac{(y+1/2)_k}{(n+1)_k}, \label{eq:moments_of_beta}
\end{equation}
where we recall that  $(a)_k $ is the rising factorial.

 The second reason is that $X_r$ is the Jeffreys prior for the binomial channel. Recall that the Jeffreys prior $\pi(\theta)$, when it exists, is defined as
\begin{equation}
\pi(\theta) \propto \sqrt{I(\theta)},
\end{equation}
where $I(\theta)$ is the Fisher information of the underlying parametric family; the interested reader is referred to \cite{kass1996selection} for detailed discussion on Jeffreys prior.  For the binomial channel in~\eqref{eq:Binomial_Channel}, the parameter of interest is $\theta=x$. The corresponding Fisher information is
\begin{equation}
I(x)=\frac{n}{x(1-x)}, \qquad x\in(0,1).
\end{equation}
Therefore,
\begin{equation}
\pi(x)\propto \frac{1}{\sqrt{x(1-x)}},
\end{equation}
which is precisely the density of the $\mathrm{Beta}(1/2,1/2)$ distribution, up to normalization. Thus, while all beta distributions are conjugate priors for the binomial channel, only $P_{X_r}$ is both a conjugate prior and Jeffreys' prior within this family.
 
In this work, we will show that $X_r$ is an asymptotically optimal input distribution. 
We now list a few important properties  of  $P_{Y_r}$. 
\begin{prop} \label{prop:properties_Y_r}
The following hold: 
\begin{itemize}
    \item \emph{Ratio.} For $y=0,\ldots,n-1$,
    \begin{align}
        \frac{P_{Y_r}(y+1)}{P_{Y_r}(y)}
        &=
        \frac{y+\tfrac12}{y+1}
        \cdot
        \frac{n-y}{n-y-\tfrac12}.
        \label{eq:qstar_ratio}
    \end{align}

    \item \emph{Minimum value.} For every $y\in\{0,\ldots,n\}$,
    \begin{equation}
        P_{Y_r}(y) \ge  P_{Y_r} \left ( \left \lfloor  \frac{n}{2}  \right \rfloor \right)
        \ge
        \frac{2}{\pi(n+2)}.
        \label{eq:qstar_min_bound}
    \end{equation}

    \item \emph{Pointwise bound.} For every $y= 0,\ldots,n $,
    \begin{equation}
    \frac{1}{
        \pi
        \sqrt{
            (y+\frac12)(n-y+\frac12)
        }} 
        \le   P_{Y_r}(y)
        \le
        \frac{1}{
        \pi\sqrt{
        \left(y+\tfrac{1}{4}\right)
        \left(n-y+\tfrac{1}{4}\right)
        }} .
        \label{eq:tail_bound}
    \end{equation}
\end{itemize}
\end{prop}

\begin{IEEEproof}
See Appendix~\ref{app:prop:properties_Y_r}.
\end{IEEEproof}

\subsection{Beta, Beta-Binomial Distributions and Orthogonal Polynomials}

In order to establish approximation bounds between  $P_{Y_r}$  and  $P_{Y^\star}$ in terms of $\chi^2$-divergence, one will need to rely on the theory of orthogonal polynomials. 
The orthogonal polynomials for the $\text{Beta}(1/2, 1/2)$ distribution are the shifted Chebyshev polynomials of the first kind, defined as: for $k \in \{0,1,\ldots \}$ 
\begin{equation}
    \tilde{T}_k(x) = T_k(2x-1) = \cos(k \arccos(2x-1)), \, x \in (0,1),
\end{equation} 
where $T_k$ is the Chebyshev polynomial of the
first kind \cite[Eq.~18.7.7]{olver2010nist}.
The $\tilde{T}_k$'s are orthogonal on $(0,1)$ with
weight $[x(1-x)]^{-1/2}$ \cite[Table~18.3.1]{olver2010nist}. For $X_r \sim \text{Beta}(1/2, 1/2)$, since 
$f_{X_r}(x) = \pi^{-1}[x(1-x)]^{-1/2}$, the normalization gives: 
\begin{equation} \label{eq:Chebyshev_orthog}
    \mathbb{E} \left[\tilde{T}_k(X_r) \tilde{T}_m(X_r) \right] = 
    \begin{cases} 
      1 & \text{if } k = m = 0, \\
      \frac{1}{2} & \text{if } k = m \ge 1, \\
      0 & \text{if } k \neq m.
   \end{cases}
\end{equation}

We now list a few additional properties.  
\begin{prop} Fix $k \le n$. Consider a polynomial expansion of $\tilde{T}_k(x)$:
    \begin{equation}
        \tilde{T}_k(x) = \sum_{i=0}^k a_i x^i, \quad x \in (0,1)
    \end{equation}
    and define a polynomial
    \begin{equation}
             H_k(y)
    =
    \sum_{j=0}^k a_j \frac{y^{\underline{j}}}{n^{\underline{j}}},
    \quad y\in\{0,\ldots,n\}. \label{Eq:H_k_basesis}
    \end{equation}

Then, the following properties hold:
\begin{itemize}
    \item \emph{$H_k$ vs. $\tilde{T}_k$}:
    \begin{equation}
        \bbE[ H_k (Y) | X =x] =  \tilde{T}_k(x), \, x \in (0,1).
    \end{equation}
    \item \emph{Orthogonality of $H_k$}: for $m <k \le n$
    \begin{equation}
        \mathbb{E}[H_k(Y_r) Y_r^m] =0.
    \end{equation}
    \item \emph{Norm of $H_k$}: 
 \begin{equation}
   \mathbb{E}[H_k^2(Y_r)]  =\frac{(n+1)_k}{n^{\underline{k}}} \mathbb{E}[\tilde{T}_k(X_r)^2] = \frac{1}{2} \prod_{j=1}^k \frac{n+j}{n-j+1}.
\end{equation}
\end{itemize}
\end{prop}
\begin{IEEEproof}
To show the first property note that since $Y|X=x\sim \mathrm{Bin}(n,x)$, we have \cite{knoblauch2008closed}
\begin{equation}
\bbE[Y^{\underline{j}}|X=x]=n^{\underline{j}} x^j,
    \qquad j=0,\ldots,n .
\end{equation}
Now, since  $H_k$ is a polynomial
in $y$ of degree at most $k$, we have that 
\begin{align}
    \bbE[H_k(Y)|X=x]
    &=
    \sum_{j=0}^k a_j \frac{1}{n^{\underline{j}}}
 \bbE[Y^{\underline{j}}|X=x]  =
    \sum_{j=0}^k a_j x^j =
    \tilde T_k(x).
\end{align}
This proves the claim.

To show the second property, note that from \eqref{eq:moments_of_beta}, there exists a polynomial $q_m$ of degree $m$ such that 
\begin{equation}
\bbE[q_m(X_r)|Y_r=y] =y^m, \quad  y \in \{0,1,\ldots, n\}. \label{eq:existance_of_qm}
\end{equation}
We also recall that by the orthogonality property of the conditional expectation, for any pair of random variables $(U,V)$ and any two functions $f,g$ \cite{resnick1999probability}:
\begin{equation}
    \bbE \left[ f(U) \bbE[ g(V)|U] \right] = \bbE \left[ \bbE[f(U)|V]  g(V) \right]. \label{eq:orthog_prop}
\end{equation}
Next, note that 
 \begin{align}
        \mathbb{E}[H_k(Y_r) Y_r^m] &=  \mathbb{E}[H_k(Y_r)  \bbE[q_m(X_r)|Y_r]] \label{eq:defining_qm}    \\
       &=\mathbb{E} \left[  \bbE[ H_k(Y_r) | X_r ]q_m(X_r) \right] \label{eq:first_orthog_use} \\
       &=\mathbb{E} \left[  \tilde T_k(X_r)q_m(X_r) \right] \\
       &=0, \label{eq:using_orthog}
    \end{align}
    where \eqref{eq:defining_qm} follows from \eqref{eq:existance_of_qm}; \eqref{eq:first_orthog_use} follows from the \eqref{eq:orthog_prop}; and \eqref{eq:using_orthog} follows by  noting that we can write  $q_m(X_r) = \sum_{i=0}^m b_i \tilde{T}_i(X_r) $ and using \eqref{eq:Chebyshev_orthog} together with $k>m$.  This concludes the proof of the second property.

  We now compute the squared norm of $H_k$. We consider $k\ge 1$; the case $k=0$ is immediate since $H_0=1$. As above, there exists a polynomial $q_k$ of degree $k$ such that
\begin{equation}
\bbE[q_k(X_r)\mid Y_r=y]=H_k(y),
\qquad y\in\{0,\ldots,n\}.
\end{equation}
Using \eqref{eq:orthog_prop}, we obtain
\begin{align}
\bbE[H_k^2(Y_r)]
&=
\bbE\left[
H_k(Y_r)\bbE[q_k(X_r)\mid Y_r]
\right] =
\bbE\left[
q_k(X_r)\bbE[H_k(Y_r)\mid X_r]
\right] =
\bbE[q_k(X_r)\tilde T_k(X_r)].
\end{align}

Since $q_k$ has degree $k$, we can write
\begin{equation}
q_k(x)
=
b_k\tilde T_k(x)+r_{k-1}(x),
\end{equation}
where $r_{k-1}$ is a polynomial of degree at most $k-1$. By the orthogonality of $\tilde T_k$ to all polynomials of degree strictly less than $k$,
\begin{equation}
\bbE[r_{k-1}(X_r)\tilde T_k(X_r)]=0.
\end{equation}
Therefore,
\begin{equation}
\bbE[H_k^2(Y_r)]
=
b_k\bbE[\tilde T_k^2(X_r)].
\end{equation}

It remains to identify $b_k$. Let $c_k$ be the leading coefficient of $\tilde T_k$, so that
\begin{equation}
\tilde T_k(x)=c_kx^k+\text{lower order terms}.
\end{equation}
By the definition of $H_k$ in \eqref{Eq:H_k_basesis},  we have
\begin{equation}
H_k(y)=\frac{c_k}{n^{\underline{k}}}y^k+\text{lower order terms}.
\end{equation}
On the other hand, by \eqref{eq:moments_of_beta},
\begin{align}
\bbE[X_r^k\mid Y_r=y]
&=
\frac{(y+1/2)_k}{(n+1)_k}=
\frac{1}{(n+1)_k}y^k+\text{lower order terms}.
\end{align}
Since
\begin{equation}
q_k(x)=b_k\tilde T_k(x)+r_{k-1}(x),
\end{equation}
the leading coefficient of $q_k$ is $b_kc_k$. Hence
\begin{equation}
\bbE[q_k(X_r)\mid Y_r=y]
=
\frac{b_kc_k}{(n+1)_k}y^k+\text{lower order terms}.
\end{equation}
Comparing this with $\bbE[q_k(X_r)\mid Y_r=y]=H_k(y)$ gives $
\frac{b_kc_k}{(n+1)_k}
=
\frac{c_k}{n^{\underline{k}}}$
and, therefore,
$
b_k=
\frac{(n+1)_k}{n^{\underline{k}}}$. Thus,
\begin{align}
\bbE[H_k^2(Y_r)]
&=
b_k\bbE[\tilde T_k^2(X_r)] =
\frac{(n+1)_k}{n^{\underline{k}}}
\bbE[\tilde T_k^2(X_r)] =
\frac{1}{2}
\prod_{j=1}^k
\frac{n+j}{n-j+1}.
\end{align}
where the last equality follows from \eqref{eq:Chebyshev_orthog}. This concludes the computation of the norm.

This concludes the proof.
\end{IEEEproof}

\section{On the Best Approximation Theory of Finite Binomial Mixtures}
\label{sec:binomial-mixture-approximation}

The key result for providing a lower bound on the support size is the following theorem. It quantifies the best possible approximation of the reference beta-binomial distribution $P_{Y_r}$ by a binomial mixture with finitely many components. Its proof follows from an adaptation of the trigonometric moment method for mixtures of Gaussian distributions of \cite[Thm.~7]{BestApporximationGaussianMixture}.
\subsection{Expression for $\chi^2$-Divergence via Parseval's Identity}

One of the ingredients for our proof will be the following representation of the $\chi^2$-divergence. 

\begin{prop}  \label{prop:Xidiv}
Let $P_Y$ and $P_{Y_r}$ be the output distributions induced by $P_X$ and $P_{X_r}$, respectively, through the binomial channel. Then,  
    \begin{equation}
        \chi^2(P_Y \| P_{Y_r})
     = \sum_{k=1}^n \frac{\epsilon_k^2}{h_k}, \label{eq:chi_square_div_Parseval} 
    \end{equation}
    where $\epsilon_k = \mathbb{E}[\tilde{T}_k(X)]$ and $h_k = \mathbb{E}[H_k^2(Y_r)]  = \frac{1}{2}
    \prod_{j=1}^k
    \frac{n+j}{n-j+1}$.
\end{prop}
\begin{IEEEproof}
We expand the density ratio in the polynomials $H_k(y)$, which were defined in \eqref{Eq:H_k_basesis}:
\begin{equation}
    \frac{P_Y(y)}{P_{Y_r}(y)} = \sum_{k=0}^n a_k H_k(y).
\end{equation}
The coefficients are given by
\begin{align}
    a_k &= \frac{1}{h_k} \mathbb{E}\left[ \frac{P_Y(Y_r)}{P_{Y_r}(Y_r)} H_k(Y_r) \right] = \frac{1}{h_k} \mathbb{E}[H_k(Y)] = \frac{1}{h_k} \mathbb{E}[ \bbE[H_k(Y)|X]] = \frac{1}{h_k} \mathbb{E}[\tilde{T}_k(X)]. \label{eq:expression_for_a_k}
\end{align} 
Note that $\epsilon_0 = 1$ and $h_0 = 1$. Next, observe that 
\begin{align}
    \chi^2(P_Y \| P_{Y_r}) &= \mathbb{E}\left[ \left( \frac{P_Y(Y_r)}{P_{Y_r}(Y_r)} - 1 \right)^2 \right] \\
    &= \mathbb{E}\left[ \left( \frac{P_Y(Y_r)}{P_{Y_r}(Y_r)}  \right)^2 \right] -1 \\
    &= \sum_{k=1}^n a_k^2 h_k  \label{eq:using_Parasevals} \\
    &= \sum_{k=1}^n \frac{\epsilon_k^2}{h_k}, \label{eq:expression_for_a_k_use}
\end{align}
where in \eqref{eq:using_Parasevals} we have used Parseval's identity; and in \eqref{eq:expression_for_a_k_use} we have used \eqref{eq:expression_for_a_k}. This concludes the proof.  
    \end{IEEEproof}

\subsection{Best Approximation Lower Bound}
The main result of this section is the following theorem, which might be of independent interest. 
\begin{theorem} \label{thm:main}
Let $P_X$ be a discrete distribution  with $K$ mass points in $[0,1]$. Let $P_Y$ and $P_{Y_r}$ be the output distributions induced by $P_X$ and $P_{X_r}$, respectively, through the binomial channel.  Then, for any integer $L$ such that $K < L \le \frac{n+2}{2}$, we have
\begin{equation}
    \chi^2(P_Y \| P_{Y_r}) \ge B_n(L) := \frac{L-K}{2L \prod_{j=1}^{2L-2} \frac{n+j}{n-j+1}}.
    \label{lowerchi}
\end{equation}
\end{theorem}
\begin{IEEEproof}
For any integer $m \le n$, we can truncate the sum to lower bound the divergence \eqref{eq:chi_square_div_Parseval}:
\begin{equation} \label{eq:chi2_trunc}
    \chi^2(P_Y \| P_{Y_r}) \ge \sum_{k=1}^{m} \frac{\epsilon_k^2}{h_k}.
\end{equation}

Let $L$ be an integer such that $K < L$ and $2L-2 \le n$. Define the $L \times L$ moment matrix $M$ with entries 
\begin{equation}
    M_{i,j} = \mathbb{E}[\tilde{T}_i(X) \tilde{T}_j(X)]
\end{equation}
for $0 \le i, j \le L-1$. Since $X$ is supported on $K$ mass points, $M$ can be written as the sum of $K$ rank-1 matrices, implying $\text{rank}(M) \le K$.

Let $D$ be the corresponding matrix for the reference distribution $X_r$, so 
\begin{equation}
    D_{i,j} = \mathbb{E}[\tilde{T}_i(X_r) \tilde{T}_j(X_r)].
\end{equation} The matrix $D$ is diagonal with entries $D_{0,0} = 1$ and $D_{i,i} = 1/2$ for $i \ge 1$. The singular values of $D$ are $1$ (with multiplicity 1) and $1/2$ (with multiplicity $L-1$). 

By the Eckart-Young-Mirsky theorem \cite{golub1988generalization}, the Frobenius distance from $D$ to any matrix of rank at most $K$ is lower-bounded by the sum of the squared $L-K$ smallest singular values of $D$. Since $K \ge 1$, the $L-K$ smallest singular values are all $1/2$. Thus,
\begin{equation} \label{eq:EYM}
    \|M - D\|_F^2 \ge (L-K) \left(\frac{1}{2}\right)^2 = \frac{L-K}{4}.
\end{equation}

Using the trigonometric identity $\cos(i\theta)\cos(j\theta) = \frac{1}{2}(\cos((i+j)\theta) + \cos(|i-j|\theta))$, the product of Chebyshev polynomials linearizes as
\begin{equation}
    \tilde{T}_i(x) \tilde{T}_j(x) = \frac{1}{2} \left( \tilde{T}_{i+j}(x) + \tilde{T}_{|i-j|}(x) \right).
\end{equation}
Thus, the entries of the moment matrix $M$ are given by $M_{i,j} = \frac{1}{2} (\epsilon_{i+j} + \epsilon_{|i-j|})$.
To  express the difference matrix $M - D$, we define a modified sequence $\tilde{\epsilon}_k$:
\begin{equation}
    \tilde{\epsilon}_k = \begin{cases} \epsilon_k & \text{if } k \ge 1, \\ 0 & \text{if } k = 0. \end{cases}
\end{equation}
Since $\epsilon_0 = 1$, it is straightforward to verify that for all $0 \le i, j \le L-1$,
\begin{equation}
    M_{i,j} - D_{i,j} = \frac{1}{2} (\tilde{\epsilon}_{i+j} + \tilde{\epsilon}_{|i-j|}).
\end{equation}
Indeed, for $i \neq j$, $D_{i,j} = 0$ and the indices $i+j, |i-j| \ge 1$. For $i = j \ge 1$, $D_{i,i} = 1/2$ and $\frac{1}{2}(\tilde{\epsilon}_{2i} + \tilde{\epsilon}_0) = \frac{1}{2}\epsilon_{2i} = M_{i,i} - 1/2$. For $i=j=0$, $D_{0,0} = 1$ and $\frac{1}{2}(\tilde{\epsilon}_0 + \tilde{\epsilon}_0) = 0 = M_{0,0} - 1$.

Using the inequality $(a+b)^2 \le 2a^2 + 2b^2$, we bound the squared entries:
\begin{equation}
    (M_{i,j} - D_{i,j})^2 \le \frac{1}{2} (\tilde{\epsilon}_{i+j}^2 + \tilde{\epsilon}_{|i-j|}^2).
\end{equation}
Summing over all $0 \le i, j \le L-1$:
\begin{equation}
    \|M - D\|_F^2 \le \frac{1}{2} \sum_{i,j=0}^{L-1} (\tilde{\epsilon}_{i+j}^2 + \tilde{\epsilon}_{|i-j|}^2) = \sum_{k=1}^{2L-2} S_k \tilde{\epsilon}_k^2,
\end{equation}
where $S_k = \frac{1}{2} (N_1(k) + N_2(k))$, with $N_1(k)$ being the number of pairs $(i,j) \in \{0,\dots,L-1\}^2$ such that $i+j=k$, and $N_2(k)$ the number of pairs such that $|i-j|=k$. The $k=0$ term vanishes because $\tilde{\epsilon}_0 = 0$.
For $1 \le k \le L-1$, $N_1(k) = k+1$ and $N_2(k) = 2(L-k)$, giving $S_k = L - \frac{k-1}{2} \le L$.
For $L \le k \le 2L-2$, $N_1(k) = 2L - 1 - k$ and $N_2(k) = 0$, giving $S_k = L - \frac{k+1}{2} \le L$.
Since $\tilde{\epsilon}_k = \epsilon_k$ for $k \ge 1$, we obtain:
\begin{equation}
    \|M - D\|_F^2 \le L \sum_{k=1}^{2L-2} \epsilon_k^2.
\end{equation}
Combining this with \eqref{eq:EYM} yields:
\begin{equation} \label{eq:epsilon_bound}
    \sum_{k=1}^{2L-2} \epsilon_k^2 \ge \frac{L-K}{4L}.
\end{equation}

Notice that $h_k$ in Proposition~\ref{prop:Xidiv} is strictly increasing with $k$ because each factor $\frac{n+j}{n-j+1} > 1$ for $j \ge 1$. Thus, $h_k \le h_{2L-2}$ for all $1 \le k \le 2L-2$. Returning to the $\chi^2$ divergence \eqref{eq:chi2_trunc} with $m = 2L-2$:
\begin{equation}
    \chi^2(P_Y \| P_{Y_r}) \ge \sum_{k=1}^{2L-2} \frac{\epsilon_k^2}{h_k} \ge \frac{1}{h_{2L-2}} \sum_{k=1}^{2L-2} \epsilon_k^2 \ge \frac{L-K}{4L h_{2L-2}}.
\end{equation}
Substituting the expression for $h_{2L-2}$ yields the final lower bound:
\begin{equation}
\label{chisq}
    \chi^2(P_Y \| P_{Y_r}) \ge \frac{L-K}{2L \prod_{j=1}^{2L-2} \frac{n+j}{n-j+1}}.
\end{equation}
This completes the proof of Theorem \ref{thm:main}.
\end{IEEEproof}

The lower bound in Theorem~\ref{thm:main} depends on the free integer parameter $L$ satisfying $K < L$ and $2L-2 \le n$. To obtain the strongest possible bound, we should maximize $B_n(L)$ over all admissible values of $L$. This optimization does not have a closed-form solution. Thus, we establish an explicit lower bound in the next proposition.
\begin{prop}\label{prop:uniform-explicit}
Let $L$ be an integer such that $K < L \le \frac{n+2}{2}$. Then,
\begin{equation}
B_{n}(L)
\ge
\frac{L-K}{2L}\exp\!\left(-\frac{(2L-2)^2}{\,n-2L+3\,}\right).
\end{equation}
\end{prop}
\begin{IEEEproof}
Set $m:=2L-2$. Then
\begin{align}
\log\!\left(\prod_{j=1}^{m}\frac{n+j}{n-j+1}\right)
&=
\sum_{j=1}^{m}\log\!\left(1+\frac{2j-1}{n-j+1}\right) \\
&\le
\sum_{j=1}^{m}\frac{2j-1}{n-j+1} \label{eq:use_logineq} \\
&\le
\frac{1}{n-m+1}\sum_{j=1}^{m}(2j-1) \\
&=\frac{m^2}{n-m+1}
\end{align}
where in \eqref{eq:use_logineq} we used $\log(1+x)\le x$ for $x>-1$. Substituting into $B_n(L)$, we have the final result
\[
B_{n}(L)
\ge
\frac{L-K}{2L}\exp\!\left(-\frac{m^2}{n-m+1}\right).
\] 

\end{IEEEproof}

\section{Capacity and Bounds on the Capacity}
\label{sec:cap_bounds}

In this section,  we provide upper and lower bounds on capacity. The upper bound relies on the dual representation of capacity as:
\begin{equation}
   C(n) = \inf_q \max_{x \in [0,1]} \kl{P_{Y|X}(\cdot|x)}{q},
\end{equation}
which, by properly choosing an auxiliary output distribution \(q\), often leads to order-tight bounds. The  reader is referred to \cite{lapidoth2008capacity,mckellips2004simple,boundsAmplt} for applications to other channels. It will also be convenient to work with continuous output, and we will use the following channel output: $\tilde{Y} = Y + U$, where $U \sim \mathcal{U}(0,1)$ is independent of $Y$. Note that because the distance between original $Y$ points is one, such additive noise can be completely filtered out, and we have $I(X; Y) = I(X; Y+ U)$ for all $X$. This trick has been used before in the context of the Poisson channel ~\cite{lapidoth2008capacity}.

The lower bound on the capacity will follow from choosing a convenient input distribution. The exact computation, however, will not be possible, and some further bounds on the entropy of the binomial distribution will be needed. Therefore,  in Appendix~\ref{app:bound_binom_entropy}, we also provide a new upper bound on the entropy of a binomial distribution. Bounds   on the entropy of a binomial distribution have been considered before in \cite{knessl1998integral,cheraghchi2018expressions}.
\subsection{Asymptotically Tight Capacity Bounds} \label{sec:asympt_bounds}
The next theorem provides a lower bound on the capacity by evaluating mutual information at the reference input $X_r\sim\mathrm{Beta}(1/2,1/2)$. 
\begin{theorem}\label{thm:capacity-lb}
For every $n\ge 1$,
\begin{equation}\label{eq:capacity-lb}
C(n) \ge \psi(n+1)-\log\bigl(1+\sqrt{3n+1}\bigr)+\frac12\log\!\left(\frac{3\pi}{2e}\right),
\end{equation}
where $\psi(\cdot)$ denotes the digamma function. In particular,
\begin{equation}\label{eq:capacity-lb-asymp}
C(n) \ge \underline{C}(n) \triangleq \frac12\log\!\left(\frac{n\pi}{2e}\right) + r_{{\rm LB}}(n),
\end{equation}
where 
\begin{equation}
r_{{\rm LB}}(n) \triangleq \frac12\log\!\left(1+\frac{1}{3n}\right)-\log\!\left(1+\frac{1}{\sqrt{3n+1}}\right)-\frac{1}{n+1}. 
\end{equation}
\end{theorem}
\begin{IEEEproof}
    The proof is given in Appendix~\ref{sec:proof_thm:capacity-lb}.
\end{IEEEproof}

The following upper bound is based on the minimax redundancy construction of Xie and Barron~\cite{xie1997minimax}.

\begin{theorem}\label{thm:capacity-ub}
For all $n\ge 28$, we have
\begin{equation}\label{eq:capacity-ub}
C(n) \le \overline{C}(n) \triangleq \frac12\log\!\left(\frac{n\pi}{2e}\right) + r_{{\rm UB}}(n),
\end{equation}
where
\begin{equation}
\label{resupper}
r_{{\rm UB}}(n) \triangleq -\log\!\left(1-2\left(\frac{2e}{n\pi}\right)^{1/4}\right)
+\frac{10}{\log\frac{n\pi}{2e}}.
\end{equation}
\end{theorem}
\begin{IEEEproof}
    The proof is given in Appendix~\ref{sec:proof_thm:capacity-up}.
\end{IEEEproof}

\subsection{Finite-$n$ Upper Bounds}
While the upper bound in Theorem~\ref{thm:capacity-ub} provides the correct asymptotic scaling in $n$, we also derive the following upper bounds, which more closely characterize the capacity for finite $n$, which is of practical interst. 
\begin{theorem} \label{thm:C1_UB_finite_n}
	For any $n \geq 1$,
	\begin{align}
		C(n) &\leq \overline{C}_1(n) \triangleq\log \left( \sum_{y=0}^n \binom{n}{y} \left( \frac{y}{n}\right)^y \left(1- \frac{y}{n}\right)^{n-y} \right), \label{eq:C1_UB_finite_n}
	\end{align}
	with the convention that $0^0 = 1$ at $y=0,n$.
\end{theorem}
\begin{IEEEproof}
    The proof is given in Appendix~\ref{sec:C_UB_finite_n}.
\end{IEEEproof}
From Theorem~\ref{thm:C1_UB_finite_n}, we can derive a slightly looser but even simpler upper bound.
\begin{corollary} \label{cor:C2_UB_finite_n}
	For any $n \geq 1$, further upper-bounding the results of Theorem~\ref{thm:C1_UB_finite_n}, provides
	\begin{align}
		C(n)  &\leq \overline{C}_1(n) \leq \overline{C}_2(n) \triangleq \log \left( 2+ \sqrt{\frac{n \pi }{2}}\right) \label{eq:C2_UB_finite_n}.
	\end{align}
\end{corollary}
\begin{IEEEproof}
    The proof is given in Appendix~\ref{sec:C_UB_finite_n}.
\end{IEEEproof}

We now show that, in all practical cases, $\overline{C}_1$ and $\overline{C}_2$ are tighter than $\overline{C}$ of Theorem \ref{thm:capacity-ub}. Define the gap between upper bounds $\Delta(n) \triangleq \overline{C}(n) - \overline{C}_2(n)$, for $n \geq 28$. A lower bound on this gap is as follows:
\begin{align}
    \Delta(n) &\geq 2 \left( \frac{2 \rme}{n \pi}\right)^{\frac{1}{4}} -2 \left( \frac{2}{n \pi} \right)^{\frac{1}{2}} + \frac{10}{\log \frac{n \pi}{2 \rme}}-\frac{1}{2} \label{eq:UB_gap1} \\
    &\geq \frac{10}{\log \frac{n \pi}{2 \rme}}-\frac{1}{2}, \label{eq:UB_gap2}
\end{align}
where~\eqref{eq:UB_gap1} holds by the bounds $\log t \leq t-1, \ t>0$ and~\eqref{eq:UB_gap2} holds by noticing that the difference of the first two terms in~\eqref{eq:UB_gap1} is always positive for $n \geq 1$. 

Since~\eqref{eq:UB_gap2} is positive at $n=28$, strictly decreasing, and crosses zero at $\bar{n} = \frac{2}{\pi}\rme^{21} \approx 8.4 \cdot 10^8$, we have that $\overline{C}_2$ is tighter than $\overline{C}$ at least up to $\bar{n}$. Finally, since $\overline{C}_1 \leq \overline{C}_2$, the same holds also for $\overline{C}_1$. 

Notice that, while the upper bound presented in Theorem~\ref{thm:capacity-ub} is not the smallest for finite values of $n$, it eventually becomes the tightest as $n$ goes to infinity. Moreover, the upper bound $\overline{C}$ might be further improved by refining the choice of the constants $\eta_n$ and $c_n$, see Appendix~\ref{sec:proof_thm:capacity-up} and~\cite[Sec.~III-B]{xie1997minimax}.

In Section VII, Fig.~\ref{fig:capacity} compares the capacity bounds derived in this section and Section~\ref{sec:asympt_bounds}
with the numerical estimate of the capacity and the existing bounds in \cite{Zieder2024} and \cite{TangCapBoundsBinomil}.

\subsection{Gap to the Capacity}
The gap to capacity will play an important role in our analysis especially when we provide a lower bound on the number of points in the support of $P_{X^\star}$. For $n \ge 0$, we define
\begin{equation}
    \mathrm{Gap}(n)\triangleq C(n) - I(X_r; Y_r).  \label{eq:Gap_Definition}
\end{equation}

Combining Theorems~\ref{thm:capacity-lb} and~\ref{thm:capacity-ub}, we arrive at the following bound on the gap.
\begin{prop}\label{prop:capacity_gap}
    \begin{align}
        \mathrm{Gap}(n) &\le r_{{\rm UB}}(n)-r_{{\rm LB}}(n), \text{ for $n\ge 28$} \\
        &\le \frac{17}{\log\!\left(\frac{n\pi}{2e}\right)}, \text{ for $n \ge 444$}. 
    \end{align}
\end{prop}
\begin{IEEEproof}
By comparing \eqref{eq:capacity-lb-asymp} and \eqref{eq:capacity-ub}, we have
    \begin{align}
{\rm Gap}(n) & \le r_{{\rm UB}}(n)-r_{{\rm LB}}(n) \\
&=-\log\!\left(1-2\left(\frac{2e}{n\pi}\right)^{1/4}\right)
+\frac{10}{\log\!\left(\frac{n\pi}{2e}\right)}  
-\frac12\log\!\left(1+\frac{1}{3n}\right)
+\log\!\left(1+\frac{1}{\sqrt{3n+1}}\right)+\frac{1}{n+1}.
\end{align}
Using $-\log(1-x)\le x + x^2$ for $0<x<1/2$, and noting that for sufficiently large $n$ (e.g., $n \ge 444$)
\[
x = 2\left(\frac{2e}{n\pi}\right)^{1/4} < \frac12,
\]
we obtain
\[
-\log\!\left(1-2\left(\frac{2e}{n\pi}\right)^{1/4}\right)
\le
2\left(\frac{2e}{n\pi}\right)^{1/4}
+
4\left(\frac{2e}{n\pi}\right)^{1/2}.
\]
Next, using $\log(1+x)\le x$ for $x\ge0$,
\[
\log\!\left(1+\frac{1}{\sqrt{3n+1}}\right) +\frac{1}{n+1}\le \frac{1}{\sqrt{3n+1}}+\frac{1}{n+1} \leq \frac{1}{\sqrt{n}} + \frac{1}{n}.
\]
Moreover,
\[
-\frac12\log\!\left(1+\frac{1}{3n}\right)
\le
0.
\]
Therefore,
\begin{align}
{\rm Gap}(n)
&\le
\frac{10}{\log\!\left(\frac{n\pi}{2e}\right)}
+
2\left(\frac{2e}{n\pi}\right)^{1/4}
+
4\left(\frac{2e}{n\pi}\right)^{1/2}
+
\left(\frac{1}{n}\right)^{1/2} +\frac{1}{n}.
\end{align}

For all $n \ge 444$, the polynomial-decay terms are dominated by $n^{-1/4}$. Hence
\begin{align}
{\rm Gap}(n)
&\le
\frac{10}{\log\!\left(\frac{n\pi}{2e}\right)}
+
C_1 n^{-1/4} \\
&= \frac{10}{\log\!\left(\frac{n\pi}{2e}\right)}
+
5 n^{-1/4} \label{eq:pick_C1} \\
&\le \frac{17}{\log\!\left(\frac{n\pi}{2e}\right)}, \qquad n \ge 444, \label{eq:bound_Un}
\end{align}
where in~\eqref{eq:pick_C1} we picked $C_1 = 5$.
\end{IEEEproof}

In Section~\ref{sec:num_res}, Fig.~\ref{fig:norm_gaps} shows the normalized gaps between the numerically computed capacity and the upper and lower bounds. The figure suggests that the gap between the upper bound and the  capacity may be further improved.

\section{Properties of the Capacity-Achieving Distributions}
\label{sec:main_I}
In this section we study properties of capacity-achieving distributions.

\subsection{Discreteness}
\label{sec:discretness}
As already mentioned in Sec.~\ref{sec:Lit}, from the Witsenhausen approach we only know that there exists a discrete distribution with at most $n+1$ mass points.  This, however, does not rule out the existence of other CAIDs (\emph{e.g.},  continuous CAIDs).  


We now show that all CAIDs are discrete and provide  a preliminary bound on the support. We recall the definition
\begin{equation}\label{eq:zeros_KKT_recall}
\cA_n =\{x \in [0,1]:\:  i(x;P_{Y^\star}) -C(n)=0 \}.
\end{equation}
\begin{prop}\label{prop:bound_cardinality_cA} $| \cA_n| \le n+1$.
\end{prop}
\begin{proof}
Note that 
\begin{align}
    |\cA_n| & =
    \sfN\lr{ \, [0,1]; \, i(x;P_{Y^\star})-C(n) \, }  \\
    &\le \sfN\lr{ \, [0,1]; \,i''(x;P_{Y^\star}) \, }+2  \label{eq:use_Rolle} \\
    &= \sfN \lr{ \, [0,1]; \, n(1-x)+ x G(x) \, }+2 \label{eq:use_der_inf_den} \\
    &=  \sfN \lr{ \, [0,1]; \, g(x) \, }+2 \label{eq:introduce_g}
    \end{align}
where \eqref{eq:use_Rolle} follows from Rolle's theorem (see~\cite[Lemma~3]{dytso2019capacity}); \eqref{eq:use_der_inf_den} follows from Proposition \ref{prop:derivatives_info_density}; and  in~\eqref{eq:introduce_g}, we have introduced the function
\begin{align}
    g(x) &=  n(1-x)+ x G(x) \\
    &= n(1-x)+ x \sum_{y=0}^{n-2} P_{Y|X}(y|x)
 (n-y)(n-y-1)  \log\left(\frac{(1-\expcnd{X}{Y=y+2})\expcnd{X}{Y=y} }{(1-\expcnd{X}{Y=y+1})\expcnd{X}{Y=y+1}} \right). \label{eq:def_g}
\end{align}
Now note that  $g(x)$ is a polynomial of degree $n-1$. Therefore, it is either zero in which case $\mathrm{N}\lr{ \, [0,1]; \, g(x) \, } =\infty$ or  $\mathrm{N}\lr{ \, [0,1]; \, g(x) \, } \le  n-1$. The first case can clearly be ruled out since coefficient in \eqref{eq:def_g} cannot be zero.  
\end{proof}
The bound in Proposition \ref{prop:bound_cardinality_cA} will be improved in Sec.~\ref{sec:card_bounds}. 
\subsection{Uniqueness of the Optimal Input Distribution} 
\label{sec:uniqness}

In this section, we show and discuss uniqueness of the CAID. To aid our discussion, it is useful to parameterize the mutual information in terms of distributions instead of random variables, that is 
\begin{equation}
    I(P_X; P_{Y|X} ) = I(X;Y). 
\end{equation}
We also let $\cP_{\cX}$ be the set of all distributions over the set $\cX$. In particular, the optimization in \eqref{eq:capacity_opt_problem} can be written as
\begin{equation}
   \max_{P_X \in \cP_{[0,1]}} I(P_X; P_{Y|X} ).
\end{equation}
A typical way to show that there is a unique maximizer  is to show that the mapping $P_X \mapsto I(P_X; P_{Y|X} ) $ over the set $\cP_{[0,1]}$ is \emph{strictly} concave \cite{smith1971information}.  However, due to the fact that the output space of the binomial channel is finite and the input space is uncountable, the mutual information is not strictly concave over $\cP_{[0,1]}$.  For example, when $n=1$ any  distribution symmetric around $x=\frac{1}{2}$ will induce 
\begin{equation}
    P_{Y}(0) = P_{Y}(1) = \frac{1}{2} 
\end{equation}
which is the capacity-achieving output distribution for $n=1$.  Therefore, to show uniqueness of the CAID a new or slightly different argument is needed. 

 We begin by showing the following result. 

\begin{prop}\label{prop:full_rank_A}  Consider a sequence $0\le x_1 < \ldots  < x_{n+1} \le 1$ and define the
matrix $\sfA \in \bbR^{n+1 \times n+1} $ as
\begin{equation}
[\sfA]_{i,k} = P_{Y|X}(i-1|x_k), \quad  i \in [n+1], \, k \in [n+1]. 
\end{equation}
Then,  $\sfA$ is full rank. 
\end{prop}
\begin{proof}
First of all, we argue that considering $x_1=0$ and $x_{n+1}=1$ comes without loosing generality. In fact, in this case
the first and last columns of $\sfA$ are $\bfe_1$ and $\bfe_{n+1}$, respectively, where $\bfe_i$ is a zero vector with a $1$ in the $i$-th position. As a consequence, we have $\det(\sfA) = \det(\tilde{\sfA})$, where
\begin{equation}
 [\tilde{\sfA}]_{i,k} = [\sfA]_{i+1,k+1}, \quad i \in [n-1], \, k \in [n-1].   
\end{equation}  
     Next, note that we can rewrite the binomial law as $P_{Y|X}(y|x) = \binom{n}{y}(1+\rme^{t})^{-n}\rme^{ty}$, where $x = \frac{\rme^t}{1+\rme^t}$. The matrix $\sfB$ with $[\sfB]_{y,k} = \rme^{t_k y}$ and $y \in [n-1]$ is a Vandermonde matrix, which is full rank since the $t_k$'s are all distinct \cite{golub2013matrix}. Thanks to the multilinear property of the determinant, we can write that 
    \begin{equation}
        \det(\tilde{\sfA}) = \det(\sfB) \prod_{y=1}^{n-1} \binom{n}{y} \prod_{k=2}^{n} (1+\rme^{t_k})^{-n} >0
    \end{equation}
    where in the last step is due to $\det(\sfB)>0$ and to the positivity of the products. As a consequence, $\sfA$ is a full rank matrix.
\end{proof}

With  the aid of  Proposition~\ref{prop:full_rank_A},  we show the following result. 
\begin{theorem}\label{thm:stric:_concavity}
    Let $\cX \subset [0,1]$ be a  set of cardinality $n+1$.  Then, $P_X \mapsto I(P_X; P_{Y|X})$ is \emph{strictly} concave over $\cP_{\cX}$. 
\end{theorem}
\begin{proof} 
    Let $P_X, Q_X \in \cP_{\cX}$, and let $P_X^\epsilon = (1-\epsilon) P_X +\epsilon Q_X$ for $\epsilon \in (0,1)$, which is also in $\cP_{\cX}$.  Moreover, 
    let $P_X \to P_{Y|X} \to P_Y$,  $Q_X \to P_{Y|X} \to Q_Y$ and $P_X^\epsilon \to P_{Y|X} \to P_Y^\epsilon$. 
    Then, first note that
    \begin{align}
    &I( P_X^\epsilon; P_{Y|X}) - (1-\epsilon)  I(  P_X ; P_{Y|X} ) - \epsilon I(   Q_X; P_{Y|X} ) \notag\\
    & \quad =     D (P_{Y|X} \| P_{Y}^\epsilon| P_X^\epsilon)  -(1-\epsilon)  D (P_{Y|X} \| P_{Y}| P_X)  - \epsilon  D (P_{Y|X} \| Q_{Y}| Q_X) \\
     & \quad = (1-\epsilon) D(P_{Y}\| P_{Y}^\epsilon) +\epsilon  D(Q_{Y}\| P_{Y}^\epsilon). \label{eq:concavity_decomposition}
    \end{align}

    We now show that every $ P_X \in \cP_{\cX}$ induces a distinct output distribution (i.e., $P_X \to P_{Y|X} \to P_Y$ is an injective mapping), which implies that \eqref{eq:concavity_decomposition} is strictly positive and, therefore, the mutual information is strictly concave.  Define $\bfp_X = [  P_X(x_1), \ldots, P_X(x_{n+1})]$, with $x_k \in \cX$ and $\bfp_Y = [P_Y(0), \ldots, P_Y(n) ]$. Then, the mapping $P_X \to P_{Y|X} \to P_{Y}$ can be written as the following system of linear equations:
    \begin{equation}
    \sfA \bfp_X = \bfp_Y \label{eq:linear_system}
    \end{equation}
    where the matrix $\sfA \in \bbR^{n+1 \times n+1}$ is such that
\begin{equation}
    [\sfA]_{i,k} = P_{Y|X}(i-1|x_k), \,  i \in [n+1], \, x_k \in \cX. 
\end{equation}
    From Proposition~\ref{prop:full_rank_A}, we know that $\sfA$ is full rank for any $\cX$ of cardinality $n+1$. Therefore, from standard linear algebra result, it follows that the mapping in \eqref{eq:linear_system} is injective (i.e., every $\bfp_X$ induces a distinct $\bfp_Y$).  Therefore, we conclude that \eqref{eq:concavity_decomposition} is positive and the mutual information is strictly concave.  
\end{proof}

Note that,  by Proposition~\ref{prop:bound_cardinality_cA}, $\cA_n$ has cardinality at most $n+1$. Therefore, Theorem~\ref{thm:stric:_concavity} implies the following corollary.  

\begin{corollary}
    $P_X \mapsto I(P_X; P_{Y|X})$ is strictly concave over $\cP_{\cA_n}$. Consequently, $P_{X^\star}$ is unique.   
\end{corollary}

\subsection{Symmetry} 
\label{sec:symmetry}

The binomial channel exhibits the following symmetry
\begin{equation} \label{eq:channel_symmetry}
    P_{Y|X}(y|x) = P_{Y|X}(n-y|1-x), \quad x \in [0,1], \, y \in \{0\}\cup [n]. 
\end{equation}
which  immediately leads to the following result.  

\begin{prop}\label{prop:input_symmetry}
If $X^\star$ is capacity-achieving, then  $X^\star \stackrel{d}{=}1-X^\star$.\footnote{Here $\stackrel{d}{=}$ denotes equality in distribution.} 
\end{prop}

\subsection{On the Location of Support Points}
\label{sec:location}

 Following the same lines of~\cite[Sec.~V]{abou2001capacity} we have that: 
\begin{prop}\label{prop:0_is_optimal}
    Let $P_{X^\star}$ be a CAID. Then, $\{0,1\} \subseteq \supp(P_{X^\star})$.
\end{prop}
\begin{proof}
    By using symmetry (Proposition~\ref{prop:input_symmetry}), we can just prove the result for the point at $x=0$. Let $0\le x_0 < x_1  < \ldots < x_N \le 1$ be the support points of $P_X$. Suppose that $x_0>0$. Then, we have that
    \begin{align}
        &\frac{\del}{\del x_0}I(X;Y) = P_X(x_0) \frac{\del}{\del x_0} \expcnd{\log\frac{P_{Y|X}(Y|x_0)}{P_Y(Y)}}{X=x_0} \\
        &= \frac{P_X(x_0)}{x_0(1-x_0)} \expcnd{(Y-nx_0)\log\frac{P_{Y|X}(Y|x_0)}{P_Y(Y)}}{X=x_0}. \label{eq:der_mut_inf}
    \end{align}
    Next, we prove that the function $f: y\mapsto \log\frac{P_{Y|X}(y|x_0)}{P_Y(y)}$ is decreasing. Note that
    \begin{align}
        &\frac{P_Y(y)}{P_{Y|X}(y|x_0)} = P_{X}(x_0) + \sum_{i=1}^N P_{X}(x_i) \left(\frac{x_i}{x_0}\right)^y \left(\frac{1-x_i}{1-x_0}\right)^{n-y} = P_{X}(x_0) + \sum_{i=1}^N P_{X}(x_i) \left(\frac{\frac{1}{x_0}-1}{\frac{1}{x_i}-1}\right)^y \left(\frac{1-x_i}{1-x_0}\right)^{n}
    \end{align}
    is an increasing function of $y$, since $x_0 < x_i$ for $i\ge1$. As a consequence, the function $f$ is decreasing. By noting that $\expcnd{Y}{X=x_0} = n x_0$ and by applying \cite[Lemma 1]{abou2001capacity} to \eqref{eq:der_mut_inf}, we get that $\frac{\del}{\del x_0}I(X;Y)<0$ for all $0<x_0<x_1$. This implies that $x_0 = 0 \in \supp(P_{X^\star})$. 
\end{proof}

An important consequence of Proposition~\ref{prop:0_is_optimal} is given next. 
\begin{corollary}\label{cor:rel_cap_probY}
The channel capacity is equal to
\begin{equation}
    C(n) = \log\frac{1}{P_{Y^\star}(0)} = \log\frac{1}{P_{Y^\star}(n)}.
\end{equation}
\end{corollary}
\begin{proof}
    Thanks to Proposition \ref{prop:0_is_optimal}, we know that $0 \in \supp(P_{X^\star})$. By using the KKT condition \eqref{eq:KKT_equality}, we can write
    \begin{equation*}
        C(n) = i(0;P_{Y^\star}) = \sum_{y=0}^n \binom{n}{y} 0^y \log\frac{\binom{n}{y} 0^y}{P_{Y^\star}(y)} = \log\frac{1}{P_{Y^\star}(0)}.
    \end{equation*}
    By symmetry, we can argue that $P_{Y^\star}(0) = P_{Y^\star}(n)$. 
\end{proof}

We next show that there is at most one support point in the interval $\left(0,\frac{1}{n} \right]$ and, by symmetry, at most one point in $\left[1-\frac{1}{n},1 \right)$.  The proof technique we use was developed in \cite{mceliece1979practical} in the context of Poisson noise channels. 
\begin{prop} \label{prop:loc_info}
    For all $n\ge 1$, we have
    \begin{align}
        \left| \supp(P_{X^\star}) \cap \left(0,\frac{1}{n}\right] \right|&\le 1, \label{eq:limit_supp_1} \\
        \left| \supp(P_{X^\star}) \cap \left[1-\frac{1}{n},1\right) \right|&\le 1.\label{eq:limit_supp_2}
    \end{align}
\end{prop}
\begin{proof}
    For $n=1$, the claim follows from the fact that there are only two mass points at $\{0,1\}$. Next, we consider the case $n\ge 2$. 
    
    From expression~\eqref{eq:i_doubleprime} of Proposition~\ref{prop:der_info_den}, we have that 
    \begin{align}
        &x(1-x)i''(x;P_{Y^\star}) \\
        &= n+\frac{x}{1-x}  \expcnd{(n-Y)(n-Y-1)\log\frac{\expcnd{X^\star}{Y=Y}}{\expcnd{1-X^\star}{Y=Y+1}}\frac{\expcnd{1-X^\star}{Y=Y+2}}{\expcnd{X^\star}{Y=Y+1}}}{X=x} \\
        &= n+\sum_{y=0}^{n-2} \binom{n}{y} x^{y+1}(1-x)^{n-y-1} (n-y)(n-y-1)\log\frac{\expcnd{X^\star}{Y=y}}{\expcnd{X^\star}{Y=y+1}}\frac{\expcnd{X^\star}{Y=n-y-2}}{\expcnd{X^\star}{Y=n-y-1}} \label{eq:sec_der_inf_den_limited_supp},
    \end{align} 
    where in the last step we have exploited the channel symmetry~\eqref{eq:channel_symmetry}.
    Since $y\mapsto \expcnd{X^\star}{Y=y}$ is a non-decreasing function (see Lemma~\ref{lem:cond_exp_non_increasing}), all the terms in the summation of~\eqref{eq:sec_der_inf_den_limited_supp} are nonpositive. Moreover, the functions $x\mapsto x^{y+1}(1-x)^{n-y-1}$ for $y=0, \ldots, n-2$ are increasing for $x \le \frac{1}{n}$. As a consequence, the function $x\mapsto g(x)=x(1-x)i''(x;P_{Y^\star})$ is non-increasing for $x \in \left(0,\frac{1}{n}\right]$. Since $g(0)=n\ge 2$, the function $g$ has at most one zero in the interval $\left(0,\frac{1}{n}\right]$. Then, $i''$ has at most one zero in the interval $\left(0,\frac{1}{n}\right]$, hence $i(x;P_{Y^\star})-C(n)$ has at most one zero crossing in the interval $\left(0,\frac{1}{n}\right]$. This proves~\eqref{eq:limit_supp_1}. To prove~\eqref{eq:limit_supp_2}, we can use symmetry $X^\star \stackrel{d}{=} 1-X^\star$ from Proposition~\ref{prop:input_symmetry}.
\end{proof}

\subsection{Bounds on the Probabilities}
\label{sec:probabilities}

We begin by recalling that for $P_X \to P_{Y|X} \to P_Y$ and $Q_X  \to P_{Y|X} \to Q_{Y}$, we have that
\begin{equation}
    \kl{P_X}{Q_X} = \kl {P_Y}{Q_Y} + D (P_{X|Y} \| Q_{X|Y} |P_Y)  ,\label{eq:KL_identity}
\end{equation}
where the conditional relative entropy is defined as
\begin{equation}
    D (P_{X|Y} \| Q_{X|Y} |P_Y) = \sum_{y=0}^n P_Y(y) \kl{P_{X|Y}(\cdot|y)}{Q_{X|Y}(\cdot|y)}
\end{equation}

The key to finding bounds on the probabilities is the following lemma. 

\begin{lemma}\label{lem:expression_for_prob} For $x^\star \in \supp(P_{X^\star})$
    \begin{equation}
    P_{X^\star}(x^\star) = \rme^{ -C(n) -  \cD(x^\star) }, 
\end{equation}
where $\cD(x^\star) =  D  \left(\delta_{x^\star} \| P_{X^\star|Y} |P_{Y|X }(\cdot|x^\star) \right)$ and $\delta_{x^\star}$ is the single point mass distribution on $x^\star$. 
\end{lemma}
\begin{proof}
     Using the  equality condition  in the KKT \eqref{eq:KKT_equality}, we have that for $x^{\star} \in \supp(P_{X^{\star}})$
\begin{align}
   C(n) &= \kl{P_{Y|X}(\cdot|x^\star)}{P_{Y^\star}}\\
    &=  \kl{ P_{Y_{x^\star} }  }{P_{Y^\star}} \label{eq:Def_output_P_Y}\\
    &= \kl{\delta_{x^\star} }{  P_{X^\star} } -    D (\delta_{x^\star} \| P_{X^\star|Y} |P_{Y_{x^\star} }) \label{eq:apply_KL_identity}\\
    &= \log \frac{1}{P_{X^\star} (x^\star)} -    D (\delta_{x^\star} \| P_{X^\star|Y} |P_{Y_{x^\star} }), \label{eq:last_step_P(x)}
\end{align}
where  \eqref{eq:Def_output_P_Y} follows by defining $\delta_{x^\star} \to P_{Y|X} \to P_{Y_{x^\star}}$;  and \eqref{eq:apply_KL_identity} follows by using \eqref{eq:KL_identity}.

By rearranging~\eqref{eq:last_step_P(x)}, and recognizing that $P_{Y_{x^\star}}(\cdot) = P_{Y|X}(\cdot |x ^\star)$,  we arrive at: for $x^{\star} \in \supp(P_{X^{\star}})$ 
\begin{equation}
    P_{X^\star}(x^\star) = \rme^{ -C(n) -  D  \left(\delta_{x^\star} \| P_{X^\star|Y} |P_{Y|X }(\cdot|x^\star) \right) }. 
\end{equation}
\end{proof}

The term $\cD(x^\star)$ measures how on average the $P_{X^\star|Y}$ is close to a point measure. We refer to $\cD(x^\star)$ as the \emph{crest-factor}.\footnote{In signal processing, the crest-factor  measures how peaky the waveform is. Specifically, it compares the peak amplitude of a waveform relative to its average value.}  

From Lemma~\ref{lem:expression_for_prob}, by  using $\cD(x^\star) \ge 0$, which follows from the nonnegativity of   relative-entropy, we immediately arrive at the following bound: 
\begin{equation}
    P_{X^\star}(x^\star) \le  \rme^{ -C(n) }, \qquad  x^\star \in \supp(P_{X^\star}). \label{eq:simple_prob_bound}
\end{equation}
The bound in \eqref{eq:simple_prob_bound} might appear ineffective due to the fact that the capacity is unknown.  However, note that for any $\tilde{X}$, from the definition of the capacity we have that 
\begin{equation}
    P_{X^\star}(x^\star) \le  \rme^{ - I(\tilde{X}; \tilde{Y}) },  \qquad  x^\star \in \supp(P_{X^\star}),
\end{equation}
which implies that any good guess results in an upper bound.

The next result improves upon the bound in \eqref{eq:simple_prob_bound}.

\begin{prop}
 \label{prop:crest_factor_bounds}
\text{ }
\begin{itemize}
    \item First Bound: for $x \in \supp(P_{X^\star}) \setminus \{0,1\} $
    \begin{equation}\label{eq:second_bound_factor_crest}
      \cD(x) \ge  \frac{(1-x)^n \log(1-x)^n+x^n \log(x^n)}{(1-x)^n+x^n-1}.
    \end{equation}
    \item Second Bound: for $x\in \supp(P_{X^\star}) \setminus\{\frac{1}{2}\}$
    \begin{equation}\label{eq:first_bound_factor_crest}  
    \cD(x) \ge \bbE \left[ \log \left( 1 + \left( \frac{x}{1-x} \right)^{n-2Y} \right) | X =x  \right].
    \end{equation}
\end{itemize}
    
\end{prop}
\begin{proof}
We begin by noticing that 
\begin{equation}
-\cD(x) = \expcnd{\log(P_{X^\star|Y}(x|Y))}{X=x} . \label{eq:Minus_D}
\end{equation}
 To show the first bound notice that for $x\in (0,1)$ 
     \begin{align}
     \expcnd{\log(P_{X^\star|Y}(x|Y))}{X=x} & \le  \expcnd{(\mathbbm{1}(Y=0)+\mathbbm{1}(Y=n))\log(P_{X^\star|Y}(x|Y))}{X=x} \\
         &= P_{Y|X}(0|x) \log\frac{P_{Y|X}(0|x) P_{X^\star}(x)}{P_{Y^\star}(0)} + P_{Y|X}(n|x) \log\frac{P_{Y|X}(n|x) P_{X^\star}(x)}{P_{Y^\star}(n)} \\
         &= (1-x)^n \left(\log((1-x)^n P_{X^\star}(x)) +C(n)\right) + x^n\left( \log(x^n P_{X^\star}(x))+C(n) \right) \label{eq:use_rel_capacity_probY}\\
         &= \left((1-x)^n+x^n\right)\left(\log(P_{X^\star}(x)) +C(n)\right) + (1-x)^n \log(1-x)^n+x^n \log(x^n) \label{eq:bound_exp_log_P_X_given_Y}\\
         &= - \left((1-x)^n+x^n\right) \cD(x) + (1-x)^n \log(1-x)^n+x^n \log(x^n), \label{eq:using_Def_0f_cD}
     \end{align}
     where in~\eqref{eq:use_rel_capacity_probY} we used Corollary~\ref{cor:rel_cap_probY}; and where in \eqref{eq:using_Def_0f_cD} we have used Lemma~\ref{lem:expression_for_prob}. Combining~\eqref{eq:using_Def_0f_cD} with \eqref{eq:Minus_D} we arrive at the desired first bound.

    To show the second bound note that by symmetry of the optimal distribution,   we have that, for $x \neq 1/2$,
\begin{align}
P_{X^*|Y}(x|y) 
&= \frac{P_{X^\star}(x) P_{Y|X}(y|x) }{P_{Y^\star}(y)}\\
&\le  \frac{P_{X^\star}(x) P_{Y|X}(y|x) }{P_{X^\star}(x)  P_{Y|X}(y|x) + P_{X^\star}(1-x)  P_{Y|X}(y|1-x) }\\
&=  \frac{P_{Y|X}(y|x) }{  P_{Y|X}(y|x) +  P_{Y|X}(y|1-x) }\\
&= \frac{1}{1 + \left( \frac{x}{1-x} \right)^{n-2y}}.
\end{align}

 Then, we get
\begin{align}
\expcnd{\log(P_{X^\star|Y}(x|Y))}{X=x} 
 &\le \bbE \left[ \log \frac{1}{1 + \left( \frac{x}{1-x} \right)^{n-2Y}} |X =x  \right]  \label{eq:Jensen_crest}
\end{align}

This concludes the proof. 
\end{proof}

\begin{figure}
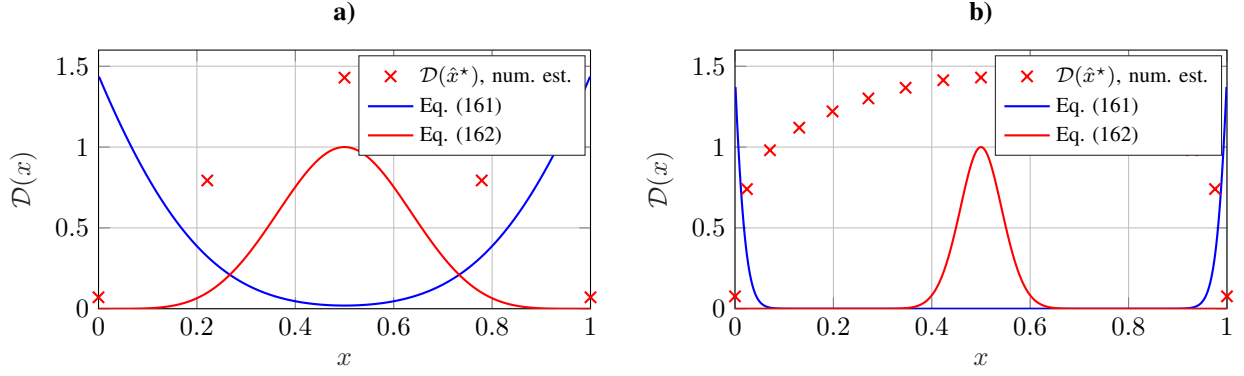

    \centering
    \begin{minipage}[t]{0.49\columnwidth}
    \centering
    \input{IEEE_TIT_submission_correction/Figures/crest_factor_bounds_v2_n10.tex}
    \end{minipage}
\hfill
\begin{minipage}[t]{0.49\columnwidth}
\centering
\input{IEEE_TIT_submission_correction/Figures/crest_factor_bounds_v2_n100.tex}
\end{minipage}
    \caption{Comparison of numerical estimates $\cD(\hat{x}^{\star})$ and bounds on the crest-factor $\cD(x)$ reported in Proposition~\ref{prop:crest_factor_bounds} for \textbf{a)} $n=10$ and \textbf{b)} $n=100$.}
    \label{fig:crest_factor_bound}
\end{figure}

 In Fig.~\ref{fig:crest_factor_bound}, for the cases $n=10$ and $n=100$ we show a comparison between the bounds on $\cD(x)$ given in Proposition~\ref{prop:crest_factor_bounds} and numerical estimates $\cD(\hat{x}^{\star})$, derived from the estimates of $P_{X^{\star}}$ in Sec.~\ref{sec:num_res}. Upper bounds on $\cD(x)$ so far have been elusive.

\begin{table*}
    \centering
    \begin{tabular}{|c|c|c|c|c|}
    \Xhline{2pt}
    $n$ & $C(n)$ & ${\cal X}\equiv\supp(P_{X^\star})$   & $\{P_{X^\star}(x),\, x \in {\cal X} \}$ & $\{P_{Y^\star}(y), \, y\in\{0, 1,\ldots, n\}  \}$ \\
    \Xhline{1pt}
    $1$ & $\log(2)$ & $\left\{0,1 \right\}$  & $\left\{\frac{1}{2},\frac{1}{2}\right\}$     &  $\left\{\frac{1}{2},\frac{1}{2}\right\}$\\
    \hline
    $2$ &  $\log\left(\frac{17}{8}\right)$ & $\left\{0, \frac{1}{2} ,1\right\}$   & $\left\{ \frac{15}{34},  \frac{2}{17},\frac{15}{34}\right\}$  &  $\left\{ \frac{8}{17}, 
 \frac{1}{17},\frac{8}{17} \right\} $ \\
    \hline
    $3$&  $\log\left(\frac{19}{8}\right)$ & $\left\{0, \frac{1}{2} ,1\right\}$ & $\left\{  \frac{15}{38},  \frac{4}{19}, \frac{15}{38}\right\}$ & $\left\{ \frac{8}{19}, 
 \frac{3}{38},\frac{3}{38},\frac{8}{19} \right\}$ \\
   
    \Xhline{1pt}
\end{tabular}
    \caption{Capacity and capacity-achieving distributions.  }
    \label{tab:CapacityExpression}
    
\end{table*}
The bounds in \eqref{eq:simple_prob_bound} also have a  useful generalization bound probability for any measurable set.  To show this, we need the following bound for the mixture of distributions, which is an easy consequence of the log-sum inequality \cite{cover2006elements}: given a mixture 
\begin{equation}
Q = \sum_i w_i Q_i,  \quad  \,  \sum_i w_i=1, \quad \, w_i \ge 0,
\end{equation}
then for  any distribution $P$
\begin{equation}
    \kl{P}{Q} \le \kl{P}{Q_i} +\log \frac{1}{w_i}. \label{eq:mixture_bound}
\end{equation}
This result can now be used to show the following generalization of the bound in \eqref{eq:simple_prob_bound}.
\begin{prop}\label{prop:bounds_on_Sets}
Let $\cA\subseteq[0,1]$ be measurable. If
$P_{X^\star}(\cA)>0$, then for every
\begin{equation}
    x_{\cA}
    \in
    \cA\cap\supp(P_{X^\star}),
\end{equation}
we have
\begin{equation}
    P_{X^\star}(\cA)
    \leq
    B_n(\cA;x_{\cA})\rme^{-C(n)},
    \label{eq:bound_on_measured_sets}
\end{equation}
where
\begin{align}
    B_n(\cA;x_{\cA})
    &=
    \exp\left(
        \sup_{x\in\cA}
        \kl{
            P_{Y|X}(\cdot|x_{\cA})}{   
            P_{Y|X}(\cdot|x)
        }
    \right)
    \label{eq:bound_first_step}\\
    &=
    \exp\left(
        n\sup_{x\in\cA}
        d(x_{\cA}\|x)
    \right),
    \label{eq:binomial_set_bound}
\end{align}
 and $d( \cdot \, \| \, \cdot)$ is the binary relative entropy.
\end{prop}
\begin{IEEEproof}
Assume that $P_{X^\star}(\cA)>0$ , otherwise \eqref{eq:bound_on_measured_sets} holds trivially. Fix $  x_{\cA}
    \in
    \cA\cap\supp(P_{X^\star})$. 
Let
\begin{equation}
    R_{\cA}
    =
    P_{Y^\star|X^\star\in\cA}.
\end{equation}
By conditioning on the event $\{X^\star\in\cA\}$, the
capacity-achieving output distribution can be written as
\begin{align}
    P_{Y^\star}
    &=
    P_{X^\star}(\cA)R_{\cA}+
    \bigl(1-P_{X^\star}(\cA)\bigr)
    P_{Y^\star|X^\star\notin\cA}.
\end{align}
Starting with the KKT equality~\eqref{eq:KKT_equality}, we have
\begin{align}
    C(n)
    &=
    \kl{
        P_{Y|X}(\cdot|x_{\cA})}
        {
        P_{Y^\star}}\\
    &\leq
    \kl{
        P_{Y|X}(\cdot|x_{\cA})}
        { R_{\cA}}
    +
    \log\frac{1}{P_{X^\star}(\cA)}
    \label{eq:applying_mixture_bound}\\
    &\leq
    \int_{\cA}
    \kl{
        P_{Y|X}(\cdot|x_{\cA})
        }
        {
        P_{Y|X}(\cdot|x)}
    \rmd P_{X^\star|X^\star\in\cA}(x)
    +
    \log\frac{1}{P_{X^\star}(\cA)}
    \label{eq:use_convexity_KL}\\
    &\leq
    \sup_{x\in\cA}
    \kl{
        P_{Y|X}(\cdot|x_{\cA})}
        {
        P_{Y|X}(\cdot|x)
    }
    +
    \log\frac{1}{P_{X^\star}(\cA)}\\
    &=
    n\sup_{x\in\cA}d(x_{\cA}\|x)
    +
    \log\frac{1}{P_{X^\star}(\cA)}.
\end{align}
Here, \eqref{eq:applying_mixture_bound} follows from
\eqref{eq:mixture_bound}, and
\eqref{eq:use_convexity_KL} follows from convexity of relative
entropy in its second argument. Rearranging concludes the
proof.
\end{IEEEproof}

For the bound in \eqref{eq:bound_on_measured_sets} to be informative, the exponent in \eqref{eq:binomial_set_bound} must remain uniformly bounded. We next construct a collection of sets $\cA$ for which this condition holds.

\begin{lemma}\label{lem:reference_partition}
Let $ N
    =
    2\left\lceil
        \frac{5\pi}{4}\sqrt{n}
    \right\rceil$
and define
\begin{equation}
    a_j
    =
    F_r^{-1}\left(\frac{j}{N}\right)
    =
    \sin^2\left(\frac{j\pi}{2N}\right),
    \qquad
    j=0,\ldots,N.
\end{equation}
Consider the partition of $[0,1]$ given by
\begin{align}
    \cA_j
    &=
    [a_j,a_{j+1}),
    \qquad
    j=0,\ldots,N-2,\\
    \cA_{N-1}
    &=
    [a_{N-1},1].
\end{align}
Then,
\begin{equation}
    P_{X^\star}(\cA_j)
    \leq
    \exp\left(
        \frac{8}{25}-C(n)
    \right),
    \qquad
    j=0,\ldots,N-1.
    \label{eq:uniform_cell_mass_bound}
\end{equation}
\end{lemma}

\begin{IEEEproof}
By construction,
\begin{equation}
    F_r(a_j)
    =
    \frac{j}{N},
    \qquad
    j=0,\ldots,N.
\end{equation}
Therefore, for every $x,x'\in\cA_j$,
\begin{equation}
    \left|
        F_r(x)-F_r(x')
    \right|
    \leq
    \frac{1}{N}.
    \label{eq:Fr_cell_diameter}
\end{equation}
Moreover, since $N$ is even,
\begin{equation}
    a_{N/2}
    =
    F_r^{-1}\left(\frac12\right)
    =
    \frac12,
\end{equation}
and hence every cell is contained in either $[0,1/2]$ or
$[1/2,1]$.

Fix $j\in\{0,\ldots,N-1\}$. If
$P_{X^\star}(\cA_j)=0$, the result holds trivially. Otherwise,
for the first and last cells, choose
\begin{equation}
    x_0=0,
    \qquad
    x_{N-1}=1,
\end{equation}
which is possible since
$\{0,1\}\subseteq\supp(P_{X^\star})$ by Proposition~\ref{prop:0_is_optimal}. For every other cell,
choose any
\begin{equation}
    x_j
    \in
    \cA_j\cap\supp(P_{X^\star}).
\end{equation}

Consider first
\begin{equation}
    j\in
    \left\{
        0,\ldots,\frac{N}{2}-1
    \right\},
\end{equation}
so that $\cA_j\subseteq[0,1/2]$. For $x\in\cA_j$, define
\begin{equation}
    \theta_j
    =
    \arcsin\sqrt{x_j},
    \qquad
    \theta
    =
    \arcsin\sqrt{x}.
\end{equation}
Then $\theta_j,\theta\in[0,\pi/4]$. Applying
\cite[Lem.~3]{rooij2009learning} with
$\phi_1=\phi_3=\theta_j$ and $\phi_2=\theta$ gives
\begin{equation}
    d(x_j\|x)
    \leq
    \begin{cases}
        4(\theta-\theta_j)^2,
        & \theta_j\leq\theta,\\[1mm]
        4(\theta_j-\theta)^2
        \dfrac{\theta_j}{\theta},
        & \theta_j>\theta.
    \end{cases}
    \label{eq:rooij_cell_bound}
\end{equation}
When $\theta=0$, we necessarily have $j=0$ and
$x_j=x=0$, and hence $d(x_j\|x)=0$. Otherwise, if
$\theta_j>\theta$, then $j\geq1$ and, since $\theta_j$ and
$\theta$ belong to the same angular cell,
\begin{align}
    \frac{\theta_j}{\theta}
    &\leq
    \frac{(j+1)\pi/(2N)}
         {j\pi/(2N)}=
    \frac{j+1}{j}\leq
    2.
    \label{eq:angle_ratio_bound}
\end{align}
It follows that
\begin{align}
    d(x_j\|x)
    &\leq
    8(\theta_j-\theta)^2
    \label{eq:using_rooij_bound}\\
    &=
    2\pi^2
    \left(
        F_r(x_j)-F_r(x)
    \right)^2
    \label{eq:using_Fr_coordinate}\\
    &\leq
    \frac{2\pi^2}{N^2},
    \qquad
    x\in\cA_j,
    \label{eq:using_Fr_cell_diameter}
\end{align}
where \eqref{eq:using_rooij_bound} follows from
\eqref{eq:rooij_cell_bound} and
\eqref{eq:angle_ratio_bound};
\eqref{eq:using_Fr_coordinate} follows from
\begin{equation}
    F_r(x)
    =
    \frac{2}{\pi}\arcsin\sqrt{x};
\end{equation}
and \eqref{eq:using_Fr_cell_diameter} follows from
\eqref{eq:Fr_cell_diameter}. Now suppose that
\begin{equation}
    j\in
    \left\{
        \frac{N}{2},\ldots,N-1
    \right\},
\end{equation}
so that $\cA_j\subseteq[1/2,1]$. Since
\begin{equation}
    d(x_j\|x)
    =
    d(1-x_j\|1-x)
\end{equation}
and
\begin{equation}
    F_r(1-x)
    =
    1-F_r(x),
\end{equation}
the preceding argument applied to $1-x_j$ and $1-x$ gives
\begin{align}
    d(x_j\|x)
    &\leq
    2\pi^2
    \left(
        F_r(1-x_j)-F_r(1-x)
    \right)^2\\
    &=
    2\pi^2
    \left(
        F_r(x_j)-F_r(x)
    \right)^2\\
    &\leq
    \frac{2\pi^2}{N^2},
    \qquad
    x\in\cA_j.
    \label{eq:right_cell_Fr_bound}
\end{align}

Consequently, for every $j\in\{0,\ldots,N-1\}$ with
$P_{X^\star}(\cA_j)>0$,
\begin{align}
    n\sup_{x\in\cA_j}d(x_j\|x)
    &\leq
    \frac{2\pi^2n}{N^2}\leq
    \frac{8}{25},
    \label{eq:cell_divergence_bound}
\end{align}
where the last inequality follows from
\begin{equation}
    N
    \geq
    \frac{5\pi}{2}\sqrt{n}.
\end{equation}
Finally, applying Proposition~\ref{prop:bounds_on_Sets} with
$\cA=\cA_j$ and $x_{\cA}=x_j$ gives
\begin{align}
    P_{X^\star}(\cA_j)
    &\leq
    \exp\left(
        n\sup_{x\in\cA_j}d(x_j\|x)-C(n)
    \right)\leq
    \exp\left(
        \frac{8}{25}-C(n)
    \right).
\end{align}
This concludes the proof.
\end{IEEEproof}


\subsection{Bounds on the Support Cardinality} \label{sec:card_bounds}
We now provide upper and lower bounds on the cardinality of the support of $P_{X^\star}$.  We start with the following \emph{exact} formula for the number of support points. 
\begin{prop}  \label{prop:ExactNumSup}
For $n \ge 1$
    \begin{equation}
| \supp(P_{X^\star})| = \frac{\rme^{ C(n) }}{    \bbE \left[ \rme^{  - \cD(U^*)  } 
 \right]},
\end{equation}
where $U^*$ is uniformly distributed on $\supp(P_{X^\star})$.
\end{prop}
\begin{proof}
   Starting with Lemma~\ref{lem:expression_for_prob} and  summing over $x^{\star} \in \supp(P_{X^{\star}})$, we arrive at
\begin{equation}
1 =   \rme^{ -C(n) } \sum_{ x^{\star} \in \supp(P_{X^{\star}}) }  \rme^{  -  \cD(x^\star) }. \label{eq:Intermideiate_step}
\end{equation}
Dividing both sides of  \eqref{eq:Intermideiate_step} by $| \supp(P_{X^\star})|$ and rearranging,  we arrive at the desired result.  
\end{proof}

From Proposition~\ref{prop:ExactNumSup} and  nonnegativity of $\cD$, we arrive at 
\begin{equation}
    | \supp(P_{X^\star})| \ge \rme^{C(n)} =\Theta( \sqrt{n}) \label{eq:sqrt(n) lower bound}
\end{equation}
where the order of the lower bound follows from the fact that $C(n)$ scales as $\frac{1}{2} \log(n)$ as will be shown in Sec.~\ref{sec:cap_bounds}. We will improve the order of this lower bound in Sec.~\ref{sec:lower_bound_support}. 

We now  move on to showing upper bounds. We already have demonstrated a bound of order $n+1$ in Proposition \ref{prop:bound_cardinality_cA}.  We now improve this upper bound by a factor of two. 
\begin{theorem} \label{thm:n/2 bound}
    For $n \ge 1$
    \begin{align}
        |\supp{(P_{X^\star})}| &\le  2 + \left\lfloor \frac{1}{2}\sfN\left((0,1); \, i''(x;P_{Y^\star}) \right) \right\rfloor \label{eq:i" bound} \\
        & \le 2+\left\lfloor\frac{n}{2} \right\rfloor. \label{eq: n/2 bound}
    \end{align}
\end{theorem}
\begin{proof}
    First of all, note that by Proposition~\ref{prop:0_is_optimal} we know that the function $i(\cdot; P_{Y^\star})$ starts with a local maximum at $x^\star =0$, and from Proposition~\ref{prop:derivatives_info_density} we know that $i''(x; P_{Y^\star})>0$  for $x\to 0$. Now, by continuity of $i(\cdot; P_{Y^\star})$,  if $i(\cdot; P_{Y^\star})$ changes concavity $k$ times, then it has at most $2+\left\lfloor \frac{k}{2} \right\rfloor$ local maxima.
    Moreover, from the KKT conditions we know that all the zeros of $i(\cdot;P_{Y^\star})-C(n)$ are local maxima. 
    
 Then, we can write 
        \begin{align}
        |\supp{(P_{X^\star})}| 
         & \le \sfN\lr{ \, [0,1]; \, i(\cdot;P_{Y^\star})-C(n) \, } \le 2 + \left\lfloor \frac{1}{2}\sfN\left((0,1); \, i''(\cdot;P_{Y^\star}) \right) \right\rfloor \le 2+\left\lfloor\frac{n}{2} \right\rfloor, \label{eq:bound_zeros_i''} 
    \end{align}
    where the last bound is from the fact that $x\mapsto x(x-1)i''(x;P_{Y^\star})$ is a polynomial of degree $n$ ( Proposition \ref{prop:derivatives_info_density}). 
\end{proof}

The proof of Theorem~\ref{thm:n/2 bound} does not rely on the uniqueness of $P_{X^\star}$. Therefore, it improves on the Witsenhausen bound by a factor of two. Furthermore, the key part of the proof leading to \eqref{eq:i" bound} is independent of the fact that the channel is binomial: Indeed, this fact is only used in \eqref{eq: n/2 bound}.   Consequently, we posit that this bound may prove more beneficial for channels where it is feasible to establish bounds on the number of zeros in~$i''(x; P_{X^\star})$.

\subsection{Exact Capacity for $n \le 3$}

The exact capacity can be computed by first making a guess of the capacity-achieving distribution according to the properties outlined in Sec.~\ref{sec:main_I}. Then, this guess can be checked against the sufficient and necessary KKT conditions in Lemma~\ref{lem:KKT}. These, somewhat tedious, computations  are performed in Appendix~\ref{sec:computations_of_cap_exact}  and Table~\ref{tab:CapacityExpression} displays the results.

\subsection{Asymptotically Optimal Output Distribution}
In Proposition~\ref{prop:capacity_gap}, we have shown that the input $X_r$ is asymptotically optimal from the rate perspective. We now show that the output distribution induced by $X_r$ also approaches the optimal output distribution.

The next two results make this connection precise.
\begin{theorem}\label{thm:optimal-output}
Let $P_{X^\star} \to P_{Y|X} \to P_{Y^\star}$ and $P_{X_r} \to P_{Y|X} \to P_{Y_r}$.  Then, the following statements hold:
\begin{itemize}
    \item There exists an absolute constant $c_\star>0$ such that, for every $n\ge 1$ and every
$y\in\{0,\dots,n\}$,
\begin{equation}
\frac{P_{Y_r}(y)}{P_{Y^\star}(y)} \ge c_\star, \label{eq:lower_bound_ratio_output_dist}
\end{equation}
where one can take  {$c_\star   =
    \frac{\rme^{-8/25}}
    {
        \sqrt{6e/\pi}
        \left(
            \frac{5\pi}{2}+6
        \right)
    }
    \approx
    0.02300399$}
; and 
\item 
\begin{align}
    \kl{P_{Y_r}}{P_{Y^{\star}}}  &\le  \mathrm{Gap}(n), \text{  for $n \ge 1$,}  \\
   & \le \frac{17}{\log\!\left(\frac{n\pi}{2e}\right)}, \text{  for $n \ge 444$}.
    \label{eq:upper_gap}
\end{align}
\end{itemize}
\end{theorem}
\begin{IEEEproof} 
   The proof is given in Appendix~\ref{proofoftheorem3}.
\end{IEEEproof}

For our purposes, it will be easier to work with $\chi^2$-divergence. To that end, we have the following result.

\begin{prop}\label{prop:chi2-output}
    \begin{align}
    \chi^2(P_{Y^\star} \| P_{Y_r}) & \le \zeta(c_\star^{-1}) \,   \mathrm{Gap}(n), \text{ for $n \ge 1$}, \label{eq:upper_gap_implicit} \\
    &\le \zeta(c_\star^{-1}) \frac{17}{\log\!\left(\frac{n\pi}{2e}\right)}, \text{ for $n \ge 444$ }
    \label{eq:upper_chi2}
\end{align}
where 
\begin{equation}
    \zeta(t) = \frac{(t-1)^2}{t-1-\log(t)}.
\end{equation}
\end{prop}
\begin{IEEEproof}
    By combining \cite[Eq.~(169)]{sason2016f} and \cite[Thm.~6]{sason2016f}, we obtain
\begin{equation}\label{eq:from_chi2_to_kl}
    \chi^2(P_{Y^\star} \| P_{Y_r}) \le \zeta(\beta_1^{-1}) \kl{P_{Y_r}}{P_{Y^{\star}}},
\end{equation}
where
\begin{equation}\label{eq:lower_beta1}
    \beta_1 := \inf_y \frac{P_{Y_r}}{P_{Y^\star}}(y) \ge c_\star >0.
\end{equation}
Putting together~\eqref{eq:from_chi2_to_kl}, \eqref{eq:upper_gap_implicit} and ~\eqref{eq:upper_gap} yields
\begin{align}
    \chi^2(P_{Y^\star} \| P_{Y_r}) &\le \zeta(\beta_1^{-1}) \mathrm{Gap}(n)  \\
    &\le \zeta(c_\star^{-1}) \mathrm{Gap}(n)   \label{eq:monotonicity_of_Z}\\
    &\le \zeta(c_\star^{-1}) \frac{17}{\log\!\left(\frac{n\pi}{2e}\right)},
\end{align}
where in \eqref{eq:monotonicity_of_Z} we used~\eqref{eq:lower_beta1} and that $t \mapsto \zeta(t)$ is increasing for $t\ge 0$.
\end{IEEEproof}

\subsection{A Lower Bound on the Support Size}
\label{sec:lower_bound_support}

We now state the main result of this work, which is an improved lower bound on the cardinality of $P_{X^\ast}$.

\begin{theorem}\label{thm:main-support}
For every $n \ge 1$,
\begin{equation}\label{eq:explicit-support}
|\supp(P_{X^\ast})| \ge \max  \left\{2, \,  \rme^{C(n)}, \,  B(n) \right \},
\end{equation}
where $C(n)$ is the capacity defined in \eqref{eq:capacity_opt_problem} and 
\begin{align}
B(n)&= 
\frac{1}{8}
\min\left\{
\sqrt{
n \,
\log^{+}\!\left(
\frac{1}{4 \,\zeta(c_\star^{-1}) {\rm Gap}(n) }
\right)}, n
\right\} \label{eq:use_def_alphan} \\
&\ge
\frac{1}{8}
\min\left\{
\sqrt{
n \,
\log^{+}\!\left(
\frac{\log\!\left(\frac{n\pi}{2e}\right)}{3268}
\right)}, n
\right\}, \text{ for $n \ge 444$.} 
\end{align}
\end{theorem}
\begin{IEEEproof}
  The first bound in \eqref{eq:explicit-support} is trivial, since otherwise the mutual information would be zero.  The second bound in  \eqref{eq:explicit-support} was shown in \eqref{eq:sqrt(n) lower bound}. 
 The proof of the third bound is given in  Appendix~\ref{sec:proof_lowerbound_supp}. 
\end{IEEEproof}

\begin{corollary}\label{cor:expC_LB}
From Theorem~\ref{thm:main-support}, we can obtain the following explicit lower bound, valid for any $n\geq1$
\begin{align} \label{eq:expC_LB}
    |\supp(P_{X^{\star}})| \geq \rme^{\underline{C}(n)},
\end{align}
where $\underline{C}(n)$ is the capacity lower bound in~\eqref{eq:capacity-lb-asymp}.
\end{corollary}
\begin{proof}
    By monotonicity of the exponential, Theorem~\ref{thm:capacity-lb}, and Theorem~\ref{thm:main-support}, we get $|\supp(P_{X^{\star}})| \geq \rme^{C(n)} \geq \rme^{\underline{C}(n)}$.
\end{proof}

We make the following remarks:
\begin{itemize}
   
    \item In view of Theorems~\ref{thm:capacity-lb} and~\ref{thm:capacity-ub}, we have that  $\rme^{C(n)} =\Theta(\sqrt{n})$. Therefore, for large $n$ the bound in \eqref{eq:explicit-support} can be expressed as 
    \begin{equation}
        |\supp(P_{X^\ast})| \ge \Omega\!\left(\sqrt{n\log\log n}\right).
    \end{equation}
    \item The explicit constants are not optimized. In particular, the improvement over the order-$\sqrt n$ lower bound is asymptotic and becomes visible only for sufficiently large $n$.
\end{itemize}

\section{Numerical Results} \label{sec:num_res}
In the following we provide numerical estimates of the optimal input distribution $P_{X^\star}$ and of the channel capacity $C(n)$. We denote by $P_{\hat{X}^\star}$ and $P_{\hat{Y}^\star}$ the numerical estimates of, respectively, the capacity-achieving input and output distributions. Moreover, we denote by $\hat{C}(n)$ the channel capacity estimate.

The numerical estimates are derived via a two-stage optimization algorithm. In the first phase, we derive the $\hat{C}(n)$ and $P_{\hat{Y}^\star}$ estimates via the algorithm described in~\cite{komninakis2001capacity}. While the dual formulation in~\cite{komninakis2001capacity} is computationally efficient, deriving $P_{\hat{X}^\star}$ from $P_{\hat{Y}^\star}$ is not a trivial task. The purpose of the second optimization stage is to retrieve a reliable estimate of $P_{\hat{X}^\star}$. In this phase, we use the DAB optimization presented in~\cite{farsad2020capacities}. Generally, direct use of the DAB optimization can be faster with CAID characterized by a low-cardinality support and when past estimates $P_{\hat{X}^\star}$, from lower $n$, can be used to reliably initialize the optimization process. On the other hand, when $n$ increases, so does the cardinality of $\supp(P_{X^\star})$ and DAB complexity. 

Since the complexity of the dual formulation technique in~\cite{komninakis2001capacity} scales more slowly with $n$, we use the first stage to provide reasonably reliable estimates of $P_{\hat{Y}^\star}$. Then, we finely sample $i(x; P_{\hat{Y}^\star})$ for $x \in [0,1]$ and derive a rough estimate of $\supp(P_{X^\star})$ selecting the peaks of $i(x; P_{\hat{Y}^\star})$. Finally, we use the DAB optimization to derive the corresponding probabilities and refine further the initial estimate of $\supp(P_{X^\star})$. The optimization stops when the KKT conditions in Lemma~\ref{lem:KKT} are satisfied up to a tolerance $\varepsilon$. Specifically, in the DAB phase, we define $P_{\hat{Y}^\star}$ as
\begin{align}
P_{\hat{Y}^\star}(y) & = \sum_x P_{\hat{X}^\star}(x) P_{Y \mid X} (y\mid x), \label{eq:pYhat}
\end{align}
and evaluate the capacity estimate as $\hat{C}(n) = \sum_{x} P_{\hat{X}^\star}(x) i(x;P_{\hat{Y}^\star})$. The DAB algorithm reaches convergence when
\begin{subequations}
    \begin{align}
    i(x;P_{\hat{Y}^\star}) - \hat{C}(n) &\le  \varepsilon, \qquad x \in \mathcal{X}, \\
    | i(x;P_{\hat{Y}^\star}) - \hat{C}(n) | &\leq \varepsilon, \qquad x \in \supp(P_{\hat{X}^\star}),  \label{eq:KKT_estimate}
\end{align}
\end{subequations}
where $\mathcal{X}$ is a fine sampling grid over the interval $[0,1]$.

Let us now provide numerical results obtained for a tolerance $\varepsilon = 10^{-6}$ and compare the derived estimates with the bounds proposed throughout this work. The shown numerical estimates were derived for $n \in \{4,10,20,30, \dots, 350\}$.

In Sec.~\ref{sec:cap_bounds}, we show that the reference input $X_r \sim \mathrm{Beta}(\frac{1}{2},\frac{1}{2})$ represents a good asymptotic approximation of the CAID. Denote by $F_{X_r}(x) = \frac{2}{\pi}\arcsin(\sqrt{x})$ the cumulative distribution of $X_r$. As $X^{\star}$ tends to $X_r$, we expect $\supp(X^{\star})$ to become uniformly spaced in the transformed domain $F_{X_r}(x)$. In Fig.~\ref{fig:pmf_evo}, we display support locations in the transformed domain $F_{X_r}(x)$ and, indeed, $F_{X_r}(\supp(\hat{X}^{\star}))$ appears to be roughly uniformly spread, especially at large values of $n$. Moreover, this representation facilitates the visualization of the support for larger $n$, where many mass points are present and would \emph{visually overlap} at the edges, in the linear domain. Notice also that the black dashed lines represent the transformed-domain thresholds $1/n$ and $1-1/n$ from Proposition~\ref{prop:loc_info}.

In Fig.~\ref{fig:suppPx} and Fig.~\ref{fig:capacity}, we show bounds and numerical estimates versus $n$ of, respectively, the input support cardinality $|\supp(P_{X^\star})|$ and the channel capacity $C(n)$ in bit per channel use (bpcu). In both cases, lower bounds appear to be tighter than the upper bounds, in the considered range of $n$.

Let us define the following relative capacity gaps, i.e., the normalized lower-gap estimate 
\begin{equation}
    \underline{G}(n) = \frac{\hat{C}(n) - \underline{C}(n)}{\min\left\{\overline{C}(n),\overline{C}_1(n)\right\} - \underline{C}(n)}, \label{eq:norm-lb-gap}
\end{equation}
and the normalized upper-gap estimate
\begin{equation} 
    \overline{G}(n) = \frac{\min\left\{\overline{C}(n),\overline{C}_1(n)\right\}-\hat{C}(n)}{\min\left\{\overline{C}(n),\overline{C}_1(n)\right\} - \underline{C}(n)}. \label{eq:norm-ub-gap}
\end{equation}
Notice that, in the considered range of $n$ values, $\overline{C}_1(n)$ is significantly tighter than $\overline{C}(n)$. In Fig.~\ref{fig:norm_gaps}, we show $\underline{G}(n)$ and $\overline{G}(n)$. The figure suggests that the gap between the upper bound and the true capacity may be further improved. 

Fig.~\ref{fig:C_vs_log2_k} shows the capacity estimate $\hat{C}(n)$ scaling, as a function of $\log_2 |\supp(P_{\hat{X}^\star})|$. The capacity samples in the figure are from a denser sampling of $n$, to more appropriately display the capacity scaling versus the support cardinality. Specifically, we considered $n \in \{1,2, \dots, 80\} \cup \{90,100, \dots, 350\}$. The red lines represent visually fitted trends, intended only to approximate the capacity scaling law behavior in the given range. The exponent $3/4$, as conjectured in~\cite{mattingly2018maximizing}, appears to provide the best fit for $n\le 350$. We caution, however, against drawing a firm conclusion from this numerical evidence alone, since the range $n\le 350$ may still correspond to a finite-$n$ regime in which the asymptotic scaling is not yet visible.

Finally, in Fig.~\ref{fig:pY_vs_pYr}, we show the evolution of the output distributions $P_{\hat{Y}^\star}$ in~\eqref{eq:pYhat} and $P_{Y_r}$ in~\eqref{eq:output_dist_X_r}. Both distributions are normalized by $P_{\hat{Y}^\star}(1)$ for each $n$ and are displayed versus the normalized upper-half of the output support, i.e., $y/n \in \left\{\frac{1}{2},\frac{1}{2}+\frac{1}{n}, \dots, 1-\frac{1}{n},1\right\}$.  Notice that, since both output distributions are symmetric about $n/2$, we only show the upper-half support to improve the readability of the figure. Fig.~\ref{fig:pY_vs_pYr} qualitatively shows how $P_{\hat{Y}^\star}$ approaches $P_{Y_r}$ as $n$ increases.

\begin{figure}
\begin{minipage}[t]{0.45\linewidth}
    \centering
    \input{IEEE_TIT_submission_correction/Figures/PMFevo.tex}
    \caption{Evolution, versus $n$ and in the transformed domain $F_{X_r}(x)$, of the support estimate $\supp(P_{\hat{X}^\star})$ (blue curves) and thresholds $1/n$ and $1-1/n$ from Prop.~\ref{prop:loc_info} (black dashed curves).}
    \label{fig:pmf_evo}
\end{minipage}
\hfill
\begin{minipage}[t]{0.45\linewidth}
    \centering
    \begin{tikzpicture}
	
	\begin{axis}[%
		width=8cm,
		height=6cm,
		xmin=4,
		xmax=350,
		xlabel style={font=\color{white!15!black}},
		xlabel={$n$},
		ymin=0,
		ymax=250,
		ylabel style={font=\color{white!15!black}},
		ylabel={Input Support Cardinality: $|\supp(P_{X^\star})|$},
        axis x line*=bottom,
        axis y line*=left,
		xmajorgrids,
		ymajorgrids,
        legend cell align=left,
        legend pos=north west
		]
        
		\addplot [color=blue, line width = 0.8pt]
		table[row sep=crcr]{%
            4   4\\
			10	7\\
			20	12\\
			30	17\\
			40	22\\
			50	27\\
			60	32\\
			70	37\\
			80	42\\
			90	47\\
			100	52\\
			110	57\\
			120	62\\
			130	67\\
			140	72\\
			150	77\\
			160	82\\
			170	87\\
			180	92\\
			190	97\\
			200	102\\
			210	107\\
			220	112\\
			230	117\\
			240	122\\
			250	127\\
			260	132\\
			270	137\\
			280	142\\
			290	147\\
			300	152\\
			310	157\\
			320	162\\
			330	167\\
			340	172\\
			350	177\\
		};
        \addlegendentry{Upper Bound in~\eqref{eq: n/2 bound}}
        
		\addplot [color=black, line width = 0.8pt]
		table[row sep=crcr]{%
            4   3\\
			10	5\\
			20	7\\
			30	9\\
			40	10\\
			50	11\\
			60	13\\
			70	14\\
			80	15\\
			90	16\\
			100	17\\
			110	18\\
			120	19\\
			130	20\\
			140	21\\
			150	22\\
			160	23\\
			170	24\\
			180	25\\
			190	26\\
			200	27\\
			210	27\\
			220	28\\
			230	29\\
			240	30\\
			250	31\\
			260	31\\
			270	32\\
			280	33\\
			290	34\\
			300	34\\
			310	35\\
			320	36\\
			330	37\\
			340	38\\
			350	38\\
		};
        \addlegendentry{$|\supp(P_{\hat{X}^\star})|$, num. estimate}
        
        \addplot [color=red, line width = 0.8pt]
        table[row sep=crcr]{%
        	4	1.01427309245884\\
            10	1.89153719073153\\
            20	2.89743031314996\\
            30	3.66917482157263\\
            40	4.31880550356318\\
            50	4.89043055466426\\
            60	5.40672535481279\\
            70	5.88115495495322\\
            80	6.32248306504266\\
            90	6.73678953273831\\
            100	7.12849580915091\\
            110	7.50093608032229\\
            120	7.85669810910735\\
            130	8.19783797067912\\
            140	8.52602138461779\\
            150	8.84262014995276\\
            160	9.14877996735867\\
            170	9.44546938364213\\
            180	9.73351590689809\\
            190	10.013633176402\\
            200	10.2864417540386\\
            210	10.5524852767467\\
            220	10.8122431752993\\
            230	11.0661408113145\\
            240	11.3145576453828\\
            250	11.5578338843645\\
            260	11.7962759401948\\
            270	12.0301609499864\\
            280	12.2597405474572\\
            290	12.4852440318636\\
            300	12.706881048043\\
            310	12.9248438666936\\
            320	13.1393093354343\\
            330	13.3504405569305\\
            340	13.5583883393453\\
            350	13.7632924557628\\
        };
        \addlegendentry{Lower Bound in~\eqref{eq:expC_LB}}
        
	\end{axis}
\end{tikzpicture}%
    \caption{Numerical estimates and bounds on the input support cardinality $|\supp(P_{X^\star})|$ versus $n$.}
    \label{fig:suppPx}
\end{minipage}
\end{figure}

\begin{figure}
    \centering
    \definecolor{mygreen}{RGB}{46,200,50}
\begin{tikzpicture}

\begin{axis}[%
width=8cm,
height=6cm,
xmin=4,
xmax=350,
xlabel style={font=\color{white!15!black}},
xlabel={$n$},
ymin=0,
ymax=11,
ylabel style={font=\color{white!15!black}},
ylabel={Channel Capacity (bpcu)},
axis background/.style={fill=white},
legend cell align=left,
legend style={
    at={(1.05,1)},
    anchor=north west,
},
xmajorgrids,
ymajorgrids
]

\addplot [color=blue, densely dashed, line width = 0.5pt]
table[row sep=crcr]{%
	4	2\\
	10	2.8073549220576\\
	20	3.58496250072116\\
	30	4.08746284125034\\
	40	4.4594316186373\\
	50	4.75488750216347\\
	60	5\\
	70	5.20945336562895\\
	80	5.39231742277876\\
	90	5.55458885167764\\
	100	5.70043971814109\\
	110	5.83289001416474\\
	120	5.95419631038688\\
	130	6.06608919045777\\
	140	6.16992500144231\\
	150	6.2667865406949\\
	160	6.35755200461808\\
	170	6.44294349584873\\
	180	6.52356195605701\\
	190	6.59991284218713\\
	200	6.6724253419715\\
	210	6.74146698640115\\
	220	6.8073549220576\\
	230	6.8703647195834\\
	240	6.93073733756289\\
	250	6.98868468677217\\
	260	7.04439411935845\\
	270	7.09803208296053\\
	280	7.14974711950468\\
	290	7.19967234483636\\
	300	7.24792751344359\\
	310	7.29462074889163\\
	320	7.33985000288462\\
	330	7.38370429247405\\
	340	7.4262647547021\\
	350	7.467605550083\\
};
\addlegendentry{Upper Bound in~\cite[Eq.~(53)]{Zieder2024}}

\addplot [color=blue, densely dotted, line width = 0.5pt]
table[row sep=crcr]{%
	4	5.25004996294657\\
	10	5.91101401039025\\
	20	6.41101401039025\\
	30	6.70349526075083\\
	40	6.91101401039025\\
	50	7.07197805783393\\
	60	7.20349526075083\\
	70	7.31469147141905\\
	80	7.41101401039025\\
	90	7.49597651111141\\
	100	7.57197805783393\\
	110	7.6407298197089\\
	120	7.70349526075083\\
	130	7.76123386946079\\
	140	7.81469147141905\\
	150	7.86445930819451\\
	160	7.91101401039025\\
	170	7.95474543101542\\
	180	7.99597651111141\\
	190	8.03497776711204\\
	200	8.07197805783393\\
	210	8.10717272177963\\
	220	8.1407298197089\\
	230	8.17279498841875\\
	240	8.20349526075083\\
	250	8.23294210527761\\
	260	8.26123386946079\\
	270	8.28845776147198\\
	280	8.31469147141905\\
	290	8.34000450795403\\
	300	8.36445930819451\\
	310	8.38811216558369\\
	320	8.41101401039025\\
	330	8.43321107006948\\
	340	8.45474543101542\\
	350	8.47565551886273\\
};
\addlegendentry{Upper Bound in~\cite[Thm.~1]{TangCapBoundsBinomil}}

\addplot [color=blue, line width = 0.8pt]
  table[row sep=crcr]{%
30	12.7698472063506\\
40	10.367809965475\\
50	9.57945965323824\\
60	9.13448646535147\\
70	8.84076330715802\\
80	8.63030732489015\\
90	8.47147402504577\\
100	8.34717115162596\\
110	8.24723578464013\\
120	8.16519909237347\\
130	8.09672699853892\\
140	8.03879665097482\\
150	7.98923019010056\\
160	7.94641589878528\\
170	7.9091338075184\\
180	7.87644250105905\\
190	7.84760330381509\\
200	7.82202811255603\\
210	7.7992426496166\\
220	7.7788600405787\\
230	7.76056146665069\\
240	7.74408176557947\\
250	7.72919855806522\\
260	7.71572392768259\\
270	7.70349797811996\\
280	7.69238378946887\\
290	7.68226343014279\\
300	7.67303477440613\\
310	7.66460894117521\\
320	7.65690821658391\\
330	7.6498643566289\\
340	7.64341719092631\\
350	7.63751346687505\\
};
\addlegendentry{$\overline{C}(n)$ in~\eqref{eq:capacity-ub}}

\addplot [color=mygreen, line width = 0.8pt]
table[row sep=crcr]{%
    4	1.68650052718322\\
    10	2.22039672596666\\
    20	2.65388195532919\\
    30	2.91637361342171\\
    40	3.10586461174323\\
    50	3.25446207806976\\
    60	3.37681767521722\\
    70	3.48087426288073\\
    80	3.57142953605741\\
    90	3.65160672632059\\
    100	3.72355426179935\\
    110	3.78881401171516\\
    120	3.84853040629056\\
    130	3.90357646750939\\
    140	3.95463349005358\\
    150	4.00224344533836\\
    160	4.04684460587758\\
    170	4.08879645489519\\
    180	4.12839753219163\\
    190	4.16589849290323\\
    200	4.2015118428507\\
    210	4.23541931703594\\
    220	4.26777755485923\\
    230	4.29872252340481\\
    240	4.3283730064202\\
    250	4.35683338637713\\
    260	4.38419588490688\\
    270	4.41054238348964\\
    280	4.4359459154162\\
    290	4.46047189780428\\
    300	4.48417915622433\\
    310	4.50712078249024\\
    320	4.52934485719842\\
    330	4.55089506184172\\
    340	4.57181120016567\\
    350	4.59212964444691\\
};
\addlegendentry{$\overline{C}_1(n)$ in~\eqref{eq:C1_UB_finite_n}}

\addplot [color=mygreen, densely dashed, line width = 0.8pt]
table[row sep=crcr]{%
    4	2.17204845549739\\
    10	2.57611752134614\\
    20	2.92694658166884\\
    30	3.1480692443047\\
    40	3.31130759367671\\
    50	3.44125362530936\\
    60	3.54943871077432\\
    70	3.64223696930579\\
    80	3.72355664291532\\
    90	3.7959728134897\\
    100	3.86127467353816\\
    110	3.92075724812236\\
    120	3.97538884508351\\
    130	4.02591290914933\\
    140	4.07291294001832\\
    150	4.11685549390981\\
    160	4.158119603883\\
    170	4.19701747060238\\
    180	4.23380936376676\\
    190	4.26871457877037\\
    200	4.30191964108371\\
    210	4.33358454984734\\
    220	4.36384759839178\\
    230	4.39282914464137\\
    240	4.42063459494126\\
    250	4.44735679067899\\
    260	4.47307793586251\\
    270	4.49787116786144\\
    280	4.52180184788455\\
    290	4.5449286292335\\
    300	4.56730434780022\\
    310	4.58897676921636\\
    320	4.60998921952582\\
    330	4.63038112054587\\
    340	4.6501884467215\\
    350	4.66944411691581\\
            };
\addlegendentry{$\overline{C}_2(n)$ in~\eqref{eq:C2_UB_finite_n}}

\addplot [color=black, line width = 0.8pt]
  table[row sep=crcr]{%
4   1.372300974844358\\
10	1.77805893825711\\
20	2.14120186959974\\
30	2.36990831715177\\
40	2.53840859956353\\
50	2.67231094790467\\
60	2.7836102069649\\
70	2.87894915126133\\
80	2.96239520782702\\
90	3.03662686798978\\
100	3.10350368687856\\
110	3.1643704039794\\
120	3.22023168127057\\
130	3.27185817601249\\
140	3.31985429782191\\
150	3.36470295038998\\
160	3.40679616616725\\
170	3.44645665355729\\
180	3.48395324923845\\
190	3.51951235102573\\
200	3.55332640807043\\
210	3.58556041155498\\
220	3.61635688025569\\
230	3.64583982041056\\
240	3.67411775942071\\
250	3.7012862592827\\
260	3.72742988355639\\
270	3.75262385366366\\
280	3.77693535619\\
290	3.80042466247182\\
300	3.82314605144426\\
310	3.84514856994925\\
320	3.86647669615506\\
330	3.88717088596043\\
340	3.90726804821599\\
350	3.92680194057999\\
};
\addlegendentry{$\hat{C}(n)$, num. estimate}

\addplot [color=red, line width = 0.8pt]
  table[row sep=crcr]{%
4	0.0204461494720559\\
10	0.919559142422956\\
20	1.53477396301536\\
30	1.87545564489915\\
40	2.11063234656042\\
50	2.28996148598214\\
60	2.43475507376538\\
70	2.55609950270823\\
80	2.66049126759248\\
90	2.75206122810392\\
100	2.83359768429628\\
110	2.90707064816435\\
120	2.97392312641943\\
130	3.03524347555359\\
140	3.09187267381503\\
150	3.14447391681716\\
160	3.1935793660855\\
170	3.23962249175675\\
180	3.28296102445364\\
190	3.3238936069816\\
200	3.36267211195167\\
210	3.39951091127802\\
220	3.43459395950979\\
230	3.4680802814974\\
240	3.50010827690378\\
250	3.53079913484466\\
260	3.56025957052566\\
270	3.58858403916311\\
280	3.61585654252474\\
290	3.64215211480137\\
300	3.66753805373458\\
310	3.69207494764363\\
320	3.71581753763043\\
330	3.73881544569994\\
340	3.76111379305034\\
350	3.78275372782034\\
};
\addlegendentry{$\underline{C}(n)$ in~\eqref{eq:capacity-lb-asymp}}

\addplot [color=red, densely dashed, line width = 0.5pt]
table[row sep=crcr]{%
	4	1\\
	10	1\\
	20	1\\
	30	1\\
	40	1\\
	50	1\\
	60	1\\
	70	1\\
	80	1\\
	90	1\\
	100	1\\
	110	1\\
	120	1\\
	130	1\\
	140	1\\
	150	1\\
	160	1\\
	170	1\\
	180	1\\
	190	1\\
	200	1\\
	210	1\\
	220	1\\
	230	1\\
	240	1\\
	250	1\\
	260	1\\
	270	1.00666983921425\\
	280	1.0418827452104\\
	290	1.0757662727985\\
	300	1.10841424878093\\
	310	1.13991101938314\\
	320	1.17033267814354\\
	330	1.19974810198397\\
	340	1.22821983024197\\
	350	1.25580481431652\\
};
\addlegendentry{Lower Bound in~\cite[Eq.~(52)]{Zieder2024}}

\addplot [color=red, densely dotted, line width = 0.5pt]
  table[row sep=crcr]{%
70	0.0211655338755794\\
80	0.0951649207981956\\
90	0.160994557830096\\
100	0.220300858537502\\
110	0.274275756459357\\
120	0.323810298427255\\
130	0.36958831161509\\
140	0.412146197978492\\
150	0.451912578713712\\
160	0.489235431065045\\
170	0.524401176766146\\
180	0.557648429219106\\
190	0.58917809978814\\
200	0.619160963385441\\
210	0.647743413968354\\
220	0.675051906467111\\
230	0.701196429585275\\
240	0.726273252871427\\
250	0.750367122950814\\
260	0.773553036495644\\
270	0.795897684293612\\
280	0.81746063709194\\
290	0.838295326773754\\
300	0.858449863887091\\
310	0.877967723257631\\
320	0.896888322457165\\
330	0.915247512633139\\
340	0.9330779971806\\
350	0.950409690636087\\
};
\addlegendentry{Lower Bound in~\cite[Thm.~2]{TangCapBoundsBinomil}}

\end{axis}
\end{tikzpicture}
    \caption{Numerical estimates and bounds on the channel capacity in (bpcu) versus $n$.}
    \label{fig:capacity}
\end{figure}
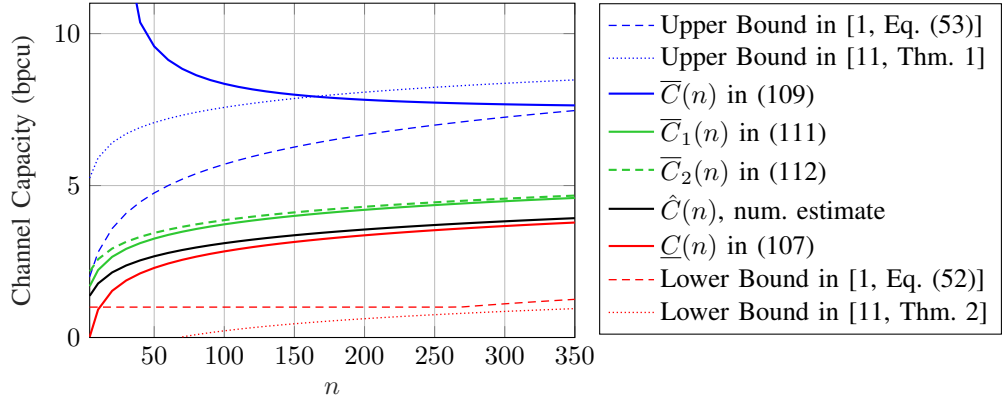

\begin{figure}
\begin{minipage}[t]{0.45\linewidth}
    \centering
    \begin{tikzpicture}

\begin{axis}[%
width=8cm,
height=6cm,
unbounded coords=jump,
xmin=30,
xmax=350,
xlabel style={font=\color{white!15!black}},
xlabel={$n$},
ymode=log,
ymin=0.1,
ymax=1.2,
yminorticks=true,
ylabel style={font=\color{white!15!black}},
ylabel={Normalized Capacity Gap},
axis background/.style={fill=white},
xmajorgrids,
ymajorgrids,
yminorgrids,
legend style={
    at={(0.98,0.55)},
    anchor=east
}
]
\addplot [color=blue, line width = 0.8pt]
    table[row sep=crcr]{%
			4	0.412051693130411\\
            10	0.340040750133115\\
            20	0.458114935511674\\
            30	0.52498401679584\\
            40	0.57017445277004\\
            50	0.603577784130805\\
            60	0.629690072971917\\
            70	0.650888343348734\\
            80	0.668579144508555\\
            90	0.683656201437278\\
            100	0.696720031735088\\
            110	0.708192013173889\\
            120	0.718378110358934\\
            130	0.727506955683021\\
            140	0.735753386435846\\
            150	0.743253838880826\\
            160	0.750116622430674\\
            170	0.756428987723432\\
            180	0.762262188886808\\
            190	0.767675048785481\\
            200	0.77271665957641\\
            210	0.77742836536227\\
            220	0.781845295850221\\
            230	0.785997472865133\\
            240	0.789910790214953\\
            250	0.79360768137429\\
            260	0.797107725302439\\
            270	0.800428068365249\\
            280	0.80358383001924\\
            290	0.806588388845188\\
            300	0.809453629954132\\
            310	0.8121901667844\\
            320	0.814807499513799\\
            330	0.817314168079561\\
            340	0.819717870215348\\
            350	0.822025575754663\\
		};
		\addlegendentry{$\overline{G}(n)$ in~\eqref{eq:norm-ub-gap}}

		\addplot [color=red, line width = 0.8pt]
		table[row sep=crcr]{%
			4	0.587948306869589\\
            10	0.659959249866885\\
            20	0.541885064488326\\
            30	0.47501598320416\\
            40	0.42982554722996\\
            50	0.396422215869195\\
            60	0.370309927028083\\
            70	0.349111656651266\\
            80	0.331420855491445\\
            90	0.316343798562722\\
            100	0.303279968264912\\
            110	0.291807986826111\\
            120	0.281621889641066\\
            130	0.272493044316979\\
            140	0.264246613564154\\
            150	0.256746161119174\\
            160	0.249883377569326\\
            170	0.243571012276568\\
            180	0.237737811113192\\
            190	0.232324951214519\\
            200	0.22728334042359\\
            210	0.22257163463773\\
            220	0.218154704149779\\
            230	0.214002527134867\\
            240	0.210089209785047\\
            250	0.20639231862571\\
            260	0.202892274697561\\
            270	0.199571931634751\\
            280	0.19641616998076\\
            290	0.193411611154812\\
            300	0.190546370045868\\
            310	0.1878098332156\\
            320	0.185192500486201\\
            330	0.182685831920439\\
            340	0.180282129784652\\
            350	0.177974424245337\\
		};
		\addlegendentry{$\underline{G}(n)$ in~\eqref{eq:norm-lb-gap}}

\end{axis}

\end{tikzpicture}%
    \caption{Normalized capacity gap estimates versus $n$.}
    \label{fig:norm_gaps}
\end{minipage}
\hfill
\begin{minipage}[t]{0.45\linewidth}
    \centering
    \begin{tikzpicture}

\begin{axis}[%
width=8cm,
height=6cm,
xmin=1,
xmax=5.248,
xlabel style={font=\color{white!15!black}},
xlabel={$\log_2|\supp(P_{\hat{X}^\star})|$},
ymin=1,
ymax=4.7,
ylabel style={font=\color{white!15!black}},
ylabel={Capacity Estimate (bpcu): $\hat{C}(n)$},
legend pos=north west,
legend cell align=left,
axis background/.style={fill=white},
xmajorgrids,
ymajorgrids,
yminorgrids,
legend style={
    at={(0.02,0.98)},
    font = \footnotesize
}
]
\addplot [color=blue, only marks, mark size=1pt, mark=*, mark options={solid, blue}, forget plot]
  table[row sep=crcr]{%
1.58496250072116	1.08746284125034\\
1.58496250072116	1.24792751344359\\
1.58496250072116	1.37230097484436\\
};

\addplot [color=blue, only marks, mark size=1pt, mark=*, mark options={solid, blue}, forget plot]
  table[row sep=crcr]{%
2	1.45802638586269\\
2	1.53595480712972\\
2	1.60804256873501\\
2	1.67159033789592\\
};

\addplot [color=blue, only marks, mark size=1pt, mark=*, mark options={solid, blue}, forget plot]
  table[row sep=crcr]{%
2.32192809488736	1.72686863853397\\
2.32192809488736	1.77805894436607\\
2.32192809488736	1.82590802276674\\
2.32192809488736	1.8703518605098\\
2.32192809488736	1.91142693374531\\
};

\addplot [color=blue, only marks, mark size=1pt, mark=*, mark options={solid, blue}, forget plot]
  table[row sep=crcr]{%
2.58496250072116	1.94969182963013\\
2.58496250072116	1.9859298239721\\
2.58496250072116	2.02028737463167\\
2.58496250072116	2.0528665871509\\
2.58496250072116	2.08374902766901\\
};

\addplot [color=blue, only marks, mark size=1pt, mark=*, mark options={solid, blue}, forget plot]
  table[row sep=crcr]{%
2.8073549220576	2.1131041980778\\
2.8073549220576	2.14120188578806\\
2.8073549220576	2.1681423213418\\
2.8073549220576	2.19400148008857\\
2.8073549220576	2.21884053432873\\
};

\addplot [color=blue, only marks, mark size=1pt, mark=*, mark options={solid, blue}, forget plot]
  table[row sep=crcr]{%
3	2.24271261271512\\
3	2.26572079334056\\
3	2.28794500813886\\
3	2.30943588403581\\
3	2.33023638741438\\
3	2.35038319109951\\
};

\addplot [color=blue, only marks, mark size=1pt, mark=*, mark options={solid, blue}, forget plot]
  table[row sep=crcr]{%
3.16992500144231	2.36990834081774\\
3.16992500144231	2.38885805135112\\
3.16992500144231	2.4072735714597\\
3.16992500144231	2.42518457036742\\
3.16992500144231	2.44261733491251\\
3.16992500144231	2.4595952002123\\
3.16992500144231	2.4761389612803\\
};

\addplot [color=blue, only marks, mark size=1pt, mark=*, mark options={solid, blue}, forget plot]
  table[row sep=crcr]{%
3.32192809488736	2.49226933236213\\
3.32192809488736	2.50801190128789\\
3.32192809488736	2.52338597551071\\
3.32192809488736	2.53840859952247\\
3.32192809488736	2.55309529349615\\
3.32192809488736	2.5674603071597\\
3.32192809488736	2.58151661487656\\
};

\addplot [color=blue, only marks, mark size=1pt, mark=*, mark options={solid, blue}, forget plot]
  table[row sep=crcr]{%
3.4594316186373	2.59527709910277\\
3.4594316186373	2.60875621438641\\
3.4594316186373	2.62196567087793\\
3.4594316186373	2.63491620904799\\
3.4594316186373	2.6476178485674\\
3.4594316186373	2.6600798753856\\
3.4594316186373	2.67231094790467\\
};

\addplot [color=blue, only marks, mark size=1pt, mark=*, mark options={solid, blue}, forget plot]
  table[row sep=crcr]{%
3.58496250072116	2.68431918747556\\
3.58496250072116	2.69611359600831\\
3.58496250072116	2.70770193134566\\
3.58496250072116	2.71909150777926\\
3.58496250072116	2.73028909506733\\
3.58496250072116	2.74130112099383\\
3.58496250072116	2.7521336064158\\
3.58496250072116	2.76279222783134\\
};

\addplot [color=blue, only marks, mark size=1pt, mark=*, mark options={solid, blue}, forget plot]
  table[row sep=crcr]{%
3.70043971814109	2.773282548729979\\
3.70043971814109	2.7836102069649\\
3.70043971814109	2.79378036001689\\
3.70043971814109	2.80379780468475\\
3.70043971814109	2.81366715549933\\
3.70043971814109	2.82339275642005\\
3.70043971814109	2.83297876113609\\
3.70043971814109	2.84242907939975\\
};

\addplot [color=blue, only marks, mark size=1pt, mark=*, mark options={solid, blue}, forget plot]
  table[row sep=crcr]{%
3.8073549220576	2.85174755530108\\
3.8073549220576	2.86093800644048\\
3.8073549220576	2.87000409804284\\
3.8073549220576	2.87894915126133\\
3.8073549220576	2.88777636721379\\
3.8073549220576	2.8964888984373\\
3.8073549220576	2.90508968777324\\
3.8073549220576	2.9135816141029\\
3.8073549220576	2.92196740417371\\
};

\addplot [color=blue, only marks, mark size=1pt, mark=*, mark options={solid, blue}, forget plot]
  table[row sep=crcr]{%
3.90689059560852	2.9302497447319\\
3.90689059560852	2.93843128088924\\
3.90689059560852	2.94651446701999\\
3.90689059560852	2.9545016376593\\
3.90689059560852	2.96239520782702\\
};

\addplot [color=blue, only marks, mark size=1pt, mark=*, mark options={solid, blue}]
  table[row sep=crcr]{%
1	                1\\
4	3.03662686798978\\
4.08746284125034	3.10350368687856\\
4.16992500144231	3.1643704039794\\
4.24792751344359	3.22023168127057\\
4.32192809488736	3.27185817601249\\
4.39231742277876	3.31985429782191\\
4.4594316186373	3.36470295038998\\
4.52356195605701	3.40679616616725\\
4.58496250072116	3.44645665355729\\
4.64385618977472	3.48395324923845\\
4.70043971814109	3.51951235102573\\
4.75488750216347	3.55332640807043\\
4.75488750216347	3.58556041155498\\
4.8073549220576	3.61635688025569\\
4.85798099512757	3.64583982041056\\
4.90689059560852	3.67411775942071\\
4.95419631038687	3.7012862592827\\
4.95419631038687	3.72742988355639\\
5	3.75262385366366\\
5.04439411935845	3.77693535619\\
5.08746284125034	3.80042466247182\\
5.08746284125034	3.82314605144426\\
5.12928301694497	3.84514856994925\\
5.16992500144231	3.86647669615506\\
5.24792751344359	3.88717088596043\\
5.24792751344359	3.90726804821599\\
5.24792751344359	3.92680194057999\\
};
\addlegendentry{$\hat{C}(n)$}

\addplot [color=red, dashed, line width = 0.8pt]
  table[row sep=crcr]{%
1	1\\
5.32192809488736	4.025349666421151\\
};
\addlegendentry{$0.3 + 0.7 \log_2|\supp(P_{\hat{X}^\star})|$}

\addplot [color=red, line width = 0.8pt]
  table[row sep=crcr]{%
1	0.7\\
5.32192809488736	3.941446071165520\\
};
\addlegendentry{$-0.05 + 0.75 \log_2|\supp(P_{\hat{X}^\star})|$}

\end{axis}
\end{tikzpicture}%
    \caption{The blue dots are the capacity estimates as a function of $\log_2|\supp(P_{\hat{X}^\star})|$. Each dot represents one value of $n$. The red lines are visually fitted trends of the capacity scaling law.}
    \label{fig:C_vs_log2_k}
\end{minipage}
\end{figure}

\begin{figure}
    \centering
    \input{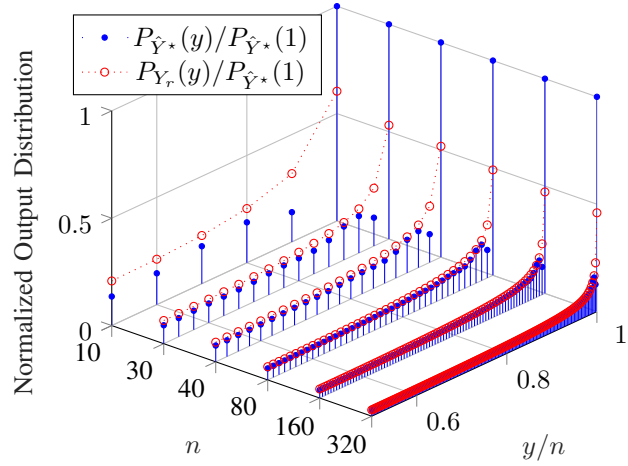}
    \caption{Evolution, versus $n$ and the normalized half-support, of the estimated output distribution $P_{\hat{Y}^{\star}}$ and of $P_{Y_r}$, both normalized by $P_{\hat{Y}^{\star}}(1)$ for each $n$.}
    \label{fig:pY_vs_pYr}
\end{figure}

\section{Conclusion} \label{sec:conclusion}

In this work, we investigated the binomial channel and derived structural properties of the optimal input and output distributions. With respect to the capacity-achieving input distribution, we proved its discreteness and uniqueness. Furthermore, we proved that the optimal input distribution is symmetric about $1/2$ and always comprises $\{0,1\}$.

We investigated the beta-binomial output, induced by a beta input distribution, and proved that it asymptotically approaches the capacity-achieving output distribution in relative entropy and $\chi^2$-divergence.

We also proved that the beta-binomial output cannot be approximated too accurately by binomial mixtures with few components and derive a lower bound on the input support cardinality of order $\Omega(\sqrt{n\log\log n})$.

We derived explicit bounds and the support size and on the channel capacity $C(n)$, showing that $C(n) = \frac{1}{2}\log\frac{n\pi}{2e}+o(1)$. Finally, we derived the channel capacity and the optimal input distribution in closed form for $n \leq 3$, and derived numerical estimates for up to $n=350$.

Possible directions for future investigations include further improving the bounds on the channel capacity and on the support size of the capacity-achieving input distribution, as well as extending the proposed analytical framework to multinomial channels.

\begin{appendices}

\section{Derivatives of Information Density}
\label{sec:Derivatives_info_density}

\begin{lemma}\label{lem:derivative_exp_cnd}
For any function $f:\bbN_0 \cap [0,n]\to \bbR$, and $x \in (0,1)$, we have
\begin{align}
    \frac{\rmd}{\rmd x}\expcnd{f(Y)}{X=x} 
    &= n \expcndn{n-1}{f(Y+1)-f(Y)}{X=x} \label{eq:der_cond_exp_n-1}\\
    &=
    \frac{1}{x}\expcnd{Yf(Y)}{X=x}-\frac{1}{1-x}\expcnd{(n-Y)f(Y)}{X=x} \label{eq:der_cond_exp_0}\\
    &= \frac{1}{x}\expcnd{Y(f(Y)-f(Y-1))}{X=x} \label{eq:der_cond_exp_1} \\
    &= \frac{1}{1-x}\expcnd{(n-Y)(f(Y+1)-f(Y))}{X=x},\label{eq:der_cond_exp_2}
\end{align}
\end{lemma}
\begin{proof}
Let us first prove \eqref{eq:der_cond_exp_n-1}. Note that 
\begin{align}
    \frac{\rmd}{\rmd x} P_{Y|X}(y|x) &= \binom{n}{y} \frac{\rmd}{\rmd x} x^y (1-x)^{n-y}=\frac{y-nx}{x(1-x)}\binom{n}{y}  x^y (1-x)^{n-y}  \\
    &= \frac{y-nx}{x(1-x)}P_{Y|X}(y|x) = \left(\frac{y}{x}-\frac{n-y}{1-x}\right)P_{Y|X}(y|x).
\end{align}
Hence, we can write
\begin{align}
    \frac{\rmd}{\rmd x}\expcnd{f(Y)}{X=x} &= \sum_{y=0}^n f(y) \frac{\rmd}{\rmd x} P_{Y|X}(y|x) \\
    &= \sum_{y=0}^n f(y) \left(\frac{y}{x}-\frac{n-y}{1-x}\right)P_{Y|X}(y|x) \label{eq:obtain_der_cond_exp_0} \\
    &= \sum_{y=0}^{n-1} f(y+1) \frac{y+1}{x}P_{Y|X}(y+1|x) -\sum_{y=0}^{n-1} f(y) \frac{n-y}{1-x}P_{Y|X}(y|x)\\
    &= \sum_{y=0}^{n-1} f(y+1) nP_{Y|X}^{n-1}(y|x)-\sum_{y=0}^{n-1} f(y) nP_{Y|X}^{n-1}(y|x)\\
    &= n \expcndn{n-1}{f(Y+1)-f(Y)}{X=x}.
\end{align}
Result~\eqref{eq:der_cond_exp_0} is obtained in~\eqref{eq:obtain_der_cond_exp_0}. 

To prove \eqref{eq:der_cond_exp_1}, first note that
\begin{align}
     \frac{\rmd}{\rmd x}\expcnd{f(Y)}{X=x} 
    &=   \sum_{y=0}^n f(y) \binom{n}{y}\frac{\rmd}{\rmd x} x^y (1-x)^{n-y} \\
    &= \sum_{y=0}^n f(y) \binom{n}{y} x^y (1-x)^{n-y} \left(\frac{y}{x}-\frac{n-y}{1-x} \right). \label{eq:diff_expcnd}
\end{align}
Next, consider the term
\begin{align}
    \sum_{y=0}^n f(y) \binom{n}{y} x^y (1-x)^{n-y} \frac{n-y}{1-x}  
    &= \sum_{y=0}^{n-1} f(y) \binom{n}{y} x^{y+1} (1-x)^{n-(y+1)} \frac{n-y}{x} \\
    &= \sum_{y=0}^{n-1} f(y) \binom{n}{y+1} x^{y+1} (1-x)^{n-(y+1)} \frac{y+1}{x} \label{eq:binomchange1} \\
    &= \sum_{y=1}^{n} f(y-1) \binom{n}{y} x^{y} (1-x)^{n-y} \frac{y}{x} \\
    &= \frac{1}{x}\expcnd{Yf(Y-1)}{X=x} \label{eq:first_expcnd_rel}
\end{align}
where \eqref{eq:binomchange1} follows from $\binom{n}{y} = \binom{n}{y+1}\frac{y+1}{n-y}$ for $y \ne n$. Using \eqref{eq:first_expcnd_rel} into \eqref{eq:diff_expcnd} proves \eqref{eq:der_cond_exp_1}.

Now consider the term
\begin{align}
    \sum_{y=0}^n f(y) \binom{n}{y} x^y (1-x)^{n-y} \frac{y}{x}  &= \sum_{y=1}^{n} f(y) \binom{n}{y} x^{y-1} (1-x)^{n-(y-1)} \frac{y}{1-x} \\
    &= \sum_{y=1}^{n} f(y) \binom{n}{y-1} x^{y-1} (1-x)^{n-(y-1)} \frac{n-y+1}{1-x} \label{eq:binomchange2} \\
    &= \sum_{y=0}^{n-1} f(y+1) \binom{n}{y} x^{y} (1-x)^{n-y} \frac{n-y}{1-x} \\
    &= \frac{1}{1-x}\expcnd{(n-Y)f(Y+1)}{X=x} \label{eq:second_expcnd_rel}
\end{align}
where \eqref{eq:binomchange2} follows from $\binom{n}{y} = \binom{n}{y-1}\frac{n-y+1}{y}$ for $y \ne 0$. Using \eqref{eq:second_expcnd_rel} into \eqref{eq:diff_expcnd} proves \eqref{eq:der_cond_exp_2}.
\end{proof}

\begin{prop}\label{prop:der_info_den}
For $n\ge 1$ and $x \in (0,1)$, we have
\begin{align}
    &i'(x;P_Y)  =n\log\left(\frac{x}{1-x}\right) + \frac{1}{1-x} \expcnd{(n-Y)\log\frac{\expcnd{1-X}{Y=Y+1}}{\expcnd{X}{Y=Y}}}{X=x}.
\end{align}
For $n\ge 2$ and $x \in (0,1)$, we have
\begin{align}
   i'(x;P_Y) &= n\log\left(\frac{x}{1-x}\right)+ \hspace{-0.1cm}n \expcndn{n-1}{\log\frac{\expcndn{n-1}{1-X}{Y}}{\expcndn{n-1}{X}{Y}}}{X=x} \label{eq:inf_den_derivative_no_bregman} \\
    &= \frac{n \, \expcndn{n-1}{\ell_b(x,\expcndn{n-1}{X}{Y})+\frac{x-\expcndn{n-1}{X}{Y}}{1-\expcndn{n-1}{X}{Y}}}{X=x}  }{x} \label{eq:inf_den_derivative_with_bregman}
\end{align}
and $i''(x;P_Y)$ is 
\begin{align}
    i''(x;P_Y) 
    &= \frac{n}{x(1-x)}+\frac{\expcnd{(n-Y)(n-Y-1)\log\frac{\expcnd{X}{Y=Y}}{\expcnd{1-X}{Y=Y+1}}\frac{\expcnd{1-X}{Y=Y+2}}{\expcnd{X}{Y=Y+1}}}{X=x} }{(1-x)^2}  \label{eq:i_doubleprime} \\
    &= \frac{n \left(1+\expcndn{n-1}{Y\log\frac{\expcndn{n-1}{1-X}{Y} }{\expcndn{n-1}{X}{Y}}}{X=x} \right)-\frac{n-1}{1-x} \left(i'(x;P_Y)-n\log\left(\frac{x}{1-x}\right) \right) }{x(1-x)}. \label{eq:i''_versions}
\end{align}
For $n\ge 3$ and $x \in (0,1)$, we have $i''(x;P_Y)$ is 
\begin{align}
    i''(x;P_Y) &= \frac{n}{x(1-x)}+n(n-1) \expcndn{n-2}{\log\frac{\expcndn{n-1}{1-X}{Y+1}}{\expcndn{n-1}{X}{Y+1}}\frac{\expcndn{n-1}{X}{Y}}{\expcndn{n-1}{1-X}{Y}}}{X=x} \\
    &=\frac{n}{x(1-x)}\left(1+(n-1)\expcndn{n-2}{ \ell_b(x, \expcndn{n-1}{X}{Y+1}) + \frac{x-\expcndn{n-1}{X}{Y+1}}{1-\expcndn{n-1}{X}{Y+1}} }{X=x} \right) \notag \\
    & \quad -\frac{n-1}{1-x} i'(x;P_Y). \label{eq:i''_versions_n>3}
\end{align}
\end{prop}
\begin{proof}
Let us start from the expression
    \begin{align}        i(x;P_{Y})&=\expcnd{\log\frac{\binom{n}{Y}}{P_Y(Y)} }{X=x}+nx\log(x) + n(1-x)\log(1-x). 
    \end{align}
    A way to compute the first derivative is as follows:
    \begin{align}
        i'(x;P_{Y}) &=\frac{\rmd}{\rmd x}\expcnd{\log\frac{\binom{n}{Y}}{P_Y(Y)}  }{X=x}+n\log\left(\frac{x}{1-x}\right) \\
        &= \frac{\expcnd{(n-Y)\log\frac{\binom{n}{Y+1}}{P_Y(Y+1)}\frac{P_Y(Y)}{\binom{n}{Y}}}{X=x}}{1-x}  +n\log\left(\frac{x}{1-x}\right) \label{eq:use_last_identity_der_lemma} \\
        &= \frac{\expcnd{(n-Y)\log\frac{\expcnd{1-X}{Y=Y+1}}{\expcnd{X}{Y=Y}}}{X=x}}{1-x} +n\log\left(\frac{x}{1-x}\right), \label{eq:i_prime_first_ver}
    \end{align}
    where~\eqref{eq:use_last_identity_der_lemma} follows from identity~\eqref{eq:der_cond_exp_2} of Lemma~\ref{lem:derivative_exp_cnd}; and the last step follows from identity~\eqref{eq:ratio_cond_exp_n} of Lemma~\ref{lem:channel_transformations}.
    
    An alternative expression can be derived as follows:
    \begin{align}
        i'(x;P_{Y}) &=\frac{\rmd}{\rmd x}\expcnd{\log\frac{\binom{n}{Y}}{P_Y(Y)}  }{X=x}+n\log\left(\frac{x}{1-x}\right) \\
        &= n\expcndn{n-1}{\log\frac{\binom{n}{Y+1}P_Y(Y)}{\binom{n}{Y}P_Y(Y+1)}  }{X=x}+n\log\left(\frac{x}{1-x}\right) \label{eq:apply_derivative_downgrade} \\
        &= n\expcndn{n-1}{\log\frac{\expcndn{n-1}{1-X}{Y} }{\expcndn{n-1}{X}{Y}}  }{X=x}+n\log\left(\frac{x}{1-x}\right)
    \end{align}
    where~\eqref{eq:apply_derivative_downgrade} follows from Lemma~\ref{lem:derivative_exp_cnd}, and the last step follows from Lemma~\ref{lem:channel_transformations}.
    To show~\eqref{eq:inf_den_derivative_with_bregman}, just notice that the Bregman divergence for the binomial channel is
    \begin{align} \label{eq:Bregman_divergence_Binomial}
        \ell_b(x, \expcndn{n-1}{X}{Y}) 
        &= x \log\left(\frac{x \expcndn{n-1}{1-X}{Y}}{(1-x)\expcndn{n-1}{X}{Y}}\right)-\frac{x-\expcndn{n-1}{X}{Y}}{1-\expcndn{n-1}{X}{Y}}.
    \end{align}
    The computation of the second derivative is
    \begin{align}
    &i''(x;P_Y) = \frac{n}{x(1-x)}+\frac{1}{(1-x)^2}  \expcnd{(n-Y)\log\frac{\expcnd{1-X}{Y=Y+1}}{\expcnd{X}{Y=Y}}}{X=x} \nonumber\\
    &\ \frac{1}{(1-x)^2}  \expcnd{(n-Y)\left((n-Y-1)\log\frac{\expcnd{1-X}{Y=Y+2}}{\expcnd{X}{Y=Y+1}}-(n-Y)\log\frac{\expcnd{1-X}{Y=Y+1}}{\expcnd{X}{Y=Y}}\right)}{X=x} \label{eq:use_iprime_first_ver} \\
    &\ =\frac{n}{x(1-x)}+\frac{1}{(1-x)^2}  \expcnd{(n-Y)(n-Y-1)\log\frac{\expcnd{X}{Y=Y}}{\expcnd{1-X}{Y=Y+1}}\frac{\expcnd{1-X}{Y=Y+2}}{\expcnd{X}{Y=Y+1}}}{X=x}, \label{eq:sec_der_comp_proof}
    \end{align}
    where in~\eqref{eq:use_iprime_first_ver} we used~\eqref{eq:i_prime_first_ver} and identity~\eqref{eq:der_cond_exp_2} of Lemma~\ref{lem:derivative_exp_cnd}.

    An alternative formulation of the second derivative is as follows:
    \begin{align}
    i''(x;P_Y) 
    &= \frac{n}{x(1-x)}+\frac{\rmd}{\rmd x}n\expcndn{n-1}{\log\frac{\expcndn{n-1}{1-X}{Y} }{\expcndn{n-1}{X}{Y}}  }{X=x}\\
    &=\frac{n}{x(1-x)}+n(n-1)\bbE^{n-2} \left[\log\frac{\expcndn{n-1}{1-X}{Y+1} }{\expcndn{n-1}{X}{Y+1}} -\log\frac{\expcndn{n-1}{1-X}{Y} }{\expcndn{n-1}{X}{Y}} \mid X =x  \right]
\end{align}
where we applied identity~\eqref{eq:der_cond_exp_n-1} of Lemma~\ref{lem:derivative_exp_cnd}.
Another alternative expression for the second derivative, that is written in terms of the first derivative, is as follows
\begin{align}
    &i''(x;P_Y) =\frac{n}{x(1-x)}+ \frac{\rmd}{\rmd x}n\expcndn{n-1}{\log\frac{\expcndn{n-1}{1-X}{Y} }{\expcndn{n-1}{X}{Y}}  }{X=x}\\
    &=\frac{n}{x(1-x)} + n\left(\frac{1}{x}+\frac{1}{1-x}\right)\expcndn{n-1}{Y\log\frac{\expcndn{n-1}{1-X}{Y} }{\expcndn{n-1}{X}{Y}}  }{X=x}\nonumber \\
    &\quad-\frac{n-1}{1-x}n\expcndn{n-1}{\log\frac{\expcndn{n-1}{1-X}{Y} }{\expcndn{n-1}{X}{Y}}  }{X=x} \label{eq:apply_der_exp_cond_2}\\
    &= \frac{n \left(1+\expcndn{n-1}{Y\log\frac{\expcndn{n-1}{1-X}{Y} }{\expcndn{n-1}{X}{Y}}}{X=x} \right)}{x(1-x)} -\frac{n-1}{1-x} \left(i'(x;P_Y)-n\log\left(\frac{x}{1-x}\right) \right) \label{eq:use_first_derivative} \\
    &= \frac{n \left(1+(n-1)x\expcndn{n-2}{\log\frac{\expcndn{n-1}{1-X}{Y+1} }{\expcndn{n-1}{X}{Y+1}}}{X=x} \right)}{x(1-x)} -\frac{n-1}{1-x} \left(i'(x;P_Y)-n\log\left(\frac{x}{1-x}\right) \right) \label{eq:use_channel_downgrade} \\
    &= \frac{n \left(1+(n-1)x\expcndn{n-2}{\log\frac{x\expcndn{n-1}{1-X}{Y+1} }{(1-x)\expcndn{n-1}{X}{Y+1}}}{X=x} \right)}{x(1-x)} -\frac{n-1}{1-x} i'(x;P_Y)  \\
    &= \frac{n}{x(1-x)} \Bigg(1+(n-1)\bbE^{n-2} \Big [ \ell_b(x, \expcndn{n-1}{X}{Y+1}) + \left. \frac{x-\expcndn{n-1}{X}{Y+1}}{1-\expcndn{n-1}{X}{Y+1}} \mid {X=x} \right] \Bigg)-\frac{n-1}{1-x} i'(x;P_Y)
\end{align}
where~\eqref{eq:apply_der_exp_cond_2} follows from identity~\eqref{eq:der_cond_exp_0}; in~\eqref{eq:use_first_derivative} we used result~\eqref{eq:inf_den_derivative_no_bregman}; in~\eqref{eq:use_channel_downgrade} we made a change of measure by using identity~\eqref{eq:tr_ch_1} of Lemma~\ref{lem:channel_transformations}; and in the last step we used the Bregman divergence for the Binomial channel~\eqref{eq:Bregman_divergence_Binomial}.
\end{proof}

\section{Proof of Proposition~\ref{prop:properties_Y_r} }
\label{app:prop:properties_Y_r}

Under the reference input $X_r$, the output $Y_r$ follows the beta-binomial distribution
\begin{align}
    q_y
    =
    P_{Y_r}(y)
    =
    \frac{
    \Gamma\!\left(y+\tfrac{1}{2}\right)
    \Gamma\!\left(n-y+\tfrac{1}{2}\right)
    }
    {
    \pi\,\Gamma(y+1)\Gamma(n-y+1)
    },
    \qquad y\in \{0\} \cup [n].
    \label{eq:Yr_beta_binomial_proof}
\end{align}
Using the Gamma recurrence $\Gamma(z+1)=z\Gamma(z)$, we obtain
\begin{align}
    \frac{q_{y+1}}{q_y}
    &=
    \frac{
    \Gamma\!\left(y+\tfrac{3}{2}\right)
    \Gamma\!\left(n-y-\tfrac{1}{2}\right)
    }
    {
    \Gamma(y+2)\Gamma(n-y)
    }
    \frac{
    \Gamma(y+1)\Gamma(n-y+1)
    }
    {
    \Gamma\!\left(y+\tfrac{1}{2}\right)
    \Gamma\!\left(n-y+\tfrac{1}{2}\right)
    } =
    \frac{y+\tfrac12}{y+1}
    \cdot
    \frac{n-y}{n-y-\tfrac12},
    \qquad y\in \{0\} \cup [n-1].
\end{align}
Similarly, for $y\in[n]$, we have $\frac{q_{y-1}}{q_y}
    =
    \frac{y}{y-1/2}
    \cdot
    \frac{n-y+1/2}{n-y+1}$. 
    Moreover, the ratio in
\eqref{eq:qstar_ratio} satisfies
\begin{align}
    &\frac{q_{y+1}}{q_y}\le 1
    \quad\Longleftrightarrow\quad
    \left(y+\tfrac12\right)(n-y)
    \le
    (y+1)\left(n-y-\tfrac12\right) 
    \quad\Longleftrightarrow\quad
    y\le \frac{n-1}{2}.
\end{align}
Thus $\{q_y\}_{y=0}^n$ decreases up to the center and then increases.
Hence, for every $y\in\{0\} \cup [n]$, $q_y \ge q_{\lfloor n/2\rfloor}$.

It remains to lower-bound the value at the center. From
\eqref{eq:Yr_beta_binomial_proof},
\begin{equation}
    q_{\lfloor n/2\rfloor}
    =
    \frac{1}{\pi}
    \frac{
    \Gamma\!\left(\lfloor n/2\rfloor+\tfrac12\right)
    }{
    \Gamma\!\left(\lfloor n/2\rfloor+1\right)
    }
    \frac{
    \Gamma\!\left(\lceil n/2\rceil+\tfrac12\right)
    }{
    \Gamma\!\left(\lceil n/2\rceil+1\right)
    } .
\end{equation}
Using Wendel's Gamma-ratio inequality \cite[Eq.~(4)]{Wendel1948}, namely
\begin{equation}
        \frac{\Gamma(x+s)}{\Gamma(x)}
    \le x^s,
    \qquad x>0,\quad 0<s<1,
\end{equation}
with $x=m+\frac12$ and $s=\frac12$, we obtain $\frac{\Gamma(m+1)}{\Gamma(m+\frac12)}
    \le
    \sqrt{m+\frac12}$. Hence,
\begin{equation}
    \frac{\Gamma(m+\frac12)}{\Gamma(m+1)}
    \ge
    \frac{1}{\sqrt{m+\frac12}}
    \ge
    \frac{1}{\sqrt{m+1}},
    \qquad m\ge 0, \label{eq:lowerbound_ratio_of_gammas}
\end{equation}
thus we obtain
\begin{equation}
    q_{\lfloor n/2\rfloor}
    \ge
    \frac{1}{
    \pi
    \sqrt{
    \bigl(\lfloor n/2\rfloor+1\bigr)
    \bigl(\lceil n/2\rceil+1\bigr)
    }} \ge
    \frac{2}{\pi(n+2)},
\end{equation}
where the last step follows from $\sqrt{ab}\le (a+b)/2$, with $a=\lfloor n/2\rfloor+1$, $b=\lceil n/2\rceil+1$, and $a+b=n+2$. This proves \eqref{eq:qstar_min_bound}. 

By Kershaw's Gamma-ratio inequality~\cite[Eq.~(1.3)]{kershaw1983some}, for $x>0$ and
$0<s<1$,
\begin{equation}
    \frac{\Gamma(x+1)}{\Gamma(x+s)}
    >
    \left(x+\frac{s}{2}\right)^{1-s}.
\end{equation}
Taking \(s=\frac12\) and \(x=y\), we obtain, for \(y>0\),
\begin{equation}
    \frac{\Gamma(y+1)}{\Gamma(y+\frac12)}
    >
    \left(y+\frac14\right)^{1/2}.
\end{equation}
Equivalently,
\begin{equation}
    \frac{\Gamma(y+\frac12)}{\Gamma(y+1)}
    <
    \left(y+\frac14\right)^{-1/2},
    \qquad y>0.
\end{equation}
The same bound also holds at \(y=0\), since
\(\Gamma(1/2)=\sqrt{\pi}<2=(1/4)^{-1/2}\). Therefore,
\begin{equation}
    \frac{\Gamma(y+\frac12)}{\Gamma(y+1)}
    \le
    \left(y+\frac14\right)^{-1/2},
    \qquad y\ge 0.
\end{equation}
Applying this bound to both Gamma ratios in
\eqref{eq:Yr_beta_binomial_proof} gives, for every $y\in\{0\} \cup [n]$,
\begin{equation}
    q_y
    \le
    \frac{1}{
    \pi\sqrt{
    \left(y+\tfrac{1}{4}\right)
    \left(n-y+\tfrac{1}{4}\right)
    }} .
\end{equation}
Similarly, by applying the first bound in \eqref{eq:lowerbound_ratio_of_gammas} to ratios in
\eqref{eq:Yr_beta_binomial_proof}, we arrive at
\begin{equation}
    q_y
    \geq
    \frac{1}{
        \pi
        \sqrt{
            (y+\frac12)(n-y+\frac12)
        }
    }.
\end{equation} 
This completes the proof.

\section{Capacity Computation for $n \le 3$}
\label{sec:computations_of_cap_exact}

\subsection{The Case of $n=1$}
Follows immediately from Proposition~\ref{prop:loc_info}.

\subsection{The Case of $n=2$}
From Proposition~\ref{prop:loc_info}, for $n=2$, we infer that $\supp(P_{X^\star}) \subseteq  \left\{0, \frac{1}{2},1 \right\}.$

Now let $p= P_{X^\star}(\frac{1}{2})$.  Using Corollary~\ref{cor:rel_cap_probY} and the equations for $P_{Y^\star}$, we have that 
\begin{equation}
    P_{Y^\star}(0) = P_{Y^\star}(2) = \rme^{-C(2)}, \qquad P_{Y^\star}(1) = p 2\frac{1}{2}\left(1-\frac{1}{2}\right) = \frac{p}{2}.
\end{equation}
From $\sum_{y=0}^2 P_{Y^\star}(y)=1$ it follows that:
\begin{equation}
    2 \rme^{-C(2)}+\frac{p}{2} = 1
\end{equation}
or $p = 2(1-2\rme^{-C(2)})$. From the KKT equality condition in \eqref{eq:KKT_equality}, we have that
\begin{align}
    C(2) &= i \left(\frac{1}{2};P_{Y^\star} \right) = \sum_{y=0}^2 \binom{2}{y} \frac{1}{2^y} \left(1-\frac{1}{2} \right)^{2-y} \log\frac{\binom{2}{y} \frac{1}{2^y} \left(1-\frac{1}{2} \right)^{2-y}}{P_{Y^\star}(y)} \\
    &= \frac{1}{4} \left(\log\frac{1}{4}+C(2) \right) + \frac{1}{2} \log\frac{1}{2(1-2\rme^{-C(2)})}  + \frac{1}{4} \left(\log\frac{1}{4}+C(2) \right)
\end{align}
that can be rewritten as $C(2) = \log\frac{1}{4}+\log\frac{1}{2(1-2 \rme^{-C(2)})}$,
whose solution is $C(2) = \log\frac{17}{8}$. We also have $p = \frac{2}{17}$ and $P_{X^\star}(0) = P_{X^\star}(1) = \frac{15}{34}$.

\subsection{The Case of $n=3$}

From Proposition~\ref{prop:loc_info} and Proposition~\ref{thm:n/2 bound}, for $n=3$, we infer that $\supp(P_{X^\star}) \subseteq  \left\{0, \frac{1}{2},1 \right\}.$

Now, let $p= P_{X^\star}(\frac{1}{2})$. Corollary~\ref{cor:rel_cap_probY}  and direct computations imply that   
\begin{align}
P_{Y^\star}(0) &=  P_{Y^\star}(3) =\rme^{-C(3)} ,\qquad P_{Y^\star}(1) =  P_{Y^\star}(2) = \frac{3}{8} p. 
\end{align}
Now using above and the fact that $\sum_{y=0}^3 P_{Y^*}(y)=1$, we have that 
\begin{equation}
    p = \frac{4}{3} \left(1 -2 \rme^{-C(3)} \right).  \label{eq:expression_for_p}
\end{equation}

Next, it can be shown that $i \left(\frac{1}{2};P_{Y^\star} \right)  = \frac{1}{4} \log \left(  \frac{\rme^{C(3)}}{8p^3} \right)$. From the KKT equality condition in  \eqref{eq:KKT_equality}, we have that 
\begin{equation}
  C(3)=  i \left(\frac{1}{2};P_{Y^\star} \right)  = \frac{1}{4} \log \left(  \frac{\rme^{C(3)}}{8p^3} \right)  
\end{equation}
using the expression for $p$ in \eqref{eq:expression_for_p} and simplifying, we arrive at
\begin{equation}
    C(3) = \log \left(  \frac{1}{ \frac{8}{3} \left( 1- 2\rme^{-C(3)} \right)}\right). 
\end{equation}
Solving for $C(3)$ we arrive at $C(3) =\log \left( \frac{19}{8} \right)$. We also have that 
\begin{align}
P_{Y^\star}(0) &=  P_{Y^\star}(3) =\frac{8}{19} ,\quad P_{Y^\star}(1) &=  P_{Y^\star}(2) = \frac{3}{38}, \quad  P_{X^*}(0) &= P_{X^*}(1) =\frac{15}{38},\quad P_{X^*} \left( \frac{1}{2} \right) &= \frac{4}{19}.
\end{align}

\section{Bounds on the Entropy of a Binomial Random Variable}\label{app:bound_binom_entropy}
First, we need the following result.
\begin{lemma}\label{lem:expect_log_binomial}
    Let $P_{Y|X}(\cdot|x)$ be a binomial distribution with $n$ trials and success probability $x$ per trial. Then,
    \begin{align}\label{eq:log_expect_binomial}
    &\expcnd{\mathbbm{1}(0<Y\le n)\log\left(\frac{Y}{n}\right)}{X=x} \ge   (1-(1-x)^n)\log(x)-1.
    \end{align}
\end{lemma}
\begin{proof}
    Inspired by the approach of~\cite[Appendix~B]{lapidoth2008capacity}, we bound the expectation as follows:
\begin{align}
    &\expcnd{\mathbbm{1}(0<Y\le n)\log\left(\frac{Y}{n}\right)}{X=x} \notag\\
    &\quad =\expcnd{\mathbbm{1}(0<Y\le n)\log\left(x\right)}{X=x} +\expcnd{\mathbbm{1}(0<Y\le n)\log\left(\frac{Y}{nx}\right)}{X=x} \\
    &\quad =(1-(1-x)^n)\log\left(x\right)  +\expcnd{\mathbbm{1}(0<Y\le n)\log\left(\frac{Y}{nx}\right)}{X=x} \\
    &\quad =(1-(1-x)^n)\log\left(x\right)+\sum_{y=1}^{n-1} P_{Y|X}(y|x) \log\left(\frac{y}{nx} \right) \\
    &\quad \ge (1-(1-x)^n)\log\left(x\right)+\int_0^n P_{Y|X}(\lfloor y \rfloor|x) \log\left(\frac{y}{nx} \right) \rmd y \\
    &\quad = (1-(1-x)^n)\log\left(x\right)+n\int_0^1 P_{Y|X}(\lfloor nt \rfloor|x) \log\left(\frac{t}{x} \right) \rmd t  \label{eq:exp_log_two_integrals}
\end{align}
where the inequality holds because $x\mapsto \log(x)$ is an increasing function and negative for $x \in (0,1)$.

Now introduce a continuous variable $Z$ with density $f_{Z}(z) = n P_{Y|X}(\lfloor nz \rfloor|x)$ for $z \in [0,1]$. Then, the integral in~\eqref{eq:exp_log_two_integrals} becomes:
\begin{align}
     &n\int_0^{1} P_{Y|X}(\lfloor nt \rfloor|x) \log\left(\frac{t}{x} \right) \rmd t = \int_0^{1} f_{Z}(t) \log\left(\frac{t}{x} \right) \rmd t =\int_0^{x} f_{Z}(t) \log\left(\frac{t}{x} \right) \rmd t+\int_x^{1} f_{Z}(t) \log\left(\frac{t}{x} \right) \rmd t.
\end{align}
Let us now bound the two integrals separately. For the first integral, by integrating by parts we have
\begin{align}
     \int_0^{x} f_{Z}(t) \log\left(\frac{t}{x} \right) \rmd t &= \left[\Pr(Z \le t) \log\left(\frac{t}{x}\right) \right]_0^{x}-\int_0^{x} \Pr(Z \le t) \frac{1}{t} \rmd t \\
     &\ge-\int_0^{x} \int_0^{t}n P_{Y|X}(\lfloor nz \rfloor|x) \rmd z \frac{1}{t} \rmd t \label{eq:bound_on_cdf} \\
     &\ge -\int_0^{x} n P_{Y|X}(\lfloor nt \rfloor|x)   \rmd t \label{eq:use_increasing_part_of_binomial} \\
     &= -\int_0^{x} f_Z(t)   \rmd t \\
     &\ge -1
\end{align}
where in~\eqref{eq:use_increasing_part_of_binomial} we used that 
$ \int_0^{t}n P_{Y|X}(\lfloor nz \rfloor|x) \rmd z\le tn P_{Y|X}(\lfloor nt \rfloor|x)$  thanks to the following lemma and to $t \le x$:
\begin{lemma}
    Let $P_{Y|X}$ be a binomial distribution. Then, $y\mapsto P_{Y|X}(y|x)$ is increasing for $y \le \lfloor (n+1)x \rfloor$, and decreasing for $y \ge \lceil (n+1)x \rceil$.
\end{lemma}
\begin{proof}
    From the ratio $\frac{P_{Y|X}(y|x)}{P_{Y|X}(y-1|x)} = \frac{n-y+1}{y} \frac{x}{1-x}$
    we see that the condition $P_{Y|X}(y|x) \ge P_{Y|X}(y-1|x)$ is satisfied for $y \le \lfloor (n+1)x \rfloor$.
\end{proof}
For the second integral, write
\begin{align}
    \int_{x}^1 f_Z(t) \log\left(\frac{t}{x} \right) \rmd t \ge 0.
\end{align}
Putting together the two results, we get the result in~\eqref{eq:log_expect_binomial}.
\end{proof}

We are now ready to give the main result of this appendix.
\begin{lemma}
    For $x \in [0,1]$, the entropy of a binomial distribution is bounded as follows
    \begin{equation}
        H(Y|X=x) \le \frac{1}{2}\log\left(2\pi \rme \left(n x(1-x)+\frac{1}{12}\right)\right),
    \end{equation}
    \begin{align}
        H(Y|X=x)&\ge (1-(1-x)^n-x^n)\frac{1}{2}\log\left( 2\pi n\right) +\frac{1}{2}(1-(1-x)^n)\log(x) +\frac{1}{2}(1-x^n)\log(1-x)-1.
    \end{align}
\end{lemma}
\begin{proof}
    For the upper bound, write
    \begin{align}
        H(Y|X=x) &= h(Y+U|X=x) \le \frac{1}{2}\log\left(2\pi \rme \left(n x(1-x)+\frac{1}{12}\right)\right),
    \end{align}
    where the first equality follows from \cite[Lemma 17]{lapidoth2008capacity} where $U\sim {\cal U}[0,1]$ is independent of $Y$; and the last step follows from the Gaussian maximizes entropy principle. 

Next we prove the lower bound. First of all, compute
\begin{align}
    -H(Y|X=x) &= \expcnd{\log\left(\binom{n}{Y}x^{Y}(1-x)^{n-Y}\right)}{X=x} \\
    &= \expcnd{\log\binom{n}{Y}}{X=x} + nx\log(x)  +n(1-x)\log(1-x) \\
    &\le \expcnd{\log\binom{n}{Y}}{X=x}-nH_2(x)
\end{align}
By using the bound $\binom{n}{Y}\le \sqrt{\frac{n}{2\pi Y(n-Y)}}\rme^{n H_2(\frac{Y}{n})}$ for $0<Y<n$ (see, e.g., \cite[Problem~5.8]{gallager1968information}), we can write:
\begin{align}
    & \expcnd{\log\binom{n}{Y}}{X=x} \notag \\
    & =\expcnd{\mathbbm{1}(0<Y<n)\log\binom{n}{Y}}{X=x} \\
    & \le (1-(1-x)^n-x^n)\frac{1}{2}\log\left( \frac{n}{2\pi}\right) -\frac{1}{2}\expcnd{\mathbbm{1}(0<Y<n)\log(Y(n-Y))}{X=x} +n\expcnd{H_2(\frac{Y}{n})}{X=x} \label{eq:bound_bin_coeff} \\
    &= -(1-(1-x)^n-x^n)\frac{1}{2}\log\left(2\pi n\right) -\frac{1}{2}\expcnd{\mathbbm{1}(0<Y<n)\log\left(\frac{Y}{n}\frac{n-Y}{n}\right)}{X=x} +n\expcnd{H_2 \left(\frac{Y}{n} \right)}{X=x} \\
    &= -(1-(1-x)^n-x^n)\frac{1}{2}\log\left( 2\pi n\right) -\frac{1}{2}\expcnd{\mathbbm{1}(0<Y\le n)\log\left(\frac{Y}{n}\right)}{X=x} \nonumber\\
    &\quad-\frac{1}{2}\expcnd{\mathbbm{1}(0< Y\le n)\log\left(\frac{Y}{n}\right)}{X=1-x} +n\expcnd{H_2(\frac{Y}{n})}{X=x}\label{eq:use_symmetry_bino_pmf} \\
    &\le -(1-(1-x)^n-x^n)\frac{1}{2}\log\left( 2\pi n\right)  -\frac{1}{2}\expcnd{\mathbbm{1}(0<Y\le n)\log\left(\frac{Y}{n}\right)}{X=x}\nonumber\\
    &\quad-\frac{1}{2}\expcnd{\mathbbm{1}(0< Y\le n)\log\left(\frac{Y}{n}\right)}{X=1-x}+nH_2(x), \label{eq:exp_log_bin_coeff}
\end{align}
where in~\eqref{eq:use_symmetry_bino_pmf} we used the channel symmetry $P_{Y|X}(y|x) = P_{Y|X}(n-y|1-x)$; and in the last step we used Jensen's inequality and $\expcnd{Y}{X=x}=nx$.

By using Lemma~\ref{lem:expect_log_binomial}, we have
\begin{align}
    &\expcnd{\mathbbm{1}(0<Y\le n)\log\left(\frac{Y}{n}\right)}{X=x}  \ge (1-(1-x)^n)\log(x)-1
\end{align}
and
\begin{align}
   & \expcnd{\mathbbm{1}(0<Y\le n)\log\left(\frac{Y}{n}\right)}{X=1-x}  \ge (1-x^n)\log(1-x)-1.
\end{align}
Therefore, we have
\begin{align}
    &\expcnd{\log\binom{n}{Y}}{X=x} \notag \\
    & \ \le -(1-(1-x)^n-x^n)\frac{1}{2}\log\left( 2\pi n\right)  -\frac{1}{2}(1-(1-x)^n)\log(x)-\frac{1}{2}(1-x^n)\log(1-x) +1+nH_2(x)
\end{align}
and 
\begin{align}
    &-H(Y|X=x) \le -(1-(1-x)^n-x^n)\frac{1}{2}\log\left( 2\pi n\right)  -\frac{1}{2}(1-(1-x)^n)\log(x)-\frac{1}{2}(1-x^n)\log(1-x)+1.
\end{align}
\end{proof}
\section{Proof of Theorem~\ref{thm:capacity-lb}}
\label{sec:proof_thm:capacity-lb}

\begin{IEEEproof}
Let $X_r \sim \mathrm{Beta}\!\left(\frac{1}{2},\frac{1}{2}\right)$, with density as in
\eqref{eq:pdf_beta}, and let $Y_r$ be the output induced by $X_r$ through the
binomial channel. Then,
\begin{equation}
    C(n) \ge I(X_r;Y_r)=H(Y_r)-H(Y_r|X_r).
\end{equation}
We bound the two entropy terms separately.

First, for the conditional entropy term, we have
\begin{align}
    H(Y_r|X_r)
    &\le
    \frac{1}{2}\log(2\pi e)
    +
    \frac{1}{2}
    \mathbb{E}\left[
        \log\left(nX_r(1-X_r)+\frac{1}{12}\right)
    \right] =
    \frac{1}{2}\log(2\pi e)
    +
    \log\!\left(
        \frac{1+\sqrt{3n+1}}{4\sqrt{3}}
    \right),
\end{align}
where the first inequality follows from
\cite[Lemma~10]{zieder2024binomial}. The last equality follows by evaluating
the expectation in closed form. Indeed, using the substitution
$x=\sin^2(t/2)$ and the integral identity in~\cite[Sec.~4.226]{Ryzhik}:
\begin{equation}
    \int_0^{\pi/2} \log(a+b\sin^2 t)\,\mathrm{d}t
    =
    \pi \log\!\left(
        \frac{\sqrt{a}+\sqrt{a+b}}{2}
    \right), \qquad a,b>0,
 \end{equation}
gives the desired expression.

For the output entropy term, using the tail bound in \eqref{eq:tail_bound}, we obtain
\begin{align}
    H(Y_r)
    &=
    \mathbb{E}\left[\log\frac{1}{P_{Y_r}(Y_r)}\right] \\
    &\ge
    \log \pi
    +
    \frac{1}{2}
    \mathbb{E}\left[
        \log\!\left(Y_r+\frac{1}{4}\right)
        +
        \log\!\left(n-Y_r+\frac{1}{4}\right)
    \right] \\
    &=
    \log \pi
    +
    \mathbb{E}\left[
        \log\!\left(Y_r+\frac{1}{4}\right)
    \right] \label{eq:Yr_symmetry_entropy_lb} \\
    &\ge
    \log \pi
    +
    \mathbb{E}\left[
        \psi\!\left(Y_r+\frac{1}{2}\right)
    \right] \label{eq:log_psi} \\
    &=
    \log\!\left(\frac{\pi}{4}\right)+\psi(n+1),
    \label{eq:betabin_identity}
\end{align}
where  \eqref{eq:Yr_symmetry_entropy_lb} follows from the symmetry of
$Y_r$, namely $Y_r\stackrel{d}{=}n-Y_r$; \eqref{eq:log_psi},  follows from 
\begin{equation}
\label{digamma}
    \log\!\left(y+\frac{1}{4}\right)
    \ge
    \psi\!\left(y+\frac{1}{2}\right),
    \qquad y\ge 0,
\end{equation}
where  $\psi$ denotes the digamma function and  \eqref{digamma} follows from the standard digamma bound \(\psi(x)\le \log(x-\frac14)\) for \(x\ge \frac12\)
\cite{Alzer1997GammaPsi}, applied with \(x=y+\frac12\).
; and 
\eqref{eq:betabin_identity}, follows from using the beta-binomial identity which is derived by differentiating the
beta-binomial normalization identity with respect to \(\alpha\): 
\begin{equation}
\mathbb{E}\left[\psi(Y_{\alpha,\beta}+\alpha)\right]
    =
    \psi(n+\alpha+\beta)+\psi(\alpha)-\psi(\alpha+\beta),
\end{equation}
with $\alpha=\beta=\frac{1}{2}$, together with
$\psi(1/2)-\psi(1)=-2\log 2$.
See, e.g., \cite{JohnsonKotzKempUnivariateDiscrete,Minka2000Dirichlet}.

Combining the lower bound on $H(Y_r)$ with the upper bound on
$H(Y_r|X_r)$ gives
\begin{equation}
    C(n)\ge
    \psi(n+1)
    -\log\!\left(1+\sqrt{3n+1}\right)
    +
    \frac{1}{2}\log\!\left(\frac{3\pi}{2e}\right),
\end{equation}
which proves \eqref{eq:capacity-lb}. 
Finally, using the digamma lower bound \cite{Alzer1997GammaPsi}
\begin{equation}
    \psi(n+1)
    \ge
    \log\!\left(n+\frac{1}{3}\right)-\frac{1}{n+1},
\end{equation}
we obtain \eqref{eq:capacity-lb-asymp}, with 
\begin{equation}
    r_{{\rm LB}}(n)
    =
    \frac{1}{2}\log\!\left(1+\frac{1}{3n}\right)
    -
    \log\!\left(1+\frac{1}{\sqrt{3n+1}}\right)
    -
    \frac{1}{n+1}.
\end{equation}
This completes the proof.
\end{IEEEproof}

\section{Proof of Theorem~\ref{thm:capacity-ub}}
\label{sec:proof_thm:capacity-up}
\begin{IEEEproof}
By the dual characterization of capacity,
\begin{equation} \label{eq:dual_capacity}
    C(n) \le \sup_{x \in[0,1]} \kl{P_{Y|X}(\cdot|x)}{  Q_n}
\end{equation}
for any output distribution $Q_n$. Consider the output distribution
\begin{align}\label{eq:Q_n}
    Q_n(y)
    &=
    \eta_n
P_{Y|X}\left(y \, \middle| \, \frac{c_n}{n}\right) 
    +
    \eta_n
P_{Y|X}\left(y \, \middle| \, 1-\frac{c_n}{n}\right) 
    +
    (1-2\eta_n)
    \int_0^1 P_{Y|X}(y|x) f_{X_r}(x)\,\mathrm{d}x,
\end{align}
where $f_{X_r}$ is as in \eqref{eq:pdf_beta}, $\eta_n \ge \left(\frac{2 \mathrm{e}}{n\pi} \right)^{\frac{1}{4}}$, and $c_n \le \frac{1}{2}\log\frac{n\pi}{2 \mathrm{e}}$.
The constants \(\eta_n\) and \(c_n\) are chosen following the boundary-corrected Jeffreys mixture construction of Xie and Barron \cite[Sec.~III.B]{xie1997minimax}. The ordinary Jeffreys mixture, which corresponds here to the arcsine density \(f_{X_r}\), gives the correct interior asymptotics but needs to be modified near the boundary points \(x=0\) and \(x=1\). The parameter \(c_n\) specifies the location of the near-boundary correction points, \(c_n/n\) and \(1-c_n/n\), while \(\eta_n\) specifies the amount of probability mass assigned to these
boundary corrections. The conditions $\eta_n \ge \left(\frac{2 \mathrm e}{n\pi}\right)^{1/4}$ and $c_n \le \frac12 \log\frac{n\pi}{2\mathrm e}$ are sufficient conditions used in \cite[Sec.~III.B]{xie1997minimax} to obtain the uniform redundancy bound in Theorem~\ref{thm:Barron}.
Specifically, 
Xie and Barron have proven the following \cite[Sec.~III.B]{xie1997minimax}.
\begin{theorem}\label{thm:Barron} For the output distribution $Q_n$  in \eqref{eq:Q_n} and for all sufficiently large $n$ (e.g., pick $n\ge 28$) we have
\begin{equation} \label{eq:XieBarron_upper}
    \max_{x \in [0,1]} \kl{P_{Y|X}(\cdot|x) }{ Q_n} \le \frac{1}{2}\log\frac{n}{2\pi \mathrm{e}} +\log\frac{\pi}{1-2\eta_n} +\frac{5}{c_n}.
\end{equation}
\end{theorem}

To make the bound explicit, we choose $\eta_n =
    \left(\frac{2 \mathrm e}{n\pi}\right)^{1/4}$. Note that the quantity $\log(\pi/(1-2\eta_n))$ in the right-hand side of \eqref{eq:XieBarron_upper} is well defined only when $\eta_n<1/2$. With our choice of $\eta_n$, the condition \(\eta_n<1/2\) is equivalent to $\left(\frac{2 \mathrm e}{n\pi}\right)^{1/4}<\frac12$ or, equivalently, $n>\frac{32\mathrm e}{\pi}$. Since $\frac{32\mathrm e}{\pi}<28$, we have $\eta_n<1/2$ for all $n\ge 28$.

Using Theorem~\ref{thm:Barron} with \eqref{eq:dual_capacity} concludes the proof.
\end{IEEEproof}

{
\section{Proofs of Theorem~\ref{thm:C1_UB_finite_n} and Corollary~\ref{cor:C2_UB_finite_n}}
\label{sec:C_UB_finite_n}
Let us begin by proving the results in Theorem~\ref{thm:C1_UB_finite_n}. 
\begin{IEEEproof}
Define the \emph{upper envelope} of the binomial kernel $M_n(y) = \max_{x \in [0,1]} P_{Y \mid X} (y \mid x)$. Given the bound on the dual expression of the capacity in~\eqref{eq:dual_capacity}, let us consider the auxiliary output distribution
\begin{align}
    \tilde{Q}_n(y) = \frac{M_n(y)}{S_n}, \quad y \in \{0\} \cup [n],
\end{align}
where $S_n = \sum_{y=0}^n M_n(y)$. Then, we get
\begin{align}
		\kl{P_{Y\mid X}(y \mid x)}{\tilde{Q}_n(y)} &= \sum_y P_{Y\mid X}(y \mid x) \log  \frac{S_n \cdot P_{Y\mid X}(y \mid x)}{M_n(y)} \\
		& \leq \sum_y P_{Y\mid X}(y \mid x) \log S_n \label{eq:pyx_logSn}\\
		& = \log S_n, \label{eq:logSn}
\end{align}
where~\eqref{eq:pyx_logSn} holds because, by definition, we have $P_{Y\mid X}(y \mid x) \leq M_n(y)$. By evaluating the derivatives with respect to $x$ of $P_{Y \mid X} (y \mid x)$, it is straightforward to see that the binomial kernel is maximized in $x = y/n$. 

By plugging $M_n(y) = P_{Y\mid X}(y \mid x=y/n)$ into~\eqref{eq:logSn}, we get the bound in~\eqref{eq:C1_UB_finite_n}. 
\end{IEEEproof}
Let us now consider the proof of Corollary~\ref{cor:C2_UB_finite_n}.
\begin{IEEEproof}
    By Robbins' bounds on Stirling's formula, we get
	\begin{align}
			\binom{n}{y}
			&\le
			\sqrt{\frac{n}{2\pi y(n-y)}} \frac{n^n}{y^y(n-y)^{n-y}}
			\exp\left(
			\frac{1}{12n}
			-\frac{1}{12y+1}
			-\frac{1}{12(n-y)+1}
			\right) \\
			&\le
			\sqrt{\frac{n}{2\pi y(n-y)}}
			\frac{n^n}{y^y(n-y)^{n-y}}, \label{eq:ub_binom_coeff}
	\end{align}
    where the last inequality is given by bounding the exponential term with $1$.

    Let us exclude the cases $y=0,n$, for which $M_n(0) = M_n(n) = 1$. By applying~\eqref{eq:ub_binom_coeff} to the \emph{upper envelope} $M_n(y) = P_{Y\mid X}(y \mid x= y/n)$, we get
    \begin{align}
        M_n(y) &= \binom{n}{y}
		\frac{y^y(n-y)^{n-y}}{n^n} \leq \sqrt{\frac{n}{2\pi y(n-y)}}. \label{eq:UB_Mn}
    \end{align}
    We want to further upper-bound the sum of $M_n(y)$, for $1\leq y \leq n-1$. 
    Consider $g(x) = 1/\sqrt{x(1-x)}$. Since $g(x)$ is decreasing on $(0,1/2]$ and increasing on $[1/2,1)$, by applying left and right Riemann sum inequalities we get
    \begin{align}
        \frac{1}{n} \sum_{y=1}^{n-1} g(y/n) \leq \int_0^1 g(x) \,  \mathrm{d}x = \pi. \label{eq:int_pi}
    \end{align}
    Then, by~\eqref{eq:UB_Mn} and~\eqref{eq:int_pi}, it holds
    \begin{align}
        \sum_{y=1}^{n-1} M_n(y) \leq \sum_{y=1}^{n-1} \sqrt{\frac{n}{2\pi y(n-y)}} = \sqrt{\frac{n}{2\pi}} \frac{1}{n} \sum_{y=1}^{n-1} g(y/n) \leq \sqrt{\frac{n \pi}{2}}. \label{eq:UB_Sn}
    \end{align}
    Finally, by $M_n(0) = M_n(n) = 1$ and by applying the upper bound~\eqref{eq:UB_Sn} to~\eqref{eq:logSn}, we get the final result~\eqref{eq:C2_UB_finite_n}.
\end{IEEEproof}
}

\section{Proof of Theorem ~\ref{thm:optimal-output}}
\label{proofoftheorem3}

Fix $y\in\{0,\ldots,n\}$ and let $\cA_j \subseteq [0,1]$ be as defined in  Lemma~\ref{lem:reference_partition}, then,  
\begin{align}
    P_{Y^*}(y)
    &=
    \sum_{j=0}^{N-1}
    \int_{\cA_j}
    P_{Y|X}(y|x)\,
    \rmd P_{X^\star}(x)\\
    &\leq
    \sum_{j=0}^{N-1}
    P_{X^\star}(\cA_j)
    \sup_{x\in\cA_j}P_{Y|X}(y|x)\\
    &\leq
    \exp\left(
        \frac{8}{25}-C(n)
    \right) \sum_{j=0}^{N-1} \sup_{x\in\cA_j}P_{Y|X}(y|x), \label{eq:upper_bound_on_P_{Y^*}}
\end{align}
where in \eqref{eq:upper_bound_on_P_{Y^*}} we have used the bound in \eqref{eq:uniform_cell_mass_bound}. 

We now upper bound the sum in \eqref{eq:upper_bound_on_P_{Y^*}}.  
\begin{lemma}\label{lem:upper_sum_partition}
Let $\{\cA_j\}_{j=0}^{N-1}$ be the partition constructed in
Lemma~\ref{lem:reference_partition}.
Then, for every $y\in\{0,\ldots,n\}$,
\begin{align}
    \sum_{j=0}^{N-1}
    \sup_{x\in\cA_j}
    P_{Y|X}(y|x)
    \leq
    \left(
        \frac{5\pi}{2}+6
    \right)\sqrt{n}\,P_{Y_r}(y).
    \label{eq:upper_sum_uniform_constant}
\end{align}
\end{lemma}

\begin{IEEEproof}
For a fixed $y\in\{0,\ldots,n\}$, define
\begin{equation}
    f_y(t)
    =
    P_{Y|X}\left(
        y\middle|F_r^{-1}(t)
    \right),
    \qquad
    t\in[0,1],
\end{equation}
and let
\begin{equation}
    M_n(y)
    =
    \sup_{t\in[0,1]}f_y(t)
    =
    \sup_{x\in[0,1]}P_{Y|X}(y|x).
\end{equation}
Since $x\mapsto P_{Y|X}(y|x)$ is unimodal and $F_r^{-1}$
is increasing, $f_y$ is unimodal. Therefore, its upper Darboux
sum over the uniform partition of $[0,1]$ satisfies
\begin{align}
    \sum_{j=0}^{N-1}
    \sup_{x\in\cA_j}P_{Y|X}(y|x)
    &\leq
    N\int_0^1 f_y(t)\,\rmd t
    +2M_n(y)
    \label{eq:upper_darboux_unimodal}\\
    &=
    N P_{Y_r}(y)+2M_n(y),
    \label{eq:upper_darboux_q}
\end{align}
where \eqref{eq:upper_darboux_unimodal} follows by splitting the
partition at a maximizer of $f_y$ and telescoping the upper-sum
errors over its increasing and decreasing parts; and
\eqref{eq:upper_darboux_q} follows since
$F_r(X_r)\sim\operatorname{Unif}[0,1]$.

We next show that
\begin{equation}
    M_n(y)
    \leq
    2\sqrt{n}\,P_{Y_r}(y),
    \qquad
    y=0,\ldots,n.
    \label{eq:mode_vs_reference}
\end{equation}
Suppose first that $1\leq y\leq n-1$. The maximum is attained
at $x=y/n$, and hence, by \eqref{eq:UB_Mn},
\begin{equation}
    M_n(y)
    \leq
    \sqrt{
        \frac{n}{2\pi y(n-y)}
    }.
    \label{eq:mode_upper_for_sum}
\end{equation}

Combining \eqref{eq:mode_upper_for_sum} and
lower bound in \eqref{eq:tail_bound}, we obtain
\begin{align}
    \frac{M_n(y)}{P_{Y_r}(y)}
    &\leq
    \sqrt{
        \frac{\pi n}{2}
        \left(
            1+\frac{1}{2y}
        \right)
        \left(
            1+\frac{1}{2(n-y)}
        \right)
    }\leq
    \sqrt{\frac{9\pi n}{8}}
    \leq
    2\sqrt{n},
    \label{eq:interior_mode_q_bound}
\end{align}
where the second inequality follows since
$y\geq1$ and $n-y\geq1$. For $y=0$, we have $M_n(0)=1$ and
\begin{equation}
    P_{Y_r}(0)
    =
    \frac{\Gamma(n+\frac12)}
    {\sqrt{\pi}\Gamma(n+1)}
    \geq
    \frac{1}{\sqrt{\pi(n+\frac12)}}.
\end{equation}
For $n\geq2$, this implies
\begin{equation}
    P_{Y_r}(0)
    \geq
    \frac{1}{2\sqrt{n}},
\end{equation}
while for $n=1$ the same inequality follows directly from
$P_{Y_r}(0)=1/2$. Thus,
\begin{equation}
    M_n(0)
    \leq
    2\sqrt{n}\,P_{Y_r}(0).
\end{equation}
The case $y=n$ follows by symmetry, which proves
\eqref{eq:mode_vs_reference}.

In addition, since
\begin{equation}
    N
    =
    2\left\lceil
        \frac{5\pi}{4}\sqrt{n}
    \right\rceil
    \leq
    \frac{5\pi}{2}\sqrt{n}+2,
\end{equation}
and $2\leq2\sqrt{n}$ for $n\geq1$, we obtain
\begin{equation}
    N+4\sqrt{n}
    \leq
    \left(
        \frac{5\pi}{2}+6
    \right)\sqrt{n}. \label{eq:bound_onN_plus_q}
\end{equation}
Substituting \eqref{eq:mode_vs_reference} into
\eqref{eq:upper_darboux_q} and using
\eqref{eq:bound_onN_plus_q} proves
\eqref{eq:upper_sum_uniform_constant}.
\end{IEEEproof}

Combining the bound in \eqref{eq:upper_bound_on_P_{Y^*}} and \eqref{eq:upper_sum_uniform_constant}, we arrive at
\begin{align}
    \frac{P_{Y^*}(y)}{P_{Y_r}(y)} \le \exp\left(
        \frac{8}{25}-C(n)
    \right) \left(
        \frac{5\pi}{2}+6
    \right)\sqrt{n}. \label{eq:ratio_n_dependent}
\end{align}

Finally, Theorem~\ref{thm:capacity-lb}, together with
$\psi(n+1)\geq\log n$ and
$1+\sqrt{3n+1}\leq3\sqrt n$, gives
\begin{equation}
    C(n)
    \geq
    \frac12
    \log\left(
        \frac{n\pi}{6e}
    \right),
\end{equation}
and hence
\begin{equation}
    \rme^{-C(n)}\sqrt n
    \leq
    \sqrt{\frac{6e}{\pi}}.
    \label{eq:capacity_ratio_bound}
\end{equation}
Substituting 
\eqref{eq:capacity_ratio_bound} into
\eqref{eq:ratio_n_dependent} yields
\begin{equation}
    \frac{P_{Y^*}(y)}{P_{Y_r}(y)}
    \leq
    \rme^{8/25}
    \sqrt{\frac{6e}{\pi}}
    \left(
        \frac{5\pi}{2}
        +6
    \right)
    =
    c_\star^{-1}.
\end{equation}
This concludes the proof of the bound on the ratio.


To show \eqref{eq:upper_gap}, note that by using the \emph{golden formula} \cite{topsoe1967information,csiszar2011information}, we have
\begin{align}
    \kl{P_{Y_r}}{P_{Y^{\star}}} &\le C(n)-I(X_r;Y_r) \le {\rm Gap}(n) \le \frac{17}{\log\!\left(\frac{n\pi}{2e}\right)}, \label{eq:upper_gap_last_Step}
\end{align}
where the last inequality follows from Proposition~\ref{prop:capacity_gap}.

\section{Proof of Theorem~\ref{thm:main-support}} \label{sec:proof_lowerbound_supp}

Before showing the proof of the third bound in Theorem~\ref{thm:main-support}, we state the following helper proposition. 

\begin{prop}\label{prop:invert-chi2}
Let $P_X$ be a discrete distribution  with $K$ mass points in $[0,1]$ and let $P_X  \to P_{Y|X} \to P_Y$.  Let $X_r \sim \text{Beta}(1/2, 1/2)$ and let $P_{X_r}  \to P_{Y|X} \to P_{Y_r}$.  Assume that
\begin{equation} \label{eq:assumption_chi2}
\chi^2(P_Y\|P_{Y_r}) \le \delta
\end{equation}
for some $\delta \in(0,1/4)$, and define $\alpha\triangleq \log\frac{1}{4 \delta}$. Then, 
\begin{equation}
K  \ge  \frac{1}{8}
    \min\left\{
    \sqrt{n\alpha}, n
    \right\}. 
\end{equation}

\end{prop}
\begin{IEEEproof}   
    If $K>(n+2)/4$, there is nothing to prove. Assume instead that $K\le \frac{n+2}{4}$. Then $L=2K$ is admissible in Theorem~\ref{thm:main}, because $2L-2=4K-2\le n$. Applying Proposition~\ref{prop:uniform-explicit} with $L=2K$ gives
\begin{equation}
\chi^2(P_Y\|P_{Y_r}) \ge \frac14\exp\!\left(-\frac{(4K-2)^2}{n-4K+3}\right).
\end{equation}
Combining this with the assumption $\chi^2(P_Y\|P_{Y_r})\le \delta$, we obtain
\begin{equation}
\delta \ge \frac14\exp\!\left(-\frac{(4K-2)^2}{n-4K+3}\right),
\end{equation}
and, therefore, $\frac{(4K-2)^2}{n-4K+3}\ge \alpha_n$. Rearranging yields
\begin{equation}
16K^2+(4\alpha_n-16)K+(4-3\alpha_n-\alpha_n n)\ge 0.
\end{equation}
Solving this quadratic inequality for $K$ gives 
\begin{equation}\label{eq:invert-chi2}
K \ge \frac{4-\alpha+\sqrt{\alpha(4n+\alpha+4)}}{8}.
\end{equation}

We now show the following bound
\begin{equation}
    \frac{
    4-\alpha+\sqrt{\alpha (4n+\alpha+4)}
    }{8}
    \ge
    \frac{1}{8}
    \min\left\{
    \sqrt{n\alpha}, n
    \right\}, \label{eq:use_min_alphan_bound}
\end{equation}
which holds for every $\alpha \ge 0$.  For $\alpha>0$, rationalizing gives
\begin{align}
    \sqrt{\alpha(4n+\alpha+4)}-\alpha
    &=
    \frac{4\alpha(n+1)}
    {
    \sqrt{\alpha(4n+\alpha+4)}+\alpha
    }  \ge
    \frac{\alpha(n+1)}
    {
    \alpha+\sqrt{\alpha(n+1)}
    },
\end{align}
where we used
$\sqrt{\alpha(4n+\alpha+4)}
\le \alpha+2\sqrt{\alpha(n+1)}$. The last expression is lower-bounded,
up to the factor $1/8$, by $\sqrt{n\alpha}$ when $\alpha\le n+1$ and by
$n$ when $\alpha\ge n+1$. This proves \eqref{eq:use_min_alphan_bound}.

Consequently, either $K>(n+2)/4$ or $K$ is lower-bounded as in \eqref{eq:use_min_alphan_bound}. The proof is concluded by noting that \eqref{eq:use_min_alphan_bound} is the smaller of the two bounds. 
\end{IEEEproof}

We are now ready to prove Theorem~\ref{thm:main-support}. 
\begin{IEEEproof}
Starting with the bound in Proposition~\ref{prop:chi2-output}, we have that 
\begin{equation}
  \chi^2(P_{Y^\star} \| P_{Y_r}) \le \zeta(c_\star^{-1}) {\rm Gap}(n) \triangleq \delta_n,
\end{equation}
and, in addition, let $\alpha_n \triangleq \log^{+} \frac{1}{4\delta_n}$.
Therefore, by applying Proposition~\ref{prop:invert-chi2}, 
\begin{align}
    K
&\ge
\frac{1}{8}
\min\left\{
\sqrt{n\alpha_n}, n
\right\}  \\
&=
\frac{1}{8}
\min\left\{
\sqrt{
n \,
\log^{+}\!\left(
\frac{1}{4 \,\zeta(c_\star^{-1}) {\rm Gap}(n) }
\right)}, n
\right\}  \\
&\ge
\frac{1}{8}
\min\left\{
\sqrt{
n \,
\log^{+}\!\left(
\frac{\log\!\left(\frac{n\pi}{2e}\right)}{3268} 
\right)}, n
\right\}, \text{ for $n \ge 444$} 
\end{align}
where in the last step we have used the bound in   \eqref{eq:upper_chi2}  and $ 4 \cdot 17 \cdot \,\zeta(c_\star^{-1})\le  3268$.  This concludes the proof. 
\end{IEEEproof}

\end{appendices}

\bibliographystyle{IEEEtran}
\bibliography{biblio}

\end{document}